\documentclass[11pt]{article}
\usepackage{fullpage}

\usepackage{times}
\usepackage{comment,amsfonts,amssymb,amsmath,amsthm,graphicx,algorithm,algorithmic}
\newcommand{\commentout}[1]{}

\usepackage[utf8]{inputenc}
\usepackage{titlesec}
\usepackage{amsmath}
\usepackage{enumitem}
\usepackage{amsfonts}
\usepackage{amsthm}
\usepackage{mathtools}
\usepackage{nameref}
\usepackage{graphicx}
\graphicspath{ {Images/} }
\usepackage{blindtext}

\ifx\pdftexversion\undefined
\usepackage[colorlinks,linkcolor=black,filecolor=black,citecolor=black,urlco
lor=black,pdfstartview=FitH]{hyperref}
\else
\usepackage[colorlinks,linkcolor=blue,filecolor=blue,citecolor=blue,urlcolor
=blue,pdfstartview=FitH,linktocpage=true]{hyperref}
\fi

\newcommand{\alert}[1]{\textbf{\color{red}
		[[[#1]]]}\marginpar{\textbf{\color{red}**}}\typeout{ALERT:
		\the\inputlineno: #1}}

\def\MathF{\hbox{\rm I\kern-2pt F}}
\def\MathP{\hbox{\rm I\kern-2pt P}}
\def\MathR{\hbox{\rm I\kern-2pt R}}
\def\MathZ{\hbox{\sf Z\kern-4pt Z}}
\def\MathN{\hbox{\rm I\kern-2pt I\kern-3.1pt N}}
\def\MathC{\hbox{\rm \kern0.7pt\raise0.8pt\hbox{\footnotesize I}
		\kern-4.2pt C}}
\def\MathQ{\hbox{\rm I\kern-6pt Q}}

\newcommand{\N}{\mathbb{N}}

\newcommand{\R}{\mathbb{R}}

\newcommand{\E}{{\mathbb{E}}}

\newcommand{\mommit}[1]{}
\newcommand{\namedref}[2]{\hyperref[#2]{#1~\ref*{#2}}}
\newcommand{\sectionref}[1]{\namedref{Section}{#1}}

\newtheorem{theorem}{Theorem}
\newtheorem{lemma}{Lemma}
\newtheorem{corollary}[lemma]{Corollary}
\newtheorem{remark}{Remark}

\newtheorem{fact}[theorem]{Fact}

\newtheorem{definition}{Definition}

\usepackage{pdfsync}
\usepackage{authblk}

\def\tO{\tilde{O}}

\usepackage{authblk}

\begin{document}
\title{A Unified Framework for Hopsets and Spanners\footnote{A preliminary version of this paper was published in ESA'22 \cite{NS22}.}}
\author{Ofer Neiman\footnote{Ben-Gurion University of the Negev, Israel. E-mail: \texttt{neimano@cs.bgu.ac.il}}\qquad  Idan Shabat\footnote{Ben-Gurion University of the Negev, Israel. Email:\texttt{idansha6@gmail.com}}}
\date{}
\maketitle

\begin{abstract}
Given an undirected graph $G=(V,E)$, an {\em $(\alpha,\beta)$-spanner} $H=(V,E')$ is a sub-graph that approximately preserves distances; for every $u,v\in V$, $d_H(u,v)\le \alpha\cdot d_G(u,v)+\beta$. An {\em $(\alpha,\beta)$-hopset} is a graph $H=(V,E'')$, so that adding its edges to $G$ guarantees every pair has an $\alpha$-approximate shortest path that has at most $\beta$ edges (hops), that is, $d_G(u,v)\le d_{G\cup H}^{(\beta)}(u,v)\le \alpha\cdot d_G(u,v)$. Given the usefulness of spanners and hopsets for fundamental algorithmic tasks, several different algorithms and techniques were developed for their construction, for various regimes of the stretch parameter $\alpha$.

In this work we develop a single algorithm that can attain all state-of-the-art spanners and hopsets for general graphs, by choosing the appropriate input parameters. In fact, in some cases it also improves upon the previous best results. We also show a lower bound on our algorithm.

In \cite{BP20}, given a parameter $k$, a $(O(k^{\epsilon}),O(k^{1-\epsilon}))$-hopset of size $\tO(n^{1+1/k})$ was shown for any $n$-vertex graph and parameter $0<\epsilon<1$, and they asked whether this result is best possible. We resolve this open problem, showing that any $(\alpha,\beta)$-hopset of size $O(n^{1+1/k})$ must have $\alpha\cdot \beta\ge\Omega(k)$.

\end{abstract}

\newpage
\tableofcontents
\newpage

\section{Introduction}

Hopsets and spanners are fundamental graph theoretic structures, that have gained much attention recently \cite{KS97,C00,EP04,E04,TZ06,P10,BKMP10,B09,MPVX15,HKN16,FL16,EN16,ABP18,EN19,HP17,EGN19,BP20}. They play a pivotal role in central algorithmic applications such as approximating shortest paths, distributed computing tasks, geometric algorithms, and many more.

Given a graph $G=(V,E)$, possibly with non-negative weights on the edges $w:E\to\R$, an {\em $(\alpha,\beta)$- hopset} is a graph $H=(V,E')$ such that every pair in $V$ has an $\alpha$-approximate shortest path in $G\cup H$ with at most $\beta$ hops. That is, for all $u,v\in V$,
\[
d_G(u,v)\le d_{G\cup H}^{(\beta)}(u,v)\le \alpha\cdot d_G(u,v)~,
\]
where $d_G(u,v)$ is the distance between $u,v$ in $G$, and $d_{G\cup H}^{(\beta)}(u,v)$ stands for the length of the shortest path in $G\cup H$ between $u,v$ that has at most $\beta$ edges. The weight of an edge $(x,y)\in E'$ of $H$ is defined as $d_G(x,y)$.

An {\em $(\alpha,\beta)$-spanner} of $G$ is a sub-graph $H=(V,E')$ such that for all $u,v\in V$,
\[
d_G(u,v)\le d_H(u,v)\le \alpha\cdot d_G(u,v)+\beta~.
\]
In the case $\beta = 0$ the spanner is called {\em multiplicative}, and when $\alpha=1+\epsilon$ for some small $\epsilon>0$ it is called {\em nearly additive}. If $H$ is not a sub-graph of $G$, but rather a graph on the same vertices set $V$, then $H$ is called an {\em $(\alpha,\beta)$-emulator} of $G$.

There is seemingly a close connection between spanners, emulators and hopsets; many techniques that were developed for the construction of $(\alpha,\beta)$-spanners can produce $(\alpha,\beta)$-hopsets and emulators, and vice versa, often without any change in the algorithm (see, e.g., \cite{HP17,EN19}).

\paragraph{Multiplicative Stretch.} These spanners were introduced by \cite{PU89}, and are well-understood. For any integer parameter $k\ge 1$, any weighted graph with $n$ vertices has a $(2k-1,0)$-spanner with $O(n^{1+1/k})$ edges \cite{ADDJS93}, which is asymptotically optimal assuming Erdos' girth conjecture.  The celebrated distance oracles of \cite{TZ01} can also be viewed as a $(2k-1,0)$-spanner, and as a $(2k-1,2)$-hopset (see \cite{BP20}).

\paragraph{Nearly Additive Stretch.} For $0<\epsilon<1$, near additive $(1+\epsilon,\beta)$-spanners for unweighted graphs were introduced by \cite{EP04}, who, for any parameter $k$, devised a spanner of size $O(\beta\cdot n^{1+1/k})$ with $\beta = O\left(\frac{\log k}{\epsilon}\right)^{\log k}$. A similar result that works simultaneously for every $\epsilon$ was achieved by \cite{TZ06}, using a different parametrization of their \cite{TZ01} algorithm. The factor of $\beta$ in the size was recently improved by \cite{P08,BP20,EGN19}.
By a result of \cite{ABP18}, any such spanner of size $O(n^{1+1/k})$ must have $\beta=\Omega\left(\frac{1}{\epsilon\cdot\log k}\right)^{\log k}$. So whenever the stretch parameter $\epsilon$ is sufficiently small with respect to $1/\log k$, the result of \cite{EP04} cannot be substantially improved. There are also numerous results on purely-additive spanners, with $\alpha=1$, which will not be discussed here.

Hopsets were introduced by \cite{C00}, where $(1+\epsilon,\beta)$-hopsets of size $O(n^{1+1/k}\cdot\log n)$ and $\beta = O\left(\frac{\log n}{\epsilon}\right)^{\log k}$ were designed. This was recently improved by \cite{EN16,HP17,EN19}, adapting  the techniques of \cite{EP04,TZ06} for $(1+\epsilon,\beta)$-spanners, yielding $\beta=O\left(\frac{\log k}{\epsilon}\right)^{\log k}$ and size $O(n^{1+1/k})$.

\paragraph{Hybrid Stretch.} The lower bound of \cite{ABP18} is meaningful only when the multiplicative stretch is very close to 1. This motivates the natural question: what $\beta$ (hopbound or additive stretch) can be obtained with some large multiplicative stretch? This question was studied by \cite{EGN19,BP20}, who showed $(3+\epsilon,\beta)$-spanners and hopsets of size $O(k\cdot n^{1+1/k}\cdot\log\Lambda)$ with improved $\beta = k^{\log 3 + O(1/\epsilon)}$, where $\Lambda$ is the aspect ratio of the graph\footnote{The aspect ratio is the ratio between the largest distance to the smallest distance in the graph.} (in fact, \cite{EGN19} did not have the $\log\Lambda$ factor in the size, albeit their $\beta$ had a somewhat worse exponent).

More generally, for any $0<\epsilon<1$, \cite{BP20} devised a $(k^\epsilon,O_\epsilon(k))$-spanner of size $O_\epsilon(k\cdot n^{1+1/k})$, and a $(O(k^\epsilon),O_\epsilon(k^{1-\epsilon}))$-hopset of size $O(k^\epsilon\cdot n^{1+1/k}\cdot\log\Lambda)$.
The latter algorithm of \cite{BP20} is rather complicated, and contains a three-stage construction involving a truncated application of the \cite{TZ06} algorithm, a superclustering phase (based on the constructions of \cite{EP04}), and a multiplicative spanner built on some cluster graph.
The tightness of this $(O(k^\epsilon),O_\epsilon(k^{1-\epsilon}))$-hopset was asked as open question in \cite{BP20}.

\subsection{Our Results}

In this paper we develop a generalization of the TZ-algorithms \cite{TZ01,TZ06}, that achieves (and sometimes even improves on) all the above results on hopsets, spanners and emulators. This unifies all previous algorithms in a single  framework, and greatly simplifies the constructions for hybrid stretch spanners, emulators and hopsets. We also remove the $\log\Lambda$ factor from the size.

In addition, we affirmatively resolve the open problem of \cite{BP20} mentioned above, by proving that an $(\alpha,\beta)$-hopset of size $O(n^{1+1/k})$ must have $\alpha\cdot\beta\ge \Omega(k)$. This lower bound asymptotically matches the upper bound of $(k^\epsilon,O_\epsilon(k^{1-\epsilon}))$-hopset by \cite{BP20} for every $0<\epsilon<1$.
We remark that for $(\alpha,\beta)$-spanners, there is a better lower bound of $\alpha+\beta\ge\Omega(k)$ (since this spanner is in particular a $(\alpha+\beta,0)$-spanner). However, this lower bound cannot hold for hopsets, as indicated by the existence of $(O(k^\epsilon),O(k^{1-\epsilon}))$-hopsets for $\epsilon=1/2$, say.

In addition, we show that whenever our algorithm produces a hopset of size $O(n^{1+1/k})$ with stretch $\alpha$, it must have a hopbound of $\beta=\Omega(\frac{1}{\alpha^2}k^{1+1/2\log\alpha})$. As our algorithm generalizes all previous constructions, we conclude that resolving the question whether there exists an $(O(1), O(k))$-hopset of size $O(n^{1+1/k})$ - will have to rely on novel techniques.

Our results arising from the unified framework for hopsets and emulators are summarized in Table~\ref{table:hopsets}, and for spanners in Table~\ref{table:spanners}.

\begin{table}[ht]
\begin{tabular}{|l|l|l|}
\hline
Stretch &
Hopbound / Additive stretch &
Size  \\ \hline
\begin{tabular}[c]{@{}l@{}}
$1+\epsilon$ \end{tabular} &
$O(\frac{\log k}{\epsilon})^{\log k}$ &
$O(n^{1+\frac{1}{k}})$ \\ \hline
\begin{tabular}[c]{@{}l@{}}
$3+\epsilon$ \end{tabular} &
$k^{\log (3 + 16/\epsilon)}$ &
$O(n\log k+n^{1+\frac{1}{k}})$\\ \hline
\begin{tabular}[c]{@{}l@{}}
$O(c)$ \end{tabular} &
$k^{1 + 2/\ln c}$ &
$O(n\log k+n^{1+\frac{1}{k}})$\\ \hline
\begin{tabular}[c]{@{}l@{}}
$O(k^\epsilon)$\end{tabular} &
\begin{tabular}[c]{@{}l@{}}
$O_\epsilon(k^{1-\epsilon})~~$ / $~~O_\epsilon(k)$ \end{tabular} &
$O_\epsilon(k^\epsilon\cdot n^{1+\frac{1}{k}})$ \\ \hline
$2k-1$ &
$2~~$ / $~~0$&
$O(k\cdot n^{1+\frac{1}{k}})$  \\ \hline

\end{tabular}
\caption{Our results on $(\alpha,\beta)$-hopsets and $(\alpha,\beta)$-emulators for $n$-vertex weighted graphs, given $0<\epsilon<1$ and integer $k\ge 1$, for various stretch regimes.} \label{table:hopsets}
\end{table}

\begin{table}[ht]
\begin{tabular}{|l|l|l|}
\hline
Stretch &
Additive stretch &
Size  \\ \hline
\begin{tabular}[c]{@{}l@{}}
$3+\epsilon$ \end{tabular} &
$k^{\log (3 + 8/\epsilon)}$ &
$O(k^{\log (3 + 8/\epsilon)}\cdot n^{1+\frac{1}{k}})$\\ \hline
\begin{tabular}[c]{@{}l@{}}
$3+\epsilon$ \end{tabular} &
$k^{\log_{4/3} (5 + 16/\epsilon)}$ &
$O(n\log k+n^{1+\frac{1}{k}})$\\ \hline
\begin{tabular}[c]{@{}l@{}}
$O(c)$ \end{tabular} &
$k^{1 + 2/\ln c}$ &
$O(k^{1 + 2/\ln c}\cdot n^{1+\frac{1}{k}})$\\ \hline
\begin{tabular}[c]{@{}l@{}}
$O(k^\epsilon)$\end{tabular} &
\begin{tabular}[c]{@{}l@{}}
$O_\epsilon(k)$\end{tabular} &
$O_\epsilon(k^\epsilon\cdot n^{1+\frac{1}{k}})$ \\ \hline

\end{tabular}
\caption{Our results on $(\alpha,\beta)$-spanners for $n$-vertex weighted graphs, given $0<\epsilon<1$ and integer $k\ge 1$, for various stretch regimes.} \label{table:spanners}
\end{table}

\subsection{Our Techniques}

The lower bound on the triple tradeoff between stretch, hopbound and size of $(\alpha,\beta)$-hopsets, showing that $\alpha\cdot\beta=\Omega(k)$ whenever the size is $O(n^{1+1/k})$, uses the existence of $n^{1/g}$-regular graphs with girth $g$.  The basic idea is simple: locally (within distance less than $g/2$) the graph looks like a tree, so when considering short enough paths, of length less than $g/\alpha$, there are no alternative paths with stretch at most $\alpha$. This means that any hopset edge $(u,v)$ can only be useful to pairs whose shortest path is ``nearby" to $u,v$. Making this intuition precise, and defining what exactly is ``nearby", requires some careful counting arguments.

Before discussing our general algorithm for hopsets, emulators and spanners, let us review the previous TZ algorithms:
Let $G=(V,E)$ be a (possibly weighted) graph with $n$ vertices, and fix an integer parameter $k$. The algorithms of \cite{TZ01,TZ06} randomly sample a sequence of sets $V=A_0\supseteq A_1\supseteq ...\supseteq  A_F$, for some $F$, where each $A_{i+1}$, $0\le i<F$, is sampled by including each vertex from $A_i$ independently with some predefined probability. Then they define for each $v\in V$ its $i$-th {\em pivot} $p_i(v)$ as the closest vertex in $A_i$ to $v$, and the $i$-th  {\em bunch} as $B_i(v)=\{u\in A_i~:~ d(u,v)<d(v,p_{i+1}(v))\}$. The hopset or spanner consists of all edges (or shortest paths) between each $v$ and some of its bunches.

\paragraph{Linear-TZ for multiplicative spanners, emulators and hopsets.}
In \cite{TZ01}, the sampling probabilities of each $A_{i+1}$ from $A_i$ were uniform $n^{-1/k}$. Thus, the expected size of each set $A_i$ in the hierarchy was $n^{1-\frac{i}{k}}$. This means that the exponent of $n$ in $|A_i|$ is {\em linearly} decreasing in $i$, hence we call this construction a {\em linear-TZ}. In this version, each vertex $v\in V$ connects to vertices in $B_i(v)$ for all $0\le i\le F$. The analysis gives a spanner or emulator (or a hopset with $\beta=2$) with multiplicative stretch $2k-1$.

\paragraph{Exponential-TZ for near-additive spanners, emulators and hopsets.}
In \cite{TZ06}, the sampling probabilities of each $A_{i+1}$ from $A_i$ were roughly $n^{-2^i/k}$. This means that the expected size of the set $A_i$ is $n^{1-\frac{2^i-1}{k}}$, so we naturally call this an {\em exponential-TZ}.
As the probabilities are much lower here, the bunches are larger, so vertices in $A_i\setminus A_{i+1}$ can only connect to their $i$-th bunch (in order to keep the size under control).
This version can provide spanners/emulators/hopsets with a multiplicative stretch that is arbitrarily close to $1$ (namely, $1+\epsilon$ for any small $\epsilon>0$).

\paragraph{Our algorithm.}
In this work, we devise the following generalization of both of these algorithms. Our algorithm expects as a parameter a function $f:\N\to\N$ that determines, for each level $i$, the highest bunch-level that vertices in $A_i\setminus A_{i+1}$ will connect to (in the linear-TZ we have $f(i)=F$, while in the exponential-TZ, $f(i)=i$ for all $i$). This function $f$ implies what should the sampling probability be for each level $i$, in order to keep the total size of the spanner/emulator/hopset roughly $O(n^{1+1/k})$. We denote these probabilities by $n^{-\lambda_i/k}$, for parameters $\lambda_0,\lambda_1,...,\lambda_{F-1}$. The number of sets $F$ is in turn determined by these $\lambda_i$ (roughly speaking, it is when we expect $A_F$ to be empty). 

As this is a generalization of the algorithms of \cite{TZ01,TZ06}, clearly it achieves their results. One of our main technical contributions is showing that an {\em interleaving} of the linear-TZ and exponential-TZ probabilities, yields a spanner/emulator/hopset with hybrid stretch. This means that we divide the integers in $[F]=\{1,2,...,F\}$ to $F/c$ intervals, so that the $\lambda_i$'s are the same within each interval, and decays exponentially between intervals. The parameter $c$ controls the multiplicative stretch.

Our analysis combines ideas from previous works \cite{TZ06,EGN19,BP20}, with some novel insights that simplify some of the previously used arguments. In particular, our analysis of the hopset is not scaled-based, as it was in \cite{BP20}, which enables us to remove the $\log\Lambda$ factors from the size. In addition, our $(3+\epsilon,\beta)$-hopsets combine the best attributes of the hopsets of \cite{EGN19} and \cite{BP20}: they have no dependence on $\log\Lambda$ and works for all $\epsilon$ simultaneously like \cite{EGN19}, and have the superior $\beta$ like \cite{BP20}.

\subsection{Organization}

In \sectionref{General Lower Bound for Hopsets} we show our lower bound for hopsets. In \sectionref{Generalized TZ Construction} we describe our general algorithm for hopsets and spanners, and provide an analysis of its stretch and hopbound in \sectionref{Stretch and Hopbound Analysis Method}. In \sectionref{TZ Spanners} we show that our algorithm can also yield the state-of-the-art spanners and emulators with hybrid stretch. Finally, in \sectionref{A Lower Bound for the TZ Hopset} we present a lower bound on our algorithm.

\section{Lower Bound for Hopsets} \label{General Lower Bound for Hopsets}
In this section we build a graph $G$, such that every hopset for $G$ with stretch $\alpha$ and size $O(n^{1+\frac{1}{k}})$ must have a hopbound of at least $\approx\frac{k}{\alpha}$. $G$ has high girth (the size of the smallest simple cycle), so that any short enough path in $G$ is a unique shortest path between its ends, and every other path between these ends is much longer. Then, if a hopset has a small enough hopbound, every such path needs to ``use" a hopset edge in its hopset-path. Since there are many such paths, and each hopset edge can't be used by many paths, the hopset has to contain many edges.

For the construction, we use the following result, 
by Lubotzky, Phillips and Sarnak \cite{LPS88}:

\begin{theorem}[\cite{LPS88}] \label{thm:regGraph}
Given an integer $\gamma\geq1$, there are infinitely many integers $n\in\mathbb{N}$ such that there exists a $(p+1)$-regular graph $G=(V,E)$ with $|V|=n$ and girth $\geq\frac{4}{3}\gamma(1-o(1))$, where $p=D\cdot n^\frac{1}{\gamma}$, for some universal constant $D$.
\end{theorem}

Now, fix $\alpha,\gamma\geq1$ and a large enough $n$ as above, and let $G=(V,E)$ be the matching $(p+1)$-regular graph from the theorem. We know that the girth of $G$ is $\geq\frac{4}{3}\gamma(1-o(1))$, so we can assume that this girth is at least $\gamma$. We look at paths in $G$ of distance $\delta\coloneqq\lfloor\frac{\gamma-1}{\alpha+1}\rfloor$. For a path $P$, denote by $|P|$ its length.

\begin{lemma}
For every path $P$ between some $u,v\in V$, such that $|P|=\delta$, $P$ is the unique shortest path between $u,v$.
Moreover, for any other path $P'$ between $u$ and $v$, $|P'|>\alpha|P|=\alpha\delta$
\end{lemma}

\begin{proof}
It is enough to prove that for any $P'\neq P$ between $u$ and $v$, $|P'|>\alpha\delta$. If it is not the case, then there must be a simple cycle in $P\cup P'$, and its length is bounded by $|P|+|P'|\leq\delta+\alpha\delta=(\alpha+1)\delta\leq\gamma-1$, in contradiction to the girth of $G$ being $\geq\gamma$.
\end{proof}

Denote the set of all $\delta$-paths in $G$ by $Q_\delta=\{P_{u,v}\;|\;d(u,v)=\delta\}$ (In general, we denote by $P_{u,v}$ the shortest path between $u$ and $v$).

\begin{lemma} \label{lemma:ManyPaths}
$|Q_\delta|\geq\frac{1}{2}np^\delta$.
\end{lemma}

\begin{proof}
Given a vertex $u\in V$, denote its BFS tree, up to the $\delta$'th level, by $T$. Since $G$'s girth is $>(\alpha+1)\delta$, there are no edges between the vertices of $T$, apart from the edges of $T$ itself. That means that each vertex of $T$ has at least $p$ children at the next level, so we have at least $p^\delta$ leaves in $T$. Each of these leaves is a vertex of distance $\delta$ from $u$, and is connected to $u$ with a $\delta$-path. When summing this quantity over all the vertices $u\in V$, we count every path twice, so we get at least $\frac{1}{2}\sum_{u\in V}p^\delta=\frac{1}{2}np^\delta$ paths of length $\delta$.
\end{proof}

We are now ready to prove the main theorem:

\begin{theorem}
For every positive integer $k$, a real number $\alpha>0$, a constant $C>0$ and for infinitely many integers $n$, there exists a graph $G$ with $n$ vertices such that every hopset $H$ for $G$ with size $\leq Cn^{1+\frac{1}{k}}$ and stretch $\leq\alpha$, $H$ has a hopbound $\beta\geq\lfloor\frac{k-2}{\alpha+1}\rfloor$.
\end{theorem}

\begin{proof}
For $\alpha,n$ and a fixed $\gamma\geq1$ that will be chosen later, let $G=(V,E)$ be the $(p+1)$-regular graph from Theorem~\ref{thm:regGraph} ($|V|=n$, girth $\geq\gamma$ and $p=D\cdot n^\frac{1}{\gamma}$). Define $\delta,Q_\delta$ the same way as above.

Let $H$ be an $(\alpha,\beta)$-hopset for $G$ with size $\leq Cn^{1+\frac{1}{k}}$, where $\beta<\delta$. For $e=(x,y)\in H$, we denote the weight of $e$, which is defined to be the distance $d(x,y)$, by $w(e)$ ($d(x,y)$ is the distance in the graph $G$. We omit the subscript from $d_G(u,v)$ for brevity). To formalize our next arguments, we think of a bipartite graph $(A,B,\hat{E})$, where $A=Q_\delta$, $B=\{e\in H\;|\;w(e)\leq \alpha\delta\}$ and $\hat{E}=\{(P,(x,y))\in A\times B\;|\;P\cap P_{x,y}\neq\emptyset\}$. We prove the following two properties of this graph ($deg_{\hat{E}}$ denotes the degree of a vertex in this graph):
\begin{enumerate}
    \item $\forall_{P\in A}\;deg_{\hat{E}}(P)\geq1$,
    \item $\forall_{e\in B}\;deg_{\hat{E}}(e)\leq \alpha\delta^2n^\frac{\delta-1}{\gamma}$.
\end{enumerate}

For (1), we need to show that for every $u,v\in V$ such that $d(u,v)=\delta$, there is a $(x,y)\in H$ such that $P_{u,v}\cap P_{x,y}\neq\emptyset$ and $w(x,y)\leq \alpha\delta$. Let $P\subseteq G\cup H$ be the shortest path from $u$ to $v$, that has at most $\beta$ edges, and let $\hat{P}$ be the ``translation" of $P$ to a path in $G$. That is, $\hat{P}$ is the same path as $P$, with every $H$-edge replaced by the shortest path in $G$ that connects its ends - we call such a shortest path a \textit{detour} of $\hat{P}$. By the hopset property: $|\hat{P}|=w(P)=d_{G\cup H}^{(\beta)}(u,v)\leq\alpha\cdot d(u,v)=\alpha\delta$.

Notice that since $\beta<\delta$, and $P$ has at most $\beta$ edges, there must be some edge $(a,b)\in P_{u,v}$ such that $(a,b)\notin P$. Seeking contradiction, assume that for every $H$-edge $(x,y)$ that $P$ uses, $P_{u,v}\cap P_{x,y}=\emptyset$. Then, since $(a,b)\in P_{u,v}$, we get that $(a,b)\notin P_{x,y}$ for every $(x,y)\in H$ that $P$ uses, i.e. $(a,b)$ is not on any detour of $\hat{P}$. So we now know that $(a,b)$ cannot be in $\hat{P}$; not on its detours and not on its $G$-edges (which are the edges of $P$). See figure \ref{fig:detours} for visualization.

But then, there are two edges-disjoint paths that connects $a,b$ in $P_{u,v}\cup\hat{P}$: One is the edge $(a,b)$ itself, and the other is
\[a\xrightarrow{P_{a,u}}u\xrightarrow{\hat{P}}v\xrightarrow{P_{v,b}}b\]
(the paths $P_{a,u},P_{v,b}$ are subpaths of $P_{u,v}$ and do not contain $(a,b)$). The union of these two paths contains a simple cycle of size at most $|P_{u,v}|+|\hat{P}|\leq\delta+\alpha\delta=(\alpha+1)\delta\leq\gamma-1$, in contradiction to the girth of $G$.

Therefore, there must be an edge $(x,y)\in H$ that $P$ uses such that $P_{u,v}\cap P_{x,y}\neq\emptyset$. Since $w(P)\leq\alpha\delta$ and $(x,y)\in P$, we also have that $w(x,y)\leq\alpha\delta$.

\begin{center}
\begin{figure}[ht]
    \centering
    \includegraphics[width=12cm, height=3cm]{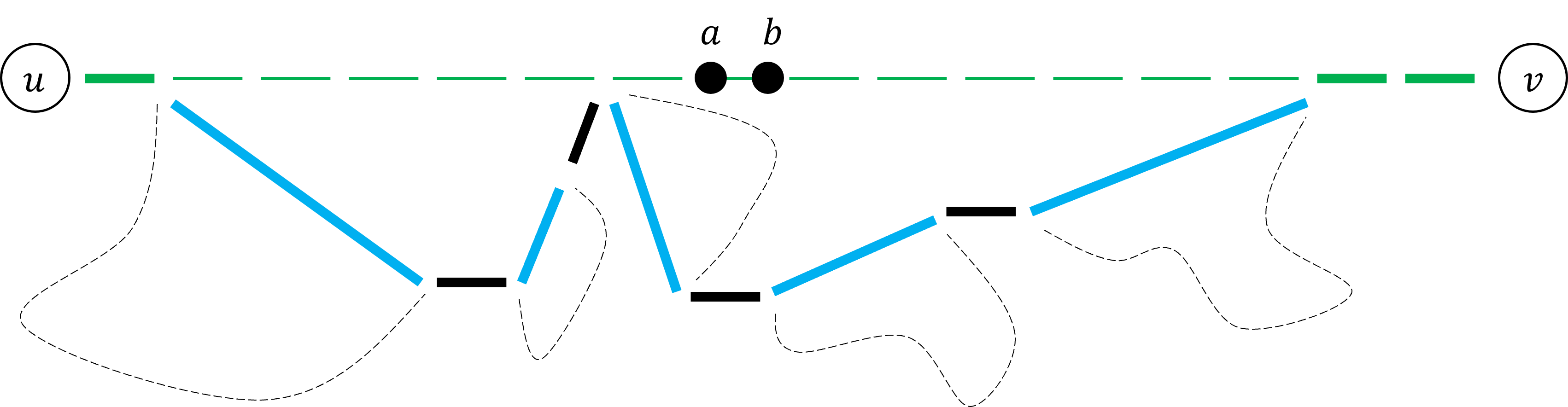}
    \caption{$P_{u,v}$ is marked by green edges. Every edge of the hopset $H$ is colored in blue, and the shortest path of $G$ connecting its ends is marked by a dashed line. The edges of the path $P\subseteq G\cup H$ are marked by a thick line, and we assume by negation that its detours are disjoint from $P_{u,v}$. Then, since also $P$ doesn't use the edge $(a,b)$ of $P_{u,v}$, there must be a simple cycle with size $\leq|P_{u,v}|+|\hat{P}|$.}
    \label{fig:detours}
\end{figure}
\end{center}

For (2), given $(x,y)\in H$ such that $w(x,y)\leq \alpha\delta$, we need to bound the number of pairs $u,v\in V$ such that $d(u,v)=\delta$ and $P_{u,v}\cap P_{x,y}\neq\emptyset$. Let $(a,b)\in P_{x,y}$. Every path of length $\delta$ that passes through $(a,b)$ is a concatenation of a path of length $i$ that ends in $a$, the edge $(a,b)$ and a path of length $\delta-1-i$ from $b$, for some $i\in[0,\delta-1]$. Fixing $i$, we consider the BFS trees $T_a, T_b$ of $a,b$ respectively, up to the $i$'th and $(\delta-1-i)$'th level respectively. Since the degree of any vertex in $G$ is $p+1$, $T_a$ contains at most $p^i$ leaves, and $T_b$ contains at most $p^{\delta-1-i}$ leaves. Therefore, the number of concatenations of paths as above is bounded by:

\[\sum_{i=0}^{\delta-1}p^i\cdot p^{\delta-1-i}=\sum_{i=0}^{\delta-1}p^{\delta-1}=\delta p^{\delta-1}~.\]

Since $P_{x,y}$ contains at most $\alpha\delta$ edges, we get that the number of paths $P_{u,v}$ such that $P_{u,v}\cap P_{x,y}\neq\emptyset$ and $|P_{u,v}|=\delta$, is bounded by $\alpha\delta\cdot\delta p^{\delta-1}=\alpha\delta^2p^{\delta-1}$.

Finally, using the two properties of the bipartite graph, we count its edges:
\[\frac{1}{2}np^\delta\stackrel{\text{Lemma \ref{lemma:ManyPaths}}}{\leq}|Q_\delta|=|A|\stackrel{\text{(1)}}{\leq}\sum_{P\in A}deg_{\hat{E}}(P)=|\hat{E}|=\sum_{e\in B}deg_{\hat{E}}(e)\stackrel{\text{(2)}}{\leq}|B|\alpha\delta^2p^{\delta-1}\leq|H|\alpha\delta^2p^{\delta-1}~.\]

Rearranging this inequality, we get
\[|H|\geq\frac{1}{2\alpha\delta^2}np=\frac{D}{2\alpha\delta^2}n\cdot n^\frac{1}{\gamma}=\frac{D}{2\alpha\delta^2}n^{1+\frac{1}{\gamma}}~.\]

Recall that $|H|\leq Cn^{1+\frac{1}{k}}$, so when choosing large enough $n$, it must be that $k\leq\gamma$.

Summarizing our proof so far, we showed that for fixed $\gamma\geq1$, $\alpha>0$ and a constant $C$, there is a graph $G$ such that every $(\alpha,\beta)$-hopset $H$ for $G$, with size $\leq Cn^{1+\frac{1}{k}}$, either satisfies $\beta\geq\delta$, or satisfies $k\leq\gamma$.

Choose $\gamma=k-1$. Now the corresponding graph $G$ has the property that every $(\alpha,\beta)$-hopset $H$ for $G$, with size $\leq Cn^{1+\frac{1}{k}}$, must have $\beta\geq\delta$. By $\delta$'s definition:
\[\beta\geq\delta=\lfloor\frac{\gamma-1}{\alpha+1}\rfloor=\lfloor\frac{k-2}{\alpha+1}\rfloor~.\]

\end{proof}

\section{A Unified Construction of Hopsets} \label{Generalized TZ Construction}
Let $G=(V,E)$ be a weighted undirected graph.
Our main construction depends on the choice of several parameters, whose roles will be discussed below.
\begin{enumerate}
    \item A positive integer $k\geq1$.
    \item A positive integer $F\geq1$.
    \item A non-decreasing function $f:\mathbb{N}\rightarrow\mathbb{N}$ such that $i\leq f(i)\leq F$ for every $0\le i\le F$.
    \item A sequence of non-negative numbers $\{\lambda_j\}_{j=0}^{F-1}$.
\end{enumerate}

The resulting construction is a hopset $H(k,f,\{\lambda_j\},F)$.
In the definition below, notice that this construction is actually not deterministic, but is created using some random choices.

The construction starts by creating a sequence of sets $V=A_0\supseteq A_1\supseteq ... \supseteq A_F$. The set $A_{j+1}$ is defined by selecting each vertex from $A_j$ independently with probability $\frac{1}{2}n^{-\frac{\lambda_j}{k}}$ (recall that $n=|V|$). The set $A_F$ is defined to be the empty set $A_F=\emptyset$. 

Given this sequence, we define some useful notations for every vertex $u\in V$ and $0\leq j\leq F-1$.
\begin{enumerate}
    \item Denote by $i(u)$ the {\em level} of $u$, which is the only $i$ such that $u\in A_i\setminus A_{i+1}$. 
    \item The $j$'th {\em pivot} of $u$, $p_j(u)$, is the vertex of $A_j$ which is the closest to $u$.
    \item The $j$'th {\em bunch} of $u$ is the set $B_j(u)=\{v\in A_j\;|\;d(u,v)<d(u,p_{j+1}(u))\}$. \newline
    (Whenever $A_{j+1}=\emptyset$, then $p_{j+1}(u)$ is undefined, and we say that $d(u,p_{j+1}(u))=\infty$, so $B_{j}(u)=A_j$).
\end{enumerate}

\begin{definition}
\[H(k,f,\{\lambda_j\},F)=\bigcup_{u\in V}\bigcup_{j=0}^{F-1}\{(u,p_j(u))\}\cup\bigcup_{u\in V}\bigcup_{j=i(u)}^{f(i(u))}\{(u,v)\;|\;v\in B_j(u)\}~.\]
\end{definition}

In words, $H(k,f,\{\lambda_j\},F)$ consists of all the edges of the form $(u,p_j(u))$, for every $u$ and $j$, and of all the edges $(u,v)$ such that $v\in B_j(u)$ for some $j\in[i(u),f(i(u))]$.

For clarity, we briefly summarize the role of each parameter in $H(k,f,\{\lambda_j\},F)$:
\begin{enumerate}
    \item The parameter $k$ controls the size of the hopset, which (under certain conditions) is expected to be roughly $O(n^{1+\frac{1}{k}})$.
    \item The function $f$ determines, for each $u\in V$, the highest index $j=f(i(u))$, such that edges of the form $(u,v)$, where $v\in B_j(u)$, are added to the hopset. Notice that $B_j(u)=\emptyset$ for $j<i(u)$, since in this case $p_{j+1}(u)=u$. Therefore, the relevant range of indexes for $u$ is $[i(u),f(i(u))]$, and this is the reason $f$ has to satisfy $f(i)\geq i$ for every $i$.
    \item The sequence $\{\lambda_j\}$ controls the sampling probabilities for the sets $\{A_j\}$: The probability of some $u\in A_j$ to be in $A_{j+1}$ is $\frac{1}{2}n^{-\frac{\lambda_j}{k}}$. 
    \item The parameter $F$ is the index of the last set in the sequence $\{A_j\}$, which is the empty set $\emptyset$.
\end{enumerate}

Throughout this paper, we use the following notation:
\[f^{-1}(j)=\min\{i\;|\;j\leq f(i)\}~.\]

The following lemma bounds the expected size of our hopset, given certain conditions on the parameters.
\begin{lemma} \label{lemma:HkfESize}
Suppose that the parameters $k,f,\{\lambda_j\},F$ satisfy
\begin{enumerate}
    \item $\sum_{j<F}\lambda_j\geq k$
    \item $\forall_j\;\lambda_j\leq1+\sum_{l<f^{-1}(j)}\lambda_l$
\end{enumerate}
Then, $\E[|H(k,f,\{\lambda_j\},F)|]=O(Fn+\max_i(f(i)-i+1)\cdot n^{1+\frac{1}{k}})$. 
\end{lemma}

\begin{proof}


Note that by the definition of the sets $A_i$, for $0\leq i\leq F-1$, every vertex of $V$ is contained in $A_i$ independently with probability 
\[\prod_{j<i}\left(\frac{1}{2}n^{-\frac{\lambda_j}{k}}\right)=\frac{1}{2^i}n^{-\frac{1}{k}\sum_{j<i}\lambda_j}~.\]
Thus, $\E[|A_i|]=\frac{1}{2^i}n^{1-\frac{1}{k}\sum_{j<i}\lambda_j}$ for every $0\leq i\leq F-1$.

Now, fix some $u\in V$ and $j$ such that $i(u)\leq j<F-1$. We show that the expected size of $B_j(u)$ is at most $2n^\frac{\lambda_j}{k}-1$. Let $u_1,u_2,u_3,...$ be all of the vertices of $A_j$, ordered by their distance from $u$. If $u_l$ is the first vertex in this list such that $u_l\in A_{j+1}$, then $B_j(u)\subseteq\{u_1,u_2,...,u_{l-1}\}$, since vertices that are closer to $u$ than $u_l$ appear before $u_l$ in this list. Notice that $l$ is stochastically dominated by a geometric random variable, since each of the $u_i$'s is sampled into $A_{j+1}$ independently with probability $\frac{1}{2}n^{-\frac{\lambda_j}{k}}$, and $u_l$ is the first vertex that was sampled. Therefore, $\E[|B_j(u)|]\leq\E[l-1]\leq2n^\frac{\lambda_j}{k}-1$.

For $j=F-1$, note that by definition $B_{F-1}(u)=A_{F-1}$, for every $u\in V$. Recall that $\E[|A_i|]=\frac{1}{2^i}n^{1-\frac{1}{k}\sum_{j<i}\lambda_j}$ for every $0\leq i\leq F-1$, and in particular,
\[\E[|B_{F-1}(u)|]=\E[|A_{F-1}|]=\frac{1}{2^{F-1}}n^{1-\frac{1}{k}\sum_{j<F-1}\lambda_j}\leq\frac{1}{2^{F-1}}n^{1-\frac{k}{k}+\frac{\lambda_{F-1}}{k}}<n^{\frac{\lambda_{F-1}}{k}}~.\]
Here, we used our assumption $\sum_{j<F}\lambda_j\geq k$.

We conclude that for every $0\leq j\leq F-1$ and for every $u\in V$, we have $\E[|B_j(u)|]<2n^{\frac{\lambda_j}{k}}$. Hence,
\begin{eqnarray*}
  \E[|H(k,f,\{\lambda_j\},F)|]&\leq&|\bigcup_{u\in V}\bigcup_{j=0}^{F-1}\{(u,p_j(u))\}|+\E\left[\right|\bigcup_{u\in V}\bigcup_{j=i(u)}^{f(i(u))}\{(u,v)\;|\;v\in B_j(u)\}\left|\right]\\
  &<&Fn+\sum_{u\in V}\sum_{j=i(u)}^{f(i(u))}2n^\frac{\lambda_j}{k}\\
  &=&Fn+\sum_{i=0}^{F-1}|A_i\setminus A_{i+1}|\sum_{j=i}^{f(i)}2n^\frac{\lambda_j}{k}~.
\end{eqnarray*}

Recall that $\E[|A_i|]=\frac{1}{2^i}n^{1-\frac{1}{k}\sum_{l<i}\lambda_l}$, so 
\[\E[|H(k,f,\{\lambda_j\},F)|]<Fn+\sum_{i=0}^{F-1}\frac{1}{2^i}n^{1-\frac{1}{k}\sum_{l<i}\lambda_l}\sum_{j=i}^{f(i)}2n^\frac{\lambda_j}{k}=Fn+\sum_{i=0}^{F-1}\sum_{j=i}^{f(i)}\frac{1}{2^{i-1}}n^{1+\frac{1}{k}(\lambda_j-\sum_{l<i}\lambda_l)}~.\]
Recall the assumption $\lambda_j\leq1+\sum_{l<f^{-1}(j)}\lambda_l$, for every $j$. For every $i\geq f^{-1}(j)$, we have \[\lambda_j\leq1+\sum_{l<f^{-1}(j)}\lambda_l\leq1+\sum_{l<i}\lambda_l~.\] 
But note that $f^{-1}(j)\leq i$ if and only of $f(i)\geq j$. We conclude that for every $i,j$ such that $j\leq f(i)$, we have $\lambda_j-\sum_{l<i}\lambda_l\leq 1$. Thus,
\begin{eqnarray*} \label{eq:HkfSize}
  \sum_{i=0}^{F-1}\sum_{j=i}^{f(i)}\frac{1}{2^{i-1}}n^{1+\frac{1}{k}(\lambda_j-\sum_{l<i}\lambda_l)}&\leq&\sum_{i=0}^{F-1}\sum_{j=i}^{f(i)}\frac{1}{2^{i-1}}n^{1+\frac{1}{k}}
  \leq\sum_{i=0}^{F-1}\frac{f(i)-i+1}{2^{i-1}}n^{1+\frac{1}{k}}\\
  &\leq&\max_i(f(i)-i+1)n^{1+\frac{1}{k}}\cdot\sum_{i=0}^{F-1}\frac{1}{2^{i-1}}\\
  &<&4\max_i(f(i)-i+1)n^{1+\frac{1}{k}}~.
\end{eqnarray*}

We conclude that 
\[\E[|H(k,f,\{\lambda_j\},F)|]<Fn+4\max_i(f(i)-i+1)n^{1+\frac{1}{k}}=O(Fn+\max_i(f(i)-i+1)\cdot n^{1+\frac{1}{k}})~.\]
\end{proof}

Given $k,f$, the following definition chooses the largest $\{\lambda_j\}$ and the smallest $F$ that satisfy the constraints of Lemma \ref{lemma:HkfESize}.

\begin{definition} \label{def:Hkf}
Given an integer $k\geq1$ and a non-decreasing function $f:\mathbb{N}\rightarrow\mathbb{N}$ such that $\forall_i\;f(i)\geq i$, the {\em General Hopset} $H(k,f)$, is the hopset $H(k,f,\{\lambda_j\},F)$, where $\{\lambda_j\}$ and F are defined by\footnote{Note that in the definition of the sequence $\{\lambda_j\}$, no explicit base case was provided (i.e. a definition of $\lambda_0$). But, notice that the definition of $\{\lambda_j\}$ actually does contain a definition for $\lambda_0$:
\[\lambda_0=1+\sum_{l<f^{-1}(0)}\lambda_l=1+\sum_{l<0}\lambda_l=1~,\]
where $f^{-1}(0)=0$ is true by the definition of $f^{-1}$ and the fact that $f(0)\geq0$.}
\begin{enumerate}
    \item $\lambda_j=1+\sum_{l<f^{-1}(j)}\lambda_l$. 
    \item $F=\min\{F'\;|\;\sum_{l<F'}\lambda_l\geq k\}$.
\end{enumerate}
\end{definition}

\section{Stretch and Hopbound Analysis Method} \label{Stretch and Hopbound Analysis Method}

In this section we show that our general hopset can provide the state-of-the-art results for $(\alpha,\beta)$-hopsets, for various regimes of $\alpha$. 

Given a weighted undirected graph $G$, and given the parameters $k,f$ that define $H(k,f)$, we add another parameter, which is a sequence of non-negative real numbers: $\{r_i\}_{i=0}^F$. We stress that these parameters only play a part in the analysis, and not in the construction of the hopset $H(k,f)$.

The following definition of the \textit{score} of a vertex is required for the lemma that will be proved afterwards.

\begin{definition} \label{def:Score}
Given the function $f$ and the sequence $\{r_i\}$, the {\em score} of a vertex $u\in V$ is
\[score(u)=\max\{i>0\;|\;d(u,p_i(u))>r_i\text{ and }\forall_{j\in[f^{-1}(i-1),i-1]}\;d(u,p_j(u))\leq r_j\}~,\]
where if $p_i(u)$ is not defined (e.g. when $i=F$ and $A_F=\emptyset$), we consider $d(u,p_i(u))$ to be $\infty$.
\end{definition}

\begin{remark}
The set in the definition of $score(u)$ is not empty, so the score of each vertex is well defined and positive. To see this, note that if $i$ is the minimal index such that $d(u,p_i(u))>r_i$, then $i>0$ (because $p_0(u)=u$, so $d(u,p_0(u))=0\leq r_0$), $i\leq F$ (because $d(u,p_F(u))=\infty>r_F$) and also for every $j\in[f^{-1}(i-1),i-1]$, by the minimality of $i$, $d(u,p_j(u))\leq r_j$.
\end{remark}

Denote $H=H(k,f)$. The following lemma lets us ``jump" from some vertex $u$ to some other vertex $u'$,  using a path in $G\cup H$, where $u'$ is at a certain distance from $u$. These paths have low hopbound and stretch, and we will use them later to connect every pair of vertices in $G$.

\begin{lemma}[Jumping Lemma] \label{lemma:HJump}
Suppose that $score(u)=i$, then for every $u'\in V$ such that
$d(u,u')\leq\frac{r_i-r_{i-1}}{2}-r_{f^{-1}(i-1)}$,
\[d_{H}^{(3)}(u,u')\leq3d(u,u')+2(r_{i-1}+r_{f^{-1}(i-1)})~.\]
Moreover, if also $d(u,u')\geq\frac{2}{t}(r_{i-1}+r_{f^{-1}(i-1)})$ for some $t>0$, then:
\[d_{H}^{(3)}(u,u')\leq(t+3)d(u,u')~.\]
\end{lemma}

\begin{proof}
Let $u\in V$ be some vertex with $score(u)=i$ and let $u'\in V$ be some other vertex. We look at the following potential path:
\[u\rightarrow p_{f^{-1}(i-1)}(u)\rightarrow p_{i-1}(u')\rightarrow u'~.\]

\begin{center}
\begin{figure}[ht]
    \centering
    \includegraphics[width=10cm, height=5cm]{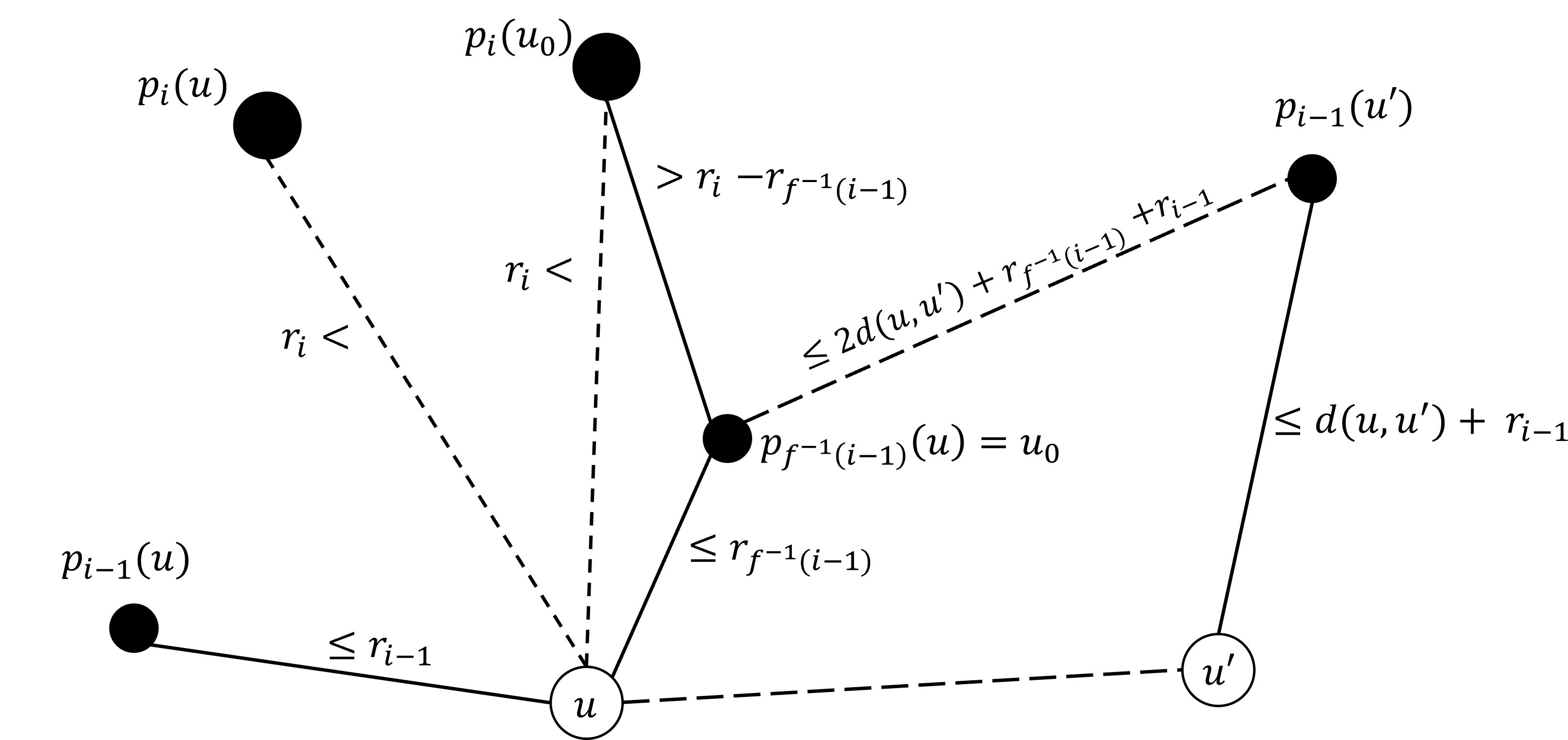}
    \caption{The potential path between $u$ and $u'$. Notice that $d(u,p_{i-1}(u))\leq r_{i-1}$ and $d(u,p_{f^{-1}(i-1)}(u))\leq r_{f^{-1}(i-1)}$, since $score(u)=i$.}
    \label{fig:jump_lemma}
\end{figure}
\end{center}

The reason we look at the $(i-1)$'th pivot of $u'$ is that this is the highest-index pivot of $u'$, such that we can still bound the distance to it:
\begin{equation}\label{eq:u'}
d(u',p_{i-1}(u'))\leq d(u',p_{i-1}(u))\leq d(u',u)+d(u,p_{i-1}(u))\leq d(u',u)+r_{i-1}~.
\end{equation}
The reason we look at the $f^{-1}(i-1)$'th pivot of $u$ is that this is the lowest-index pivot of $u$ that still can be connected to $p_{i-1}(u')$ through an $H$-edge, according to our construction.

Notice that in the path above, the first and the last arrows represent edges that do exist in $H$. We now bound the distance between $p_{f^{-1}(i-1)}(u)$ and $p_{i-1}(u')$. Recall that $score(u)=i$, and therefore $d(u,p_j(u))\leq r_j$ for every
$j\in[f^{-1}(i-1),i-1]$. In particular:
\[d(u,p_{f^{-1}(i-1)}(u))\leq r_{f^{-1}(i-1)}~,\]
and now we can see that:
\begin{equation} \label{eq:MiddleEdgeDist}
\begin{split}
d(p_{f^{-1}(i-1)}(u),p_{i-1}(u'))&\leq d(p_{f^{-1}(i-1)}(u),u)+d(u,u')+d(u',p_{i-1}(u'))\\
&\stackrel{\eqref{eq:u'}}{\leq} r_{f^{-1}(i-1)}+d(u,u')+(d(u,u')+r_{i-1})\\
&=2d(u,u')+r_{f^{-1}(i-1)}+r_{i-1}~.
\end{split}
\end{equation}

For convenience, we denote $u_0=p_{f^{-1}(i-1)}(u)$, and we also bound the distance $d(u_0,p_i(u_0))$.
\[r_i<d(u,p_i(u))\leq d(u,p_i(u_0))\leq d(u,u_0)+d(u_0,p_i(u_0))\leq r_{f^{-1}(i-1)}+d(u_0,p_i(u_0))\;\;\Rightarrow\]
\[\Rightarrow\;\;d(u_0,p_i(u_0))>r_i-r_{f^{-1}(i-1)}~.\]

Figure \ref{fig:jump_lemma} summarizes all of the computations above.

Recall that the vertex $u_0$ is connected through an edge of $H$ to any vertex $v\in B_j(u_0)$, for every $j\in[i(u_0),f(i(u_0))]$. Since $u_0$ is a $f^{-1}(i-1)$'th pivot, we know that $i(u_0)\geq f^{-1}(i-1)$, so using the fact that $f$ is non-decreasing:
\[i-1\leq f(f^{-1}(i-1))\leq f(i(u_0))~.\]

Also, since  $d(u_0,p_i(u_0))>r_i-r_{f^{-1}(i-1)}>0$, it cannot be that $i(u_0)\geq i$ (otherwise $u_0$ would be the $i$'th pivot of itself and the distance $d(u_0,p_i(u_0))$ would be $0$). We conclude that $i-1\in[i(u_0),f(i(u_0))]$. 

Therefore, $u_0$ is connected to every vertex of $B_{i-1}(u_0)$. Since $p_{i-1}(u')\in A_{i-1}$, a sufficient condition for $p_{i-1}(u')$ to be in $B_{i-1}(u_0)$, which would imply that $(u_0,p_{i-1}(u'))\in H$, is 
\[2d(u,u')+r_{f^{-1}(i-1)}+r_{i-1}\leq r_i-r_{f^{-1}(i-1)}~,\]
i.e.
\[d(u,u')\leq \frac{r_i-r_{i-1}}{2}-r_{f^{-1}(i-1)}~.\]

In case that this criteria is satisfied, $u_0=p_{f^{-1}(i-1)}(u)$ and $p_{i-1}(u')$ are connected, and we get a $3$-hops path from $u$ to $u'$. The weight of this path is:
\begin{eqnarray*}
d_{H}^{(3)}(u,u')&\stackrel{(\ref{eq:u'}),(\ref{eq:MiddleEdgeDist})}{\leq}&r_{f^{-1}(i-1)}+(2d(u,u')+r_{f^{-1}(i-1)}+r_{i-1})+(d(u,u')+r_{i-1})\\
&=&3d(u,u')+2(r_{i-1}+r_{f^{-1}(i-1)})~.
\end{eqnarray*}

Let $t>0$ be some real number. If it happens to be that the above $u,u'$ satisfy also $r_{i-1}+r_{f^{-1}(i-1)}\leq \frac{t}{2}d(u,u')$ (or equivalently $d(u,u')\geq \frac{2}{t}(r_{i-1}+r_{f^{-1}(i-1)})$), then we get a $3$-hops path between them, of weight $\leq 3d(u,u')+td(u,u')=(t+3)d(u,u')$, i.e. a stretch of $t+3$.
\end{proof}

Lemma \ref{lemma:HJump} implies that if we want to have many pairs of vertices, such that there's a $3$-hops path between them, with stretch $\approx t$, it's better to choose the parameters $\{r_i\}$ such that $\frac{2}{t}(r_{i-1}+r_{f^{-1}(i-1)})\leq\frac{r_i-r_{i-1}}{2}-r_{f^{-1}(i-1)}$, i.e. $r_i\geq(1+\frac{4}{t})r_{i-1}+(2+\frac{4}{t})r_{f^{-1}(i-1)}$. As we will see in Lemma \ref{lemma:HJump2}, using equality in this requirement is enough. Hence, from now on, we assume that the sequence $\{r_i\}$ satisfies
\begin{equation} \label{eq:rForH}
    r_i=(1+\frac{4}{t})r_{i-1}+(2+\frac{4}{t})r_{f^{-1}(i-1)}
\end{equation}

Let $u,v\in V$ be a pair of vertices, and let $u=u_0,u_1,u_2,...,u_d=v$ be the shortest path between them. We use the ``jumps" that Lemma \ref{lemma:HJump} provides, in order to find a low hopbound path in $G\cup H$ between $u,v$.

\begin{lemma} \label{lemma:HJump2}
Fix $0\leq h\leq d$, and suppose that the sequence $\{r_i\}$ satisfies Equation (\ref{eq:rForH}). Suppose that $score(u_h)=i$ and let
\[l=\max\{l'\geq h\;|\;d(u_h,u_{l'})\leq\frac{r_i-r_{i-1}}{2}-r_{f^{-1}(i-1)}\}~.\]
Then if $l<d$, we have:
\begin{enumerate}
    \item $d_{G\cup H}^{(4)}(u_h,u_{l+1})\leq(t+3)d(u_h,u_{l+1})$
    \item $d(u_h,u_{l+1})\geq\frac{4}{t}r_{f^{-1}(i-1)}$
\end{enumerate}

\end{lemma}

\begin{proof}
Denote by $W$ the weight of the edge $(u_l,u_{l+1})$. We look at two different cases.

\noindent{\bf Case 1:} $d(u_h,u_l)\geq\frac{2}{t}(r_{i-1}+r_{f^{-1}(i-1)})$.
In this case, by Lemma \ref{lemma:HJump}:
\[d_{G\cup H}^{(4)}(u_h,u_{l+1})\leq d_{G\cup H}^{(3)}(u_h,u_l)+W\leq(t+3)d(u_h,u_l)+W<(t+3)(d(u_h,u_l)+W)=(t+3)d(u_h,u_{l+1})~.\]

\noindent{\bf Case 2:} $d(u_h,u_l)<\frac{2}{t}(r_{i-1}+r_{f^{-1}(i-1)})$.
By Lemma \ref{lemma:HJump}, we have that $d_{G\cup H}^{(3)}(u_h,u_l)\leq3d(u_h,u_l)+2(r_{i-1}+r_{f^{-1}(i-1)})$. Also, recall that we assumed that Equation (\ref{eq:rForH}) holds, which is equivalent to the fact that $\frac{r_i-r_{i-1}}{2}-r_{f^{-1}(i-1)}=\frac{2}{t}(r_{i-1}+r_{f^{-1}(i-1)})$. Therefore, by the maximality of $l$,
\[d(u_h,u_{l+1})>\frac{r_i-r_{i-1}}{2}-r_{f^{-1}(i-1)}=\frac{2}{t}(r_{i-1}+r_{f^{-1}(i-1)})~.\]
We get that
\[2(r_{i-1}+r_{f^{-1}(i-1)})\leq td(u_h,u_{l+1})=t(d(u_h,u_l)+W)~,\]
and thus,
\begin{eqnarray*}
d_{G\cup H}^{(4)}(u_h,u_{l+1})&\leq& d_{G\cup H}^{(3)}(u_h,u_l)+W\\
&\leq& 3d(u_h,u_l)+2(r_{i-1}+r_{f^{-1}(i-1)})+W\\
&\leq&3d(u_h,u_l)+t(d(u_h,u_l)+W)+W\\
&=&(t+3)d(u_h,u_l)+(t+1)W\\
&<&(t+3)d(u_h,u_{l+1})~.
\end{eqnarray*}

In both cases, we saw that $d(u_h,u_{l+1})\geq\frac{2}{t}(r_{i-1}+r_{f^{-1}(i-1)})$, and since $\{r_i\}$ is a non-decreasing sequence (can be seen by Equation (\ref{eq:rForH})), we get:
\[d(u_h,u_{l+1})\geq\frac{2}{t}(r_{i-1}+r_{f^{-1}(i-1)})\geq\frac{4}{t}r_{f^{-1}(i-1)}~.\]
\end{proof}

The following theorem presents the size, the stretch and the hopbound for our hopset, $H(k,f)$. It uses Lemma \ref{lemma:HJump2} repeatedly between every pair of vertices $u,v\in V$.

\begin{theorem} \label{thm:Hkf}
Given an integer $k$, a monotone non-decreasing function $f:\mathbb{N}\rightarrow\mathbb{N}$ such that $\forall_i f(i)\geq i$, parameters $\{\lambda_j\}$ such that $\forall_j\lambda_j\leq1+\sum_{l<f^{-1}(j)}\lambda_l$ and $F$ such that $\sum_{j<F}\lambda_j\geq k$, every undirected weighted graph $G=(V,E)$ admits a hopset with the following properties, simultaneously for every $t>0$:
\begin{enumerate}
    \item Size $O(Fn+\max_i(f(i)-i+1)\cdot n^{1+1/k})$.
    \item Stretch $2t+3$.
    \item Hopbound $4r_F+3$,\newline where $\{r_i\}$ satisfies $r_0=1$ and $r_i=(1+\frac{4}{t})r_{i-1}+(2+\frac{4}{t})r_{f^{-1}(i-1)}$ for every $i>0$.
\end{enumerate}
\end{theorem}

\begin{proof}
Given $G,k,f$, build $H(k,f)$ on $G$. By Lemma \ref{lemma:HkfESize}, this hopset has the desired size in expectation.
Let $u,v\in V$ be a pair of vertices, and let $u=u_0,u_1,u_2,...,u_d=v$ be the shortest path between them. We use Lemma \ref{lemma:HJump2} to find a path between $u$ and $v$. Starting with $h=0$, find $l=\max\{l'\geq h\;|\;d(u_h,u_{l'})\leq\frac{r_i-r_{i-1}}{2}-r_{f^{-1}(i-1)}\}$, where $score(u_h)=i$. If $l=d$, stop the process and denote $v'=u_h$. Otherwise, set $h\leftarrow l+1$, and continue in the same way.

This process creates a subsequence of $u_0,...,u_d$: $u=v_0,v_1,v_2,...,v_b=v'$, such that for every $q<b$ we have $d_{G\cup H}^{(4)}(v_q,v_{q+1})\leq(t+3)d(v_q,v_{q+1})$ (by Lemma \ref{lemma:HJump2}). For $v'=v_b$ we have $d(v',v)\leq\frac{r_i-r_{i-1}}{2}-r_{f^{-1}(i-1)}$, where $score(v')=i$. For this last segment, we get from Lemma \ref{lemma:HJump} that
\[d_{G\cup H}^{(3)}(v',v)\leq3d(v',v)+2(r_{i-1}+r_{f^{-1}(i-1)})\leq3d(v',v)+4r_F\]
(again, by Equation (\ref{eq:rForH}), $\{r_i\}$ is a non-decreasing sequence).

When summing over the entire path, we get:
\begin{eqnarray*}
d_{G\cup H}^{(4b+3)}(u,v)&\leq&\sum_{q=0}^{b-1}(t+3)d(v_q,v_{q+1})+3d(v',v)+4r_F\\
&=&(t+3)d(u,v')+3d(v',v)+4r_F\\
&\leq&(t+3)d(u,v)+4r_F~.
\end{eqnarray*}

To bound $b$, we notice that by Lemma \ref{lemma:HJump2}, for every $q<b$ (here $i=score(v_q)$):
\[d(v_q,v_{q+1})\geq\frac{4}{t}r_{f^{-1}(i-1)}\geq\frac{4}{t}r_0~.\]

So, the number of these ``jumps" couldn't be greater than $\frac{d(u,v)}{\frac{4}{t}r_0}=\frac{t\cdot d(u,v)}{4r_0}$, and we finally got:
\[d_{G\cup H}^{(\frac{t\cdot d(u,v)}{r_0}+3)}(u,v)\leq(t+3)d(u,v)+4r_F~.\]

The above is true for the sequence $\{r_i\}$, which satisfies Equation (\ref{eq:rForH}) and $r_0=1$. However, note that we never used the fact that $r_0=1$, thus it is true for {\em any} sequence $\{r'_i\}$ that satisfies Equation (\ref{eq:rForH}). Given the vertices $u,v$, we define a new sequence as follows:
\[r'_i=\frac{t\cdot d(u,v)\cdot r_i}{4r_F}~.\]
The new sequence clearly still satisfies Equation (\ref{eq:rForH}), so if we use it instead of $\{r_i\}$, we get that for our specific $u,v$:
\[d_{G\cup H}^{(\frac{t\cdot d(u,v)}{r'_0}+3)}(u,v)\leq(t+3)d(u,v)+4r'_F\;\;\;\;\Rightarrow\]
\[\Rightarrow\;\;\;\;d_{G\cup H}^{(4r_F+3)}(u,v)\leq(t+3)d(u,v)+t\cdot d(u,v)=(2t+3)d(u,v)~,\]
i.e. the stretch of this new path is $2t+3$, and its hopbound is $4r_F+3$.

Although we chose $\{r'_i\}$ for a specific pair of vertices, this choice of $\{r'_i\}$ {\em doesn't change our construction at all}, but only the analysis. Hence, we conclude that for each $u,v\in V$, there is a path between them in $G\cup H$, with stretch $2t+3$ and hopbound $4r_F+3$, for our original sequence $\{r_i\}$.
\end{proof}

\subsection{Applying our Construction and Analysis for Rounding Functions} \label{sec:RoundingFunctionsHopsets}


In this section we apply the result of Theorem \ref{thm:Hkf} to functions of the form
\[f(i)=\left\lfloor\frac{i}{c}\right\rfloor\cdot c+c-1~,\]
where $c>0$ is some integer parameter. For short, this family of functions is referred as \textit{Rounding Functions}. In particular, for a given integer $c>0$, this function rounds $i$ to the minimal multiple of $c$ that is larger than $i$, then subtracts $1$. This generalizes both the construction of hopsets that is based on \cite{TZ01} - as every vertex $u\in V$ connects to every $v\in B_j(u)$ \textit{for every} $j$ between $i(u)$ and $k-1$ - and the construction of hopsets by \cite{EN19,HP17}, where every $u\in V$ connects to every $v\in B_j(u)$, only for $j=i(u)$. The former is achieved by setting $c=k$, then for every $i\leq k-1$, $f(i)=k-1$. The latter is achieved by setting $c=1$, then $f(i)=i$ for every $i$.

Let $G=(V,E)$ be an undirected weighted graph, let $u,v\in V$ be a pair of vertices, and let $u=u_0,u_1,u_2,...,u_d=v$ be the shortest path between them. In the proof of Theorem \ref{thm:Hkf}, we found a path between $u$ and $v$, by concatenating the \textit{detours} that Lemma \ref{lemma:HJump2} provided. The properties of these detours were that they consist of at most $4$ hops in $G\cup H(k,f)$, their stretch is at most $(t+3)$, and the part of the path they skip is of weight at least $\frac{4}{t}r_0$ (the exact details appear in the lemma). 

In the case of our rounding function $f$, we can improve the hopbound of the resulting path. For this purpose, we make the following observation. The part of the path that each detour skips is actually of weight at least $\frac{4}{t}r_{f^{-1}(i-1)}$, where $i$ is the score of the starting vertex of the detour. Instead of bounding it by $\frac{4}{t}r_{f^{-1}(i-1)}\geq\frac{4}{t}r_0$, we may consider two cases. If the score of the starting vertex is large, then this bound improves. Specifically, for our rounding function $f$, if the score is at least $c$, then this weight is at least $\frac{4}{t}r_c$. Otherwise, if the score is smaller than $c$, we find another $2$-hops detour with stretch at most $2c-1$, that skips a part of the $u-v$ shortest path that has weight roughly $\frac{r_c}{c}$. The existence of these new detours is proved in Lemma \ref{lemma:InvolvedJump} below. The resulting stretch of this new process is the maximum between $2c-1$ (which is the stretch of the new detours) and $2t+3$ (which is the resulting stretch in Theorem \ref{thm:Hkf}). To balance these, we choose $t\approx c$. Then, while the overall stretch is $O(c)$, each detour we use skips a part of the $u-v$ shortest path, that has weight at least $\frac{r_c}{c}$. Computing the elements of the sequence $\{r_i\}$, for our specific function $f$, gives $\frac{r_c}{c}\approx r_0$. For comparison, recall that the part of the path that each detour skipped in Theorem \ref{thm:Hkf} was of weight roughly $\frac{r_0}{t}\approx\frac{r_0}{c}$. As a result, the hopbound reduces by a factor of $c$.

The existence of these new detours is proved in the following lemma. This lemma uses the technique of \cite{TZ01} for the first $c$ levels of the hierarchy, and also uses some properties of the function $score(u)$.

\begin{lemma} \label{lemma:InvolvedJump}
Fix $0\leq h\leq d$. Suppose that $score(u_h)\leq c$ and let
\[l=\max\{l'\geq h\;|\;d(u_h,u_{l'})\leq\frac{r_c}{c}\}~.\]
Then, $d_{H}^{(2)}(u_h,u_l)\leq(2c-1)d(u_h,u_l)$.

Moreover, if $l<d$, then $d_{G\cup H}^{(3)}(u_h,u_{l+1})\leq(2c-1)d(u_h,u_{l+1})$.
\end{lemma}

Before we prove Lemma \ref{lemma:InvolvedJump}, we prove the following fact, which can be viewed as the main component of the stretch analysis in \cite{TZ01}.

\begin{fact} \label{fact:TZ2001}
For $x,y\in V$ such that $d(x,p_c(x))>c\cdot d(x,y)$,
\[d_{H}^{(2)}(x,y)\leq (2c-1)\cdot d(x,y)~.\]
\end{fact}

\begin{proof}
Let $i^*$ be the minimal index such that $p_{i^*}(x)\in B_{i^*}(y)$ or $p_{i^*}(y)\in B_{i^*}(x)$. We prove by induction that for every $j\leq i^*$,
\[d(x,p_j(x))\leq j\cdot d(x,y)~~~\textrm{and}~~~d(y,p_j(y))\leq j\cdot d(x,y)~.\]

For $j=0$, this is trivial since $p_0(x)=x$ and $p_0(y)=y$. For $j>0$, we assume by induction that both $d(x,p_{j-1}(x))\leq(j-1)d(x,y)$ and $d(y,p_{j-1}(y))\leq(j-1)d(x,y)$. The fact that $p_{j-1}(y)\notin B_{j-1}(x)$ indicates that $d(x,p_{j-1}(y))\geq d(x,p_j(x))$. Hence,
\[d(x,p_j(x))\leq d(x,p_{j-1}(y))\leq d(x,y)+d(y,p_{j-1}(y))\leq d(x,y)+(j-1)d(x,y)=j\cdot d(x,y)~.\]
Proving $d(y,p_j(y))\leq j\cdot d(x,y)$ is symmetric. This concludes the inductive proof.

As a result, since we assume $d(x,p_c(x))>c\cdot d(x,y)$, it must be that $i^*<c$. Recall that by the definition of $i^*$, either $p_{i^*}(x)\in B_{i^*}(y)$ or $p_{i^*}(y)\in B_{i^*}(x)$. We assume that $p_{i^*}(x)\in B_{i^*}(y)$, and the other case is symmetric. By the definition of $H=H(k,f)$, the edge $(x,p_{i^*}(x))$ is contained in $H$. In addition, note that $f(i)\geq c-1\geq i^*$ for every $i\geq0$, and that $i(y)\leq i^*$ (otherwise $B_{i^*}(y)$ would be empty, in contradiction to the fact that $p_{i^*}(x)\in B_{i^*}(y)$). Thus, $i^*\in[i(y),f(i(y))]$, and since $p_{i^*}(x)\in B_{i^*}(y)$, we conclude that the edge $(y,p_{i^*}(x))$ is contained in $H$ as well.

Then,
\begin{eqnarray*}
d_{H}^{(2)}(x,y)&\leq&d(x,p_{i^*}(x))+d(p_{i^*}(x),y)
\leq d(x,p_{i^*}(x))+d(p_{i^*}(x),x)+d(x,y)\\
&=&2d(x,p_{i^*}(x))+d(x,y)
\leq2i^*\cdot d(x,y)+d(x,y)\\
&=&(2i^*+1)d(x,y)
\leq(2c-1)d(x,y)~.
\end{eqnarray*}

\end{proof}

Another useful fact for the proof of Lemma \ref{lemma:InvolvedJump} is the following.

\begin{fact} \label{fact:LowScore}
Every vertex $u\in V$ such that $score(u)\leq c$ satisfies $d(u,p_c(u))>r_c$.
\end{fact}

\begin{proof}

Let $i$ be the minimal index that is greater or equal to $c$ and $d(u,p_i(u))>r_i$. There is such $i$, because for $i=F$, the pivot $p_F(u)$ does not exist, thus the distance $d(u,p_F(u))$ is $\infty$. If $i=c$, we are done. Otherwise, notice that $f^{-1}(i-1)\geq c$ (in general, $f^{-1}(i')=\left\lfloor\frac{i'}{c}\right\rfloor\cdot c$), so for all $j\in[f^{-1}(i-1),i-1]$ we have $d(u,p_j(u))\leq r_j$. By Definition \ref{def:Score} of $score$, that means that $score(x)=i>c$, in contradiction to the assumption that $score(u)\leq c$.

\end{proof}

We are now ready to prove Lemma \ref{lemma:InvolvedJump}.

\begin{proof}[Proof of Lemma \ref{lemma:InvolvedJump}]

By Fact \ref{fact:LowScore}, since $score(u_h)\leq c$, then $d(u_h,p_c(u_h))>r_c$. Note that $l$ is defined such that $d(u_h,u_l)\leq\frac{r_c}{c}$. Combining these two inequalities, we get
\[d(u_h,p_c(u_h))>r_c\geq c\cdot d(u_h,u_l)~.\]
By Fact \ref{fact:TZ2001}, it means that
\[d_{H}^{(2)}(u_h,u_l)\leq(2c-1)d(u_h,u_l)~.\]

This proves the first part of the lemma. For the second part, assume that $l<d$ and the weight of the edge $(u_l,u_{l+1})$ is $W$. Then,
\[d_{G\cup H}^{(3)}(u_h,u_{l+1})\leq(2c-1)d(u_h,u_l)+W\leq(2c-1)(d(u_h,u_l)+W)=(2c-1)d(u_h,u_{l+1})~.\]

\end{proof}

Next, we state and prove the strong version of Theorem \ref{thm:Hkf}, for the function $f(i)=\lfloor\frac{i}{c}\rfloor\cdot c+c-1$.

\begin{theorem} \label{thm:StrongHfk}
Given two integer parameters $1<c,k$, every undirected weighted graph $G=(V,E)$ admits a hopset with the following properties.
\begin{enumerate}
    \item Size $O(c\log_ck\cdot n+c\cdot n^{1+1/k})$.
    \item Stretch $8c+3$.
    \item Hopbound $O\left(\frac{k^{1+\frac{2}{\ln c}}}{c}\right)$.
\end{enumerate}
\end{theorem}

\begin{proof}

Given $G,k,c$, let $H=H(k,f)$, for the function $f(i)=\lfloor\frac{i}{c}\rfloor\cdot c+c-1$. Recall the values $\{\lambda_j\},F,\{r_i\}$ from the definition of $H(k,f)$. Here, in the definition of $\{r_i\}$, we set the parameter $t$ to be $4c$. For our specific choice of the function $f$, we prove in Appendix \ref{appendix:FProperties} that
\begin{enumerate}
    \item $f^{-1}(i)=\lfloor\frac{i}{c}\rfloor\cdot c$.
    \item $\lambda_{ac+b}=(c+1)^a$ for every integers $a\geq0$ and $b\in[0,c-1]$.
    \item $F\leq\lceil\log_{c+1}(k+1)\rceil\cdot c$. 
    \item $r_{ac}=\left[\left(1+\frac{4}{t}\right)^c\left(\frac{t}{2}+2\right)-\left(\frac{t}{2}+1\right)\right]^ar_0$ for every $a\geq0$.
\end{enumerate}

By Lemma \ref{lemma:HkfESize} and Item $3$, the hopset $H$ has size 
\[O(Fn+\max_i(f(i)-i+1)\cdot n^{1+\frac{1}{k}})=O(c\log_ck\cdot n+c\cdot n^{1+\frac{1}{k}})~,\]
as desired.

Substituting our choice $t=4c$ in Item $4$, we get for every $a\geq0$
\[r_{ac}=\left[\left(1+\frac{4}{t}\right)^c\left(\frac{t}{2}+2\right)-\left(\frac{t}{2}+1\right)\right]^ar_0
=\left[\left(1+\frac{1}{c}\right)^c\left(2c+2\right)-\left(2c+1\right)\right]^ar_0~.\]
Using the inequality $2^x\leq1+x\leq e^x$, that holds for every $x\in[0,1]$, we get
\begin{equation} \label{eq:RBounds}
\left[2c+3\right]^ar_0=\left[2\cdot\left(2c+2\right)-\left(2c+1\right)\right]^ar_0\leq r_{ac}\leq\left[e\cdot\left(2c+2\right)-\left(2c+1\right)\right]^ar_0<\left[4c+5\right]^ar_0~.
\end{equation}

Let $u,v\in V$ be a pair of vertices, and let $u=u_0,u_1,u_2,...,u_d=v$ be the shortest path between them. We use Lemmas \ref{lemma:HJump2} and \ref{lemma:InvolvedJump} to find a path between $u$ and $v$. Starting with $h=0$, define $l$ as follows:
\begin{enumerate}
    \item If $i=score(u_h)>c$, set $l=\max\{l'\geq h\;|\;d(u_h,u_{l'})\leq\frac{r_i-r_{i-1}}{2}-r_{f^{-1}(i-1)}\}$.
    \item Otherwise, set $l=\max\{l'\geq h\;|\;d(u_h,u_{l'})\leq\frac{r_c}{c}\}$.
\end{enumerate}
If $l=d$, stop the process and denote $v'=u_h$. Otherwise, set $h\leftarrow l+1$, and continue in the same way.

This process creates a sub-sequence $u=v_0,v_1,v_2,...,v_b=v'$ of $u_0,...,u_d$, such that for every $q<b$, if $score(v_q)>c$, we have $d_{G\cup H}^{(4)}(v_q,v_{q+1})\leq(4c+3)d(v_q,v_{q+1})$ (by Lemma \ref{lemma:HJump2}). For $q<b$, if $score(v_q)\leq c$, then by Lemma \ref{lemma:InvolvedJump}, we have $d_{G\cup H}^{(3)}(v_q,v_{q+1})\leq(2c-1)d(v_q,v_{q+1})<(4c+3)d(v_q,v_{q+1})$. 

For $v'=v_b$, if $score(v')\leq c$, then by Lemma \ref{lemma:InvolvedJump} we still have $d_{G\cup H}^{(3)}(v',v)\leq(2c-1)d(v',v)<(4c+3)d(v',v)$. If $score(v')>c$, then by definition $d(v',v)\leq\frac{r_i-r_{i-1}}{2}-r_{f^{-1}(i-1)}$, where $i=score(v')$. By Lemma \ref{lemma:HJump},
\[d_{G\cup H}^{(3)}(v',v)\leq3d(v',v)+2(r_{i-1}+r_{f^{-1}(i-1)})\leq3d(v',v)+4r_F<(4c+3)d(v',v)+4r_F~.\]

When summing over the entire path, we get:
\[d_{G\cup H}^{(4b+3)}(u,v)\leq\sum_{q=0}^{b-1}(4c+3)\cdot d(v_q,v_{q+1})+(4c+3)\cdot d(v',v)+4r_F=(4c+3)d(u,v)+4r_F~.\]

To bound $b$, we notice that by Lemma \ref{lemma:HJump2}, for every $q<b$ such that $score(v_q)=j>c$ (recall that we chose $t=4c$):
\[d(v_q,v_{q+1})\geq\frac{4}{4c}r_{f^{-1}(j-1)}\geq\frac{1}{c}r_{f^{-1}(c)}=\frac{r_c}{c}~.\]
For $q<b$ such that $score(v_q)\leq c$, we choose $v_{q+1}$ such that $d(v_q,v_{q+1})\geq\frac{r_c}{c}$. To see this, recall that if $v_q=u_h$, then $v_{q+1}=u_{l+1}$, where $u_l$ is the last vertex such that $d(u_h,u_l)\leq\frac{r_c}{c}$. Therefore, $d(v_q,v_{q+1})=d(u_h,u_{l+1})>\frac{r_c}{c}$.

In conclusion, we saw that $d(v_q,v_{q+1})\geq\frac{r_c}{c}$ for \textit{any} $q<b$. That means that the number of these detours cannot be larger than $\frac{d(u,v)}{\frac{r_c}{c}}=\frac{c\cdot d(u,v)}{r_c}$. Substituting in $b$, we get
\begin{equation} \label{eq:StretchAndHopbound}
d_{G\cup H}^{(\frac{4c\cdot d(u,v)}{r_c}+3)}(u,v)\leq(4c+3)d(u,v)+4r_F~.
\end{equation}

Recall from Equation (\ref{eq:RBounds}) that $(2c+3)^ar_0\leq r_{ac}\leq(4c+5)^ar_0$ for every integer $a\geq0$. Here, $r_0$ is an arbitrary parameter that does not affect neither the construction of the hopset $H=H(k,f)$, nor the analysis we made above. Thus, for the sake of the analysis of the stretch and hopbound of a specific pair of vertices $u,v$, we may choose $r_0$ dependent on $d(u,v)$ as follows.
\[r_0=\frac{c}{(4c+5)^{\lceil\log_{c+1}(k+1)\rceil}}\cdot d(u,v)~.\]
To bound $r_F$, notice that $\{r_i\}$ is non-decreasing\footnote{This can be easily observed by the definition: $r_i=(1+\frac{4}{t})r_{i-1}+(2+\frac{4}{t})r_{f^{-1}(i-1)}>1\cdot r_{i-1}+0\cdot r_{f^{-1}(i-1)}=r_{i-1}$~.}, and recall that $F\leq\lceil\log_{c+1}(k+1)\rceil\cdot c$. Therefore,
\[r_F\leq r_{\lceil\log_{c+1}(k+1)\rceil\cdot c}\leq(4c+5)^{\lceil\log_{c+1}(k+1)\rceil}r_0=c\cdot d(u,v)~.\]
We substitute it in Equation (\ref{eq:StretchAndHopbound}) and get
\[d_{G\cup H}^{(\frac{4c\cdot d(u,v)}{r_c}+3)}(u,v)\leq(4c+3)d(u,v)+4c\cdot d(u,v)=(8c+3)d(u,v)~.\]

We use Equation (\ref{eq:RBounds}) again, to bound $r_c~$.
\[r_c\geq(2c+3)r_0=\frac{(2c+3)c}{(4c+5)^{\lceil\log_{c+1}(k+1)\rceil}}\cdot d(u,v)~.\]
Thus, the hopbound for $u,v$ is at most
\begin{eqnarray*}
\frac{4c\cdot d(u,v)}{r_c}+3&\leq&\frac{4(4c+5)^{\lceil\log_{c+1}(k+1)\rceil}}{2c+3}+3
=O\left(\frac{(5(c+1))^{\log_{c+1}(k+1)}}{c}\right)\\
&=&O\left(\frac{k^{\log_{c+1}(5(c+1))}}{c}\right)
=O\left(\frac{k^{1+\log_{c+1}5}}{c}\right)\\
&=&O\left(\frac{k^{1+\frac{\ln5}{\ln(c+1)}}}{c}\right)
=O\left(\frac{k^{1+\frac{2}{\ln c}}}{c}\right)~.
\end{eqnarray*}

We conclude that the hopset $H=H(k,f)$ has stretch $8c+3$, hopbound $O\left(\frac{k^{1+\frac{2}{\ln c}}}{c}\right)$ and size 
\[O(c\log_ck\cdot n+c\cdot n^{1+\frac{1}{k}})~,\]
as desired.

\end{proof}

\subsection{Special Cases} \label{sec:HopsetSpecialCases}

To show how Theorems \ref{thm:Hkf} and \ref{thm:StrongHfk} generalize the constructions of all the state-of-the-art hopsets, we substitute several values in the function $f$ or the parameter $c$.

\vspace{5mm}
\noindent{\textbf{Stretch $\boldsymbol{(1+\epsilon)}$}.}
Choose the function $f(i)=i$. The resulting hopset $H(k,f)$ has size $O(n\log k+n^{1+\frac{1}{k}})$. It could be easily verified that the hopset $H(k,f)$ contains the construction of the hopsets of \cite{EN19,HP17}. Hence, this hopset has the same stretch and hopbound as these hopsets, namely, it has stretch $1+\epsilon$ and hopbound $O\left(\frac{\log k}{\epsilon}\right)^{\log k}$. In fact, \cite{EN19,HP17} show a different analysis of their hopsets, with the same stretch and hopbound, that exploit hop-edges from a vertex to its bunch members and only to its \textit{next} pivot (that is, the edges from a vertex $u\in A_i\setminus A_{i+1}$ to all the vertices in $B_j(u)$, for all $j>i$, and only to $p_{i+1}(u)$, instead of to $p_j(u)$ for every $j>i$). Following this analysis, and a slight modification of the sample probabilities of the sets $A_i$, the authors of \cite{EN19,HP17} provided hopsets with an improved size of $O(n^{1+\frac{1}{k}})$.

\vspace{5mm}
\noindent{\textbf{Stretch $\boldsymbol{(3+\epsilon)}$}.}
In Theorem \ref{thm:Hkf}, choose the function $f(i)=i$. Note that $f$ is a \textit{rounding function} as in Section \ref{sec:RoundingFunctionsHopsets}, with $c=1$. By the results in Appendix \ref{appendix:FProperties}, $\lambda_j=2^j$, $F\leq\lceil\log_2(k+1)\rceil$, and
\[r_i=\left[\left(1+\frac{4}{t}\right)\left(\frac{t}{2}+2\right)-\left(\frac{t}{2}+1\right)\right]^i=\left(3+\frac{8}{t}\right)^i~.\]
The result is a hopset $H$, that simultaneously for every $t>0$ has stretch $2t+3$, hopbound\newline 
$O\left(\left(3+\frac{8}{t}\right)^{\lceil\log_2(k+1)\rceil}\right)=O\left(\frac{1}{t}\cdot k^{\log_2\left(3+\frac{8}{t}\right)}\right)$ and size $O(n\log k+n^{1+\frac{1}{k}})$. Setting $t=\frac{\epsilon}{2}$, it has stretch $3+\epsilon$ and hopbound $O\left(\frac{1}{\epsilon}\cdot k^{\log_2\left(3+\frac{16}{\epsilon}\right)}\right)$. This result is essentially the same as the result of \cite{BP20}, with the factor of $\log\Lambda$ removed from the size\footnote{In Table \ref{table:hopsets}, we omitted the dependence on $k$ from the size. Actually, the size of the hopset in \cite{BP20} was $O((n^{1+\frac{1}{k}}+n\log k)\cdot\log\Lambda)$, and thus the improvement is truly by a multiplicative factor of $\log\Lambda$.}. This result improves the result of \cite{EGN19} as well, since it enables the exponent of $k$ in the hopbound be smaller than $2$ (for a large constant $\epsilon$). We note, however, that \cite{EGN19} does not have a coefficient of $O\left(\frac{1}{\epsilon}\right)$ on the hopbound.
On the other hand, the construction itself of \cite{EGN19} is identical to the construction of $H(k,f)$ in this case. Hence, the analysis of \cite{EGN19} shows that our hopset actually has the same properties as in their results.

\vspace{5mm}
\noindent{\textbf{Constant Stretch}.}
In Theorem \ref{thm:StrongHfk}, choose a constant $c$. The result is a hopset $H$, with stretch $8c+3=O(1)$, hopbound 
$O\left(k^{1+\frac{2}{\ln c}}\right)$ and size $O(n\log k+n^{1+\frac{1}{k}})$. This improves the hopset of \cite{BP20} for constant stretch (specifically, it removes the factor $\log\Lambda$ from the size).

\vspace{5mm}
\noindent{\textbf{Stretch $\boldsymbol{O(k^\epsilon)}$}.}
In Theorem \ref{thm:StrongHfk}, choose $c=\lceil k^\epsilon\rceil$ for some $0<\epsilon<1$. The result is a hopset $H$, with stretch $8\lceil k^\epsilon\rceil+3=O( k^\epsilon)$, hopbound
$O\left(\frac{k^{1+\frac{2}{\ln c}}}{c}\right)=O\left(\frac{k\cdot e^{\frac{2}{\epsilon}}}{k^\epsilon}\right)=O(e^{\frac{2}{\epsilon}}k^{1-\epsilon})=O_\epsilon(k^{1-\epsilon})$ and size $O(\epsilon^{-1}k^\epsilon\cdot n+k^\epsilon\cdot n^{1+\frac{1}{k}})=O_\epsilon(k^\epsilon\cdot n^{1+\frac{1}{k}})$. This improves the hopset of \cite{BP20} for constant stretch (again by removing the factor $\log\Lambda$ from the size).

\vspace{5mm}
\noindent{\textbf{Stretch $\boldsymbol{2k-1}$}.}
In Theorem \ref{thm:StrongHfk}, choose $c=\frac{k-2}{4}$. The result is a hopset $H$, with stretch $8c+3=2k-1$, hopbound 
$O\left(\frac{k^{1+\frac{2}{\ln c}}}{c}\right)=O\left(\frac{k^{1+\frac{2}{\ln k-\ln8}}}{k}\right)=O(1)$ and size $O(k\cdot n^{1+\frac{1}{k}})$.

We note, however, that when choosing $c=k$, Lemma \ref{lemma:InvolvedJump} implies that $d_{G\cup H}^{(2)}(u,v)=(2k-1)d(u,v)$ for every $u,v\in V$. This is true because $r_c=r_k=\infty$, thus, for $j=0$ we get $u_h=u$ and $u_l=v$. Therefore, the actual stretch and hopbound of the resulting hopset $H(k,f)$ are $2k-1$ and $2$ respectively, matching those of the hopset that is based on \cite{TZ01} from \cite{BP20} (in fact, they are exactly the same hopset).

\section{A Unified Construction of Spanners and Emulators} \label{TZ Spanners}

In this section we present two spanner constructions that are based on our general hopset algorithm from Section \ref{Generalized TZ Construction}. The main difference arises due to the fact that a spanner is a sub-graph, so we need to replace the hopset edges by shortest paths. As a consequence, the size of the spanner might increase. We show two methods, the first used by \cite{BP20} and the second by \cite{EGN19}, that handle this problem. The first solution results in a spanner with a better additive stretch than the second, but requires the desired stretch as an input (thus, not \textit{simultaneous}). The second solution results in a spanner with a smaller size, that works for all the possible multiplicative stretches simultaneously, but has a larger additive stretch. In addition to these two constructions of spanners, we also present a similar construction of an emulator, using the same techniques.

\subsection{Non-Simultaneous Spanners} \label{NSSpanner}

The first approach to construct a spanner, which is analogous to the hopset $H(k,f)$ from the previous sections, is to replace each hopset edge of $H(k,f)$ by a shortest path, only if its length is bounded by some value. More accurately, we replace every ``bunch edge", i.e., a hop-edge of the form $(u,v)$, where $u\in V$ and $v\in B_j(u)$ such that $j\in[i,f(i)]$, by a $u-v$ shortest path, only if $d_G(u,v)\leq r_F$. Here, $\{r_i\}$ is the same sequence as in Theorem \ref{thm:Hkf}. On the other hand, hop-edges of the form $(u,p_i(u))$, where $u\in V$ and $0\leq i\leq F-1$, are replaced by a shortest path regardless the distance $d_G(u,p_i(u))$. We call these hop-edges ``pivot edges".

The two key observations in the analysis of this new spanner, are that (1) in the analysis of $H(k,f)$ (specifically, in Lemma \ref{lemma:HJump}), the only bunch edges that we ever use are of weight at most $r_F$, and (2) the union of the shortest paths that replace pivot edges forms a union of $F$ forests, and thus the total number of these edges is at most $Fn$. The total number of edges that replace bunch edges, on the other hand, is at most $r_F$ times the number of bunch edges in $H(k,f)$ (recall that here we assume that the input graph is \textit{unweighted}). In Lemma \ref{lemma:HkfESize} we saw that the number of bunch edges in $H(k,f)$ is at most $4\max_i(f(i)-i+1)n^{1+\frac{1}{k}}$. Thus, we conclude that the number of edges in the new spanner is at most
\[Fn+r_F\cdot4\max_i(f(i)-i+1)n^{1+\frac{1}{k}}=O(Fn+r_F\cdot\max_i(f(i)-i+1)\cdot n^{1+\frac{1}{k}})~.\]

A slightly different analysis is required to show that the new hopset we constructed has a multiplicative stretch $\alpha$ which is essentially equal to the stretch of $H(k,f)$, and an additive stretch $\beta$ which is essentially equal to the hopbound of $H(k,f)$. Specifically, we use the fact that the graph is now unweighted, and we increase the sequence $\{r_i\}$ such that for every $i$,
\begin{equation} \label{eq:NewR}
    \frac{r_i-r_{i-1}}{2}-r_{f^{-1}(i-1)}\geq\frac{2}{t}(r_{i-1}+r_{f^{-1}(i-1)})\boldsymbol{+1}~.
\end{equation}

Consider two vertices $u,v\in V$ and the $u-v$ shortest path $P$ between them, and let $x$ be a vertex with $score(x)=i$ on the path $P$. Using the modified definition of $\{r_i\}$ and the fact that the graph is unweighted, we conclude that there must be a vertex $y$ further on $P$, say, in the direction from $u$ to $v$, that satisfies $\frac{2}{t}(r_{i-1}+r_{f^{-1}(i-1)})\leq d_G(x,y)\leq\frac{r_i-r_{i-1}}{2}-r_{f^{-1}(i-1)}$, unless $x$ is too close to $v$. This enables us to use the $3$-hops detours from Lemma \ref{lemma:HJump} between $x$ and $y$, without the need to use an additional $G$-edge, as demonstrated in Lemma \ref{lemma:HJump2}. Here, of course, these $3$ hops translate to shortest paths. For the case where $x$ is too close to $v$, in the analysis of $H(k,f)$ we modified the definition of $\{r_i\}$ to handle the last detour from $x$ to $v$. Here, instead, we simply consider the length of this last detour as the additive stretch of our spanner.

For completeness, and for presenting the specific details of the above description, we bring here the full construction and analysis of this new spanner.

Let $G=(V,E)$ be an undirected unweighted graph, let $k\geq1$ be an integer, $f:\mathbb{N}\rightarrow\mathbb{N}$ be some non-decreasing function such that $\forall_i\;f(i)\geq i$ and $t>0$ be a real number.

The basic definitions, those of $\{\lambda_j\}$ and $F$, the sets $\{A_i\}$, the pivots $\{p_i(u)\}$ and the bunches $\{B_i(u)\}$ (where $u\in V, i\in[0,F]$), are the same as in Definition \ref{def:Hkf} of $H(k,f)$. For completeness, we write them again:
\begin{enumerate}
    \item $f^{-1}(i)\coloneqq\min\{j\;|\;f(j)\geq i\}$.
    \item $\lambda_j\coloneqq1+\sum_{l<f^{-1}(j)}\lambda_l$.
    \item $F\coloneqq\min\{F'\;|\;\sum_{l<F'}\lambda_l\geq k\}$.
    \item $A_0\coloneqq V$, and for every $0\leq j<F-1$, every vertex of $A_j$ is sampled independently to $A_{j+1}$ with probability $\frac{1}{2}n^{-\frac{\lambda_j}{k}}$. Lastly, $A_F=\emptyset$.
    \item $i(u)\coloneqq\max\{i\;|\;u\in A_i\}$.
    \item $p_j(u)$ is the closest vertex to $u$ from $A_j$.
    \item $B_j(u)\coloneqq\{v\in A_j\;|\;d(u,v)<d(u,p_{j+1}(u))\}$.
\end{enumerate}

The number $t>0$ is going to be approximately the multiplicative stretch of the spanner we construct, which will be denoted by $S(k,f,t)$. One critical difference between the construction of $S(k,f,t)$ and the construction of $H(k,f)$ is that now $t$ is given in advance, and the spanner will be built based on it. That means that unlike the case of the hopset, here we build a different spanner for every desired $t$, i.e. this spanner is not simultaneous. 


We define a sequence $\{r_i\}$ similarly to its definition in Theorem \ref{thm:Hkf}, but with a slight change. We still define $r_0\coloneqq1$, but this time, to satisfy Inequality (\ref{eq:NewR}), we add $2$ to the recursive formula that defines $r_{i+1}$.
\[r_{i+1}=(1+\frac{4}{t})r_i+(2+\frac{4}{t})r_{f^{-1}(i)}+2~.\]

Now we are ready to define $S(k,f,t)$. Recall that for every $x,y\in V$, $P_{x,y}$ denotes the shortest path between $x$ and $y$.
\begin{definition} \label{def:Skft}
Given a positive integer parameter $k>0$, a positive real parameter $t>0$ and a non-decreasing function $f:\mathbb{N}\rightarrow\mathbb{N}$, such that $\forall_i\;f(i)\geq i$, the non-simultaneous spanner $S(k,f,t)$ is defined as
\[S(k,f,t)=\bigcup_{u\in V}\bigcup_{j=0}^{F-1}E(P_{u,p_j(u)})\cup\bigcup_{u\in V}\bigcup_{j=i(u)}^{f(i(u))}\bigcup_{\substack{v\in B_j(u) \\ \boldsymbol{d(u,v)\leq r_F}}}E(P_{u,v})~,\]
where $E(P)$ is the set of edges of the path $P$.
\end{definition}

In words, $S(k,f,t)$ consists of the union of all the shortest paths of the form $P_{u,p_j(u)}$, for every $u$ and $j$, and of all the shortest paths $P_{u,v}$ such that $v\in B_j(u)$ for some $j\in[i(u),f(i(u))]$ and also $d(u,v)\leq r_F$.

The limitation on the length of the ``bunch paths" enables us to bound the number of added edges by $r_F\cdot|B_j(u)|$, for specific $u$ and $j$. The following lemma uses this bound for computing the size of $S(k,f,t)$.

\begin{lemma} \label{lemma:S1ESize}
\[\E[|S(k,f,t)|]=O(Fn+r_F\cdot\max_i(f(i)-i+1)\cdot n^{1+\frac{1}{k}})~.\]
\end{lemma}

\begin{proof}

Note that $S(k,f,t)$ consists of two types of paths, i.e., $S(k,f,t)=S_1\cup S_2$, where
\[S_1=\bigcup_{u\in V}\bigcup_{j=0}^{F-1}E(P_{u,p_j(u)})\]
(``pivot paths"), and
\[S_2=\bigcup_{u\in V}\bigcup_{j=i(u)}^{f(i(u))}\bigcup_{\substack{v\in B_j(u) \\ d(u,v)\leq r_F}}E(P_{u,v})\]
(``bunch paths").

For computing the size of $S_1$, notice that if some $v\in V$ is on the shortest path between $u$ and $p_j(u)$, then $p_j(v)=p_j(u)$. That is because
\[d(u,p_j(v))\leq d(u,v)+d(v,p_j(v))\leq d(u,v)+d(v,p_j(u))=d(u,p_j(u))~,\]
so by $p_j(u)$'s definition, it must be that $d(u,p_j(u))=d(u,p_j(v))$. Therefore, the inequalities above must be equalities, and thus $d(v,p_j(u))=d(v,p_j(v))$. Since we assume that ties are broken in a consistent manner, we get that $p_j(u)=p_j(v)$.

Therefore, for a fixed $j$ and a fixed $x\in A_j$, when looking at the union of all the shortest paths $P_{u,x}$ where $x=p_j(u)$, this union is a tree - it is exactly the shortest path tree from $x$ to all these vertices $u$. Also, clearly for different $x,y\in A_j$, these trees are disjoint (because each $u\in V$ has only one $j$'th pivot). Hence, for a fixed $j$,
\[|\bigcup_{x\in A_j}\bigcup_{u|p_j(u)=x}E(P_{u,x})|\leq n\;\;\;\Rightarrow\]
\[\Rightarrow\;\;\;|S_1|=|\bigcup_{u\in V}\bigcup_{j=0}^{F-1}E(P_{u,p_j(u)})|=|\bigcup_{j=0}^{F-1}\bigcup_{x\in A_j}\bigcup_{u|p_j(u)=x}E(P_{u,x})|\leq Fn~.\]

For $|S_2|$, notice that for every fixed $u\in V$ and $j\in[i(u),f(i(u))]$, we add the paths $\{P_{u,v}\}$, where $v\in B_j(u)$ and $d(u,v)\leq r_F$. That is, we add at most $|B_j(u)|$ paths, each of length at most $r_F$, which is at most $|B_j(u)|\cdot r_F$ edges for each $u,j$. So:
\[|S_2|=|\bigcup_{u\in V}\bigcup_{j=i(u)}^{f(i(u))}\bigcup_{\substack{v\in B_j(u) \\ d(u,v)\leq r_F}}E(P_{u,v})|\leq r_F\cdot\sum_{u\in V}\sum_{j=i(u)}^{f(i(u))}|B_j(u)|~.\]
In the proof of Lemma \ref{lemma:HkfESize} we saw that $\sum_{u\in V}\sum_{j=i(u)}^{f(i(u))}|B_j(u)|<4\max_i(f(i)-i+1)n^{1+\frac{1}{k}}$. We conclude that the overall size of $S(k,f,t)$ is, in expectation,
\[\E[|S(k,f,t)|]\leq\E[|S_1|]+\E[|S_2|]=O(Fn+r_F\cdot\max_i(f(i)-i+1)\cdot n^{1+\frac{1}{k}})~.\]

\end{proof}

The analysis of the multiplicative and additive stretch is very similar to the case of the hopset $H(k,f)$ in Section \ref{Stretch and Hopbound Analysis Method}. We use the same definition of a \textit{score} of a vertex $u\in V$:
\[score(u)=\max\{i>0\;|\;d(u,p_i(u))>r_i\text{ and }\forall_{j\in[f^{-1}(i-1),i-1]}\;d(u,p_j(u))\leq r_j\}~.\]

The following lemma and its proof are analogous to Lemma \ref{lemma:HJump}. We denote $S=S(k,f,t)$.

\begin{lemma} \label{lemma:S1Jump}
Suppose that $u\in V$ has $score(u)=i$, then for every $u'\in V$ such that\newline
$d(u,u')\leq \frac{r_i-r_{i-1}}{2}-r_{f^{-1}(i-1)}$:
\[d_S(u,u')\leq3d(u,u')+2(r_{i-1}+r_{f^{-1}(i-1)})~.\]
Moreover, if also $d(u,u')\geq \frac{2}{t}(r_{i-1}+r_{f^{-1}(i-1)})$, then:
\[d_S(u,u')\leq(t+3)d(u,u')~.\]
\end{lemma}

\begin{proof}
Let $u\in V$ be some vertex with $score(u)=i>0$ and suppose that $u'\in V$ satisfies $d(u,u')\leq \frac{r_i-r_{i-1}}{2}-r_{f^{-1}(i-1)}$. By our construction, the paths $P_{u,p_{f^{-1}(i-1)}(u)}$ and $P_{u',p_{i-1}(u')}$ are contained in $S$. We now bound the distance between $p_{f^{-1}(i-1)}(u)$ and $p_{i-1}(u')$. Recall that $score(u)=i$, and therefore for every $j\in[f^{-1}(i-1),i-1]$ we have $d(u,p_j(u))\leq r_j$. In particular,
\[d(u,p_{f^{-1}(i-1)}(u))\leq r_{f^{-1}(i-1)}~,\]
and also
\[d(u',p_{i-1}(u'))\leq d(u',p_{i-1}(u))\leq d(u',u)+d(u,p_{i-1}(u))\leq d(u',u)+r_{i-1}~.\]
Hence,
\begin{eqnarray*}
d(p_{f^{-1}(i-1)}(u),p_{i-1}(u'))&\leq&d(p_{f^{-1}(i-1)}(u),u)+d(u,u')+d(u',p_{i-1}(u'))\\
&\leq&r_{f^{-1}(i-1)}+d(u,u')+(d(u,u')+r_{i-1})\\
&=&2d(u,u')+r_{f^{-1}(i-1)}+r_{i-1}\\
&\leq&2(\frac{r_i-r_{i-1}}{2}-r_{f^{-1}(i-1)})+r_{f^{-1}(i-1)}+r_{i-1}\\
&=&r_i-r_{f^{-1}(i-1)}~.
\end{eqnarray*}

For convenience, we denote $u_0=p_{f^{-1}(i-1)}(u)$, and we also bound the distance $d(u_0,p_i(u_0))$:
\[r_i<d(u,p_i(u))\leq d(u,p_i(u_0))\leq d(u,u_0)+d(u_0,p_i(u_0))\leq r_{f^{-1}(i-1)}+d(u_0,p_i(u_0))\;\;\Rightarrow\]
\[\Rightarrow\;\;d(u_0,p_i(u_0))>r_i-r_{f^{-1}(i-1)}~.\]

Recall that $P_{u_0,v}$ is contained in $S(k,f,t)$, whenever $v\in B_j(u_0)$, $j\in[i(u_0),f(i(u_0))]$ and $d(u_0,v)\leq r_F$. Since $u_0$ is a $f^{-1}(i-1)$'th pivot, we know that $i(u_0)\geq f^{-1}(i-1)$, so using the fact that $f$ is non-decreasing:
\[i-1\leq f(f^{-1}(i-1))\leq f(i(u_0))~.\]

Also, since  $d(u_0,p_i(u_0))>r_i-r_{f^{-1}(i-1)}>0$, it cannot be that $i(u_0)\geq i$ (otherwise $u_0=p_i(u_0)$, so $d(u_0,p_i(u_0))=0$). Then we got that $i-1\in[i(u_0),f(i(u_0))]$.

Since $d(u_0,p_{i-1}(u'))\leq r_i-r_{f^{-1}(i-1)}<d(u_0,p_i(u_0))$, we know that $p_{i-1}(u')\in B_{i-1}(u_0)$, and also $d(u_0,p_{i-1}(u'))\leq r_i-r_{f^{-1}(i-1)}\leq r_F$. So by definition, $P_{u_0,p_{i-1}(u')}$ is contained is $S(k,f,t)$.

The length of the path $P_{u,u_0}\circ P_{u_0,p_{i-1}(u')}\circ P_{p_{i-1}(u'),u'}$, which we proved that is contained in $S$, is at most:
\[r_{f^{-1}(i-1)}+(2d(u,u')+r_{f^{-1}(i-1)}+r_{i-1})+(d(u,u')+r_{i-1})=\]
\[=3d(u,u')+2(r_{f^{-1}(i-1)}+r_{i-1})~.\]

If it happens to be that the above $u,u'$ satisfy also $r_{i-1}+r_{f^{-1}(i-1)}\leq\frac{t}{2}d(u,u')$ (or equivalently $d(u,u')\geq \frac{2}{t}(r_{i-1}+r_{f^{-1}(i-1)})$), then this path between $u,u'$ has weight at most
\[3d(u,u')+td(u,u')=(t+3)d(u,u')~,\]
i.e. a multiplicative stretch of $t+3$.
\end{proof}

Let $u,v\in V$ be a pair of vertices. We use Lemma \ref{lemma:S1Jump} for finding a path between $u,v$, that consists of these ``jumps" that the lemma provides. Denote the shortest path between them by $u=u_0,u_1,u_2,...,u_d=v$.

\begin{lemma} \label{lemma:S1Jump2}
Fix $0\leq h\leq d$. Suppose that $score(u_h)=i$. One of the following holds:
\begin{enumerate}
    \item $d(u_h,v)\leq\frac{2}{t}(r_{i-1}+r_{f^{-1}(i-1)})$,
    \item $\exists_{l>h}$ such that $\frac{2}{t}(r_{i-1}+r_{f^{-1}(i-1)})<d(u_h,u_l)\leq\frac{r_i-r_{i-1}}{2}-r_{f^{-1}(i-1)}$.
\end{enumerate}

In addition, if (1) holds, then
\[d_S(u_h,v)\leq3d(u_h,v)+4r_F~,\]
and if (2) holds, then
\[d_S(u_h,u_l)\leq(t+3)d(u_h,u_l)~.\]
\end{lemma}

\begin{proof}
We first prove that at least one of (1),(2) holds. Notice that by $\{r_i\}$'s definition:
\begin{eqnarray*}
\frac{r_i-r_{i-1}}{2}-r_{f^{-1}(i-1)}&=&\frac{((1+\frac{4}{t})r_{i-1}+(2+\frac{4}{t})r_{f^{-1}(i-1)}+2)-r_{i-1}}{2}-r_{f^{-1}(i-1)}\\
&=&\frac{2}{t}r_{i-1}+(1+\frac{2}{t})r_{f^{-1}(i-1)}+1-r_{f^{-1}(i-1)}\\
&=&\frac{2}{t}(r_{i-1}+r_{f^{-1}(i-1)})+1~.
\end{eqnarray*}

If (1) is true, then we're done. Suppose (1) doesn't hold, then
$d(u_h,v)>\frac{2}{t}(r_{i-1}+r_{f^{-1}(i-1)})$ and there must be some $h<l\leq d$ which is the first index such that $d(u_h,u_l)>\frac{2}{t}(r_{i-1}+r_{f^{-1}(i-1)})$. Since $G$ is unweighted, we actually know that
\[d(u_h,u_l)\leq\frac{2}{t}(r_{i-1}+r_{f^{-1}(i-1)})+1=\frac{r_i-r_{i-1}}{2}-r_{f^{-1}(i-1)}~,\]
thus, (2) holds.

If (1) holds, then
\[d(u_h,v)\leq\frac{2}{t}(r_{i-1}+r_{f^{-1}(i-1)})<\frac{2}{t}(r_{i-1}+r_{f^{-1}(i-1)})+1=\frac{r_i-r_{i-1}}{2}-r_{f^{-1}(i-1)}~.\]
Then by Lemma \ref{lemma:S1Jump}, we have that
\[d_S(u_h,v)\leq3d(u_h,v)+2(r_{i-1}+r_{f^{-1}(i-1)})\leq3d(u_h,v)+4r_F~,\]
when we used the fact that $\{r_i\}$ is a non-decreasing sequence.

If (2) holds, then directly from Lemma \ref{lemma:S1Jump}, we have that
\[d_S(u_h,u_l)\leq(t+3)d(u_h,u_l)~,\]
as desired.

\end{proof}

\begin{theorem} \label{thm:Skft}
Fix an integer $k\geq1$, a non-decreasing function $f:\mathbb{N}\rightarrow\mathbb{N}$ such that $\forall_i f(i)\geq i$ and a real parameter $t>0$. Define the values $\{\lambda_j\},F,\{r_i\}$ by $\lambda_j=1+\sum_{l<f^{-1}(j)}\lambda_l$ for every $j\geq0$, $F=\min\{F'\;|\;\sum_{l<F'}\lambda_l\geq k\}$, $r_0=1$ and $r_i=(1+\frac{4}{t})r_{i-1}+(2+\frac{4}{t})r_{f^{-1}(i-1)}+2$ for every $i>0$. Then, every undirected unweighted graph $G=(V,E)$ admits a spanner with
\begin{enumerate}
    \item size $O(Fn+r_F\cdot\max_i(f(i)-i+1)\cdot n^{1+1/k})$,
    \item multiplicative stretch $t+3$, and
    \item additive stretch $4r_F$.
\end{enumerate}

\end{theorem}

\begin{proof}
Given $G,k,f,t$, build $S(k,f,t)$ on $G$. By Lemma \ref{lemma:S1ESize}, this spanner has the desired size, in expectation.

Let $u,v\in V$ be a pair of vertices, and let $u=u_0,u_1,u_2,...,u_d=v$ be the shortest path between them. We use Lemma \ref{lemma:S1Jump2} to find a path between $u$ and $v$. Starting with $h=0$, find $l>h$ such that $\frac{2}{t}(r_{i-1}+r_{f^{-1}(i-1)})<d(u_h,u_l)\leq\frac{r_i-r_{i-1}}{2}-r_{f^{-1}(i-1)}$, where $i=score(u_h)$. If there is no such $l$, stop the process and denote $v'=u_h$. Otherwise, set $h\leftarrow l$, and continue in the same way.

This process creates a subsequence of $u_0,...,u_d$: $u=v_0,v_1,v_2,...,v_b=v'$, such that for every $q<b$ we have $d_S(v_q,v_{q+1})\leq(t+3)d(v_q,v_{q+1})$ (by Lemma \ref{lemma:S1Jump2}). For $v'=v_b$ we have $d(v',v)\leq\frac{2}{t}(r_{i-1}+r_{f^{-1}(i-1)})$, where $i=score(v')$. For this last segment, we get from Lemma \ref{lemma:S1Jump2} that $d_S(v',v)\leq3d(v',v)+4r_F$.

When summing over the entire path, we get:
\begin{eqnarray*}
d_S(u,v)&\leq&\sum_{q=0}^{b-1}(t+3)d(v_q,v_{q+1})+3d(v',v)+4r_F\\
&=&(t+3)d(u,v')+3d(v',v)+4r_F\\
&\leq&(t+3)d(u,v)+4r_F~.
\end{eqnarray*}

\end{proof}

Note that in Theorem \ref{thm:Skft} we used a different sequence $\{r_i\}$ than in Theorem \ref{thm:Hkf}. The next lemma shows that these two sequences are the same up to a factor of at most $2$.

\begin{lemma} \label{lemma:RRelations}
Let $\{r_i\}$ be the sequence from Theorem \ref{thm:Skft}. That is, given a fixed number $t>0$, $r_0=1$ and $r_i=\left(1+\frac{4}{t}\right)r_{i-1}+\left(2+\frac{4}{t}\right)r_{f^{-1}(i-1)}+2$ for every $i>0$. Let $\{r'_i\}$ be the sequence from Theorem \ref{thm:Hkf}. That is, $r'_0=1$ and $r'_i=\left(1+\frac{4}{t}\right)r'_{i-1}+\left(2+\frac{4}{t}\right)r'_{f^{-1}(i-1)}$. Then, for every $i\geq0$
\[r_i=\left(1+\frac{t}{t+4}\right)r'_i-\frac{t}{t+4}~.\]
\end{lemma}

\begin{proof}

For every $x,y\in\mathbb{R}$, denote
\[Z(x,y)=\left(1+\frac{4}{t}\right)x+\left(2+\frac{4}{t}\right)y~.\]
Then, for every $i>0$ we have $r_i=Z(r_{i-1},r_{f^{-1}(i-1)})+2$ and $r'_i=Z(r'_{i-1},r'_{f^{-1}(i-1)})$. Note also that $Z$ is linear, in the sense that for every $a,x,y,x',y'\in\mathbb{R}$, $Z(x+x',y+y')=Z(x,y)+Z(x',y')$ and $Z(ax,ay)=aZ(x,y)$. Lastly, notice that
\begin{equation} \label{eq:FixedPoint}
Z\left(-\frac{t}{t+4},-\frac{t}{t+4}\right)=-\frac{t}{t+4}\cdot Z(1,1)=-\frac{t}{t+4}\cdot\left(3+\frac{8}{t}\right)=-\frac{3t+8}{t+4}=-\frac{t}{t+4}-2~.
\end{equation}

We now prove by induction over $i\geq0$, that $r_i=\left(1+\frac{t}{t+4}\right)r'_i-\frac{t}{t+4}$. For $i=0$, $\left(1+\frac{t}{t+4}\right)r'_0-\frac{t}{t+4}=1+\frac{t}{t+4}-\frac{t}{t+4}=1=r_0$. For $i>0$, assuming correctness for all $j<i$, we get
\begin{eqnarray*}
r_i&=&Z(r_{i-1},r_{f^{-1}(i-1)})+2\\
&=&Z\left(\left(1+\frac{t}{t+4}\right)r'_{i-1}-\frac{t}{t+4},\left(1+\frac{t}{t+4}\right)r'_{f^{-1}(i-1)}-\frac{t}{t+4}\right)+2\\
&=&\left(1+\frac{t}{t+4}\right)Z\left(r'_{i-1},r'_{f^{-1}(i-1)}\right)+Z\left(-\frac{t}{t+4},-\frac{t}{t+4}\right)+2\\
&\stackrel{(\ref{eq:FixedPoint})}{=}&\left(1+\frac{t}{t+4}\right)r'_i-\frac{t}{t+4}-2+2\\
&=&\left(1+\frac{t}{t+4}\right)r'_i-\frac{t}{t+4}~,
\end{eqnarray*}
which completes the proof.

\end{proof}

\subsubsection{Applying our Construction and Analysis for Rounding Functions} \label{sec:RoundingFunctionsHopsetsSpanners}


Similarly to Section \ref{sec:RoundingFunctionsHopsets}, we apply Theorem \ref{thm:Skft} to the function
\[f(i)=\left\lfloor\frac{i}{c}\right\rfloor\cdot c+c-1~.\]

Applying Theorem \ref{thm:Skft} on this function $f$ yields the following corollary.

\begin{corollary}[Corollary from Theorem \ref{thm:Skft}] \label{cor:SkftForRounding}
Given two integer parameters $1<c\leq k$, every undirected unweighted graph $G=(V,E)$ admits a spanner with
\begin{enumerate}
    \item size $O(c\cdot k^{1+\frac{2}{\ln c}}\cdot n^{1+\frac{1}{k}})$,
    \item multiplicative stretch $4c+3$, and
    \item additive stretch $O(k^{1+\frac{2}{\ln c}})$.
\end{enumerate}
\end{corollary}

\begin{proof}

Given an undirected unweighted graph $G=(V,E)$, apply $S(k,f,t)$ on the graph $G$, with the function $f(i)=\lfloor\frac{i}{c}\rfloor\cdot c+c-1$ and $t=4c$. By Lemma \ref{lemma:S1ESize} and the computations in Appendix \ref{appendix:FProperties}, the expected size of $S(k,f,t)$ is $O(c\log_ck\cdot n+r_F\cdot c\cdot n^{1+\frac{1}{k}})$. We saw in the proof of Theorem \ref{thm:StrongHfk} that 
\[r'_F\leq(4c+5)^{\lceil\log_{c+1}(k+1)\rceil}r_0=(4c+5)^{\lceil\log_{c+1}(k+1)\rceil}=O(k^{1+\frac{2}{\ln c}})~,\]
for the sequence $\{r'_i\}$ that satisfies $r'_0=1$ and $r'_i=\left(1+\frac{4}{t}\right)r'_{i-1}+\left(2+\frac{4}{t}\right)r'_{f^{-1}(i-1)}$ for every $i>0$, where $t=4c$. By Lemma \ref{lemma:RRelations}, our sequence $\{r_i\}$, that satisfies $r_0=1$ and $r_i=\left(1+\frac{4}{t}\right)r_{i-1}+\left(2+\frac{4}{t}\right)r_{f^{-1}(i-1)}+2$, is at most twice as large. Thus, the expected size of $S(k,f,t)$ is
\[O(c\log_ck\cdot n+r_F\cdot c\cdot n^{1+\frac{1}{k}})=O(c\log_ck\cdot n+c\cdot k^{1+\frac{2}{\ln c}}\cdot n^{1+\frac{1}{k}})=O(c\cdot k^{1+\frac{2}{\ln c}}\cdot n^{1+\frac{1}{k}})~.\]
The last step is due to the fact that $c\log_ck\leq k$ for $2\le c\le k$.

The multiplicative and additive stretch of $S(k,f,t)$ are $4c+3$ and $4r_F=O(k^{1+\frac{2}{\ln c}})$, respectively.

\end{proof}

We now substitute several values in $f$ and in $c$ in Theorem \ref{thm:Skft} and in Corollary \ref{cor:SkftForRounding}. 

\vspace{5mm}
\noindent{\textbf{Multiplicative Stretch $\boldsymbol{(3+\epsilon)}$}.}
In Theorem \ref{thm:Skft}, choose the function $f(i)=i$ and $t=\epsilon$. Note that $f$ is a \textit{rounding function}, with $c=1$. By the results in Appendix \ref{appendix:FProperties}, $\lambda_j=2^j$, $F\leq\lceil\log_2(k+1)\rceil$, and
\[r_i\leq2\left[\left(1+\frac{4}{t}\right)\left(\frac{t}{2}+2\right)-\left(\frac{t}{2}+1\right)\right]^i=2\left(3+\frac{8}{t}\right)^i=2\left(3+\frac{8}{\epsilon}\right)^i~,\]
where the multiplication by $2$ is due to Lemma \ref{lemma:RRelations}.
The result is a spanner $S$, with multiplicative stretch $t+3=3+\epsilon$, additive stretch\newline 
$O\left(\left(3+\frac{8}{\epsilon}\right)^{\lceil\log_2(k+1)\rceil}\right)=O\left(k^{\log_2\left(3+\frac{8}{\epsilon}\right)}\right)$ and size 
\[O(n\log k+r_F\cdot n^{1+\frac{1}{k}})=O(n\log k+k^{\log_2\left(3+\frac{8}{\epsilon}\right)}\cdot n^{1+\frac{1}{k}})=O(k^{\log_2\left(3+\frac{8}{\epsilon}\right)}\cdot n^{1+\frac{1}{k}})~.\]

\vspace{5mm}
\noindent{\textbf{Constant Multiplicative Stretch}.}
In Corollary \ref{cor:SkftForRounding}, choose a constant $c$. The result is a spanner $S$, with multiplicative stretch $4c+3=O(1)$, additive stretch\newline 
$O\left(k^{1+\frac{2}{\ln c}}\right)$ and size $O(k^{1+\frac{2}{\ln c}}\cdot n^{1+\frac{1}{k}})$.

\vspace{5mm}
\noindent{\textbf{Multiplicative Stretch $\boldsymbol{O(k^\epsilon)}$}.}
In Corollary \ref{cor:SkftForRounding}, choose $c=\lceil k^\epsilon\rceil$ for some $0<\epsilon<1$. The result is a spanner $S$, with multiplicative stretch $4\lceil k^\epsilon\rceil+3=O(k^\epsilon)$, additive stretch
$O\left(k^{1+\frac{2}{\ln c}}\right)=O\left(k\cdot e^{\frac{2}{\epsilon}}\right)=O_\epsilon(k)$ and size $O(e^{\frac{2}{\epsilon}}\cdot k^{1+\epsilon}\cdot n^{1+\frac{1}{k}})=O_\epsilon(k^{1+\epsilon}\cdot n^{1+\frac{1}{k}})$.

\vspace{5mm}
\noindent{\textbf{Multiplicative Stretch $\boldsymbol{O(k)}$}.}
Using Corollary \ref{cor:SkftForRounding} with the function $f(i)=k-1$ (i.e., $c=k$), we obtain a spanner with multiplicative stretch $4k+3$, additive stretch $O(k^{1+\frac{2}{\ln k}})=O(k)$ and size $O(k\cdot k^{1+\frac{2}{\ln k}}\cdot n^{1+\frac{1}{k}})=O(k^2n^{1+\frac{1}{k}})$. In fact, a modified analysis, similar to the one in Section \ref{sec:RoundingFunctionsHopsets}, gives multiplicative stretch $2k-1$ and additive stretch $0$. This is specifically implied by an analogous lemma to Lemma \ref{lemma:InvolvedJump}, as for $c=k$, every vertex $u$ has $score(u)\leq c$. See Section \ref{sec:EmulatorSpecialCases} for a more elaborated discussion in the case of emulators. As mentioned in Lemma \ref{lemma:InvolvedJump}, its technique was used in the celebrated result by Thorup and Zwick \cite{TZ01}, that provides a spanner (and other structures) with multiplicative stretch $2k-1$ and additive stretch $0$. Their spanner, however, has size $O(kn^{1+\frac{1}{k}})$, which is smaller than our spanner by a factor of $k$. This size can be achieved in our case as well, by slightly changing our construction: instead of inserting only short bunch paths into $S(k,f,t)$ (with length at most $r_F$; see Definition \ref{def:Skft}), insert them into $S(k,f,t)$ regardless of their length. Since, for $c=k$, this construction is the exact same construction as of \cite{TZ01}, their size analysis holds, and we obtain size of $O(kn^{1+\frac{1}{k}})$.

\subsection{Simultaneous Spanners} \label{SSpanner}
In the previous spanner construction, we solved the problem of bounding the size of the spanner by limiting the length of the added paths. We now use a different solution, inspired by the \textit{half bunches} of \cite{Pet09}. In this approach, instead of connecting (by a shortest path) every vertex $u$ with all the vertices of its bunch $B_j(u)$, we connect $u$ only to a subset of it, called half bunch. The main argument regarding these half bunches, as introduces in \cite{Pet09}, is that the number of edges in the union of all the added shortest paths is relatively small. Specifically, it is roughly
\[\sum_{u\in V}\sum_{j=i(u)}^{f(i(u))}|B_j(u)|^3~.\]
Besides the power of $3$, this is exactly the number of bunch edges (hop-edges between a vertex and a member of its bunch) in the hopset $H(k,f)$. To handle the power of $3$, we must change the sampling probabilities, i.e., change the values $\lambda_i$, such that the size of each bunch $B_j(u)$ will be smaller. The modification of the sampling probabilities does not affect the stretch and the hopbound of $H(k,f)$, which translate into the multiplicative and additive stretch of the new spanner respectively. However, the usage of half bunches instead of bunches forces us to increase the elements of the sequence $\{r_i\}$, in order to maintain Lemma \ref{lemma:HJump}. This, in turn, makes the additive stretch larger, with respect to the hopbound of the original hopset $H(k,f)$, and to the additive stretch of the spanner $S(k,f,t)$.

We now fully and formally describe the construction that uses half bunches. We again start with a given undirected unweighted graph $G=(V,E)$, an integer $k\geq1$ and a non-decreasing function $f:\mathbb{N}\rightarrow\mathbb{N}$ such that $\forall_i\;f(i)\geq i$. Instead of defining $\{\lambda_j\}$ and $F$ the same as before, we now leave them as parameters to be chosen later, and define $\{A_i\}$, $\{p_i(u)\}$ and $\{B_i(u)\}$ on top of them, the same way as in Section \ref{Generalized TZ Construction} and in subSection \ref{NSSpanner}.

We now define the \textit{half bunch} of a vertex $u\in V$:
\begin{definition}
Given $u\in V$ and $j\in[0,F-1]$, the $j$'th {\em half bunch} of $u$ is the set
\[B_j^\frac{1}{2}(u)\coloneqq\{v\in A_j\;|\;d(u,v)<\frac{1}{2}d(u,p_{j+1}(u))\}~.\]
\end{definition}

Now, instead of adding to the spanner paths from a vertex $u$ to every vertex in its bunch, we only add the paths between $u$ and the vertices in its half bunch:
\begin{definition}
\[S(k,f,\{\lambda_j\},F)=\bigcup_{u\in V}\bigcup_{j=0}^{F-1}E(P_{u,p_j(u)})\cup\bigcup_{u\in V}\bigcup_{j=i(u)}^{f(i(u))}\bigcup_{v\in B_j^\frac{1}{2}(u)}E(P_{u,v})~,\]
where $E(P)$ is the set of edges of the path $P$.
\end{definition}

The choice to add only paths to vertices in the half bunch of a vertex, instead of the whole bunch of the vertex, reduce significantly the number of added edges to the spanner. We formalize and prove this claim:

\begin{lemma} \label{lemma:halfBunches}
Fix some $i\in[0,F-1]$ and $j\in[i,f(i)]$. Denote by $Q_{i,j}$ the sub-graph of $G$ that is induced by the union of all the shortest paths
\[\{P_{u,v}\;|\;u\in A_i\setminus A_{i+1}\text{ and }v\in B_j^\frac{1}{2}(u)\}~.\]
Then,
\[|E_{Q_{i,j}}|\leq n+4\sum_{u\in A_i}|B_j(u)|^3~,\]
where $E_{Q_{i,j}}$ is the set of edges of $Q_{i,j}$.
\end{lemma}

\begin{proof}
Let $\{P_1,P_2,P_3,...\}$ be an enumeration of the paths in
\[\{P_{u,v}\;|\;u\in A_i\setminus A_{i+1}\text{ and }v\in B_j^\frac{1}{2}(u)\}~.\]

Suppose we are building the graph $Q_{i,j}$ by adding these paths one by one. Denote by $d_{l,z}$ the degree of the vertex $z\in V$, after the paths $\{P_1,P_2,...,P_l\}$ were added (also we denote $d_{0,z}=0$ for every $z\in V$). We define a mapping from pairs of the form $(l,z)$ to tuples of the form $(u,v,x,y)\in V^4$:

Given $z\in V$ and $l>0$, if $d_{l,z}>d_{l-1,z}>0$ (which means that the addition of $P_l$ increased $z$'s degree, which was already positive before that), set $\varphi(l,z)=(u,v,x,y)$, where $P_l=P_{x,y}$ and $P_{u,v}$ is the first path such that after it was added, the degree of $z$ became positive. If the condition $d_{l,z}>d_{l-1,z}>0$ is not satisfied, we say that $(l,z)$ is {\em not mapped}.

Fix some $(u,v,x,y)\in Im(\varphi)$ and consider the set $\varphi^{-1}(u,v,x,y)=\{(l,z)\;|\;\varphi(l,z)=(u,v,x,y)\}$. For every $(l,z)$ in this set, $l$ is the unique index such that $P_l=P_{x,y}$. Also, $z\in P_{u,v}\cap P_{x,y}$, where the addition of $P_{x,y}$ increased $z$'s degree. Notice that since $P_{u,v}$ and $P_{x,y}$ are shortest paths, $P_{u,v}\cap P_{x,y}$ must be a shortest path, and the only vertices that $P_{x,y}$ could increase their degrees are the two ends of this shortest path (because for every internal vertex of $P_{u,v}\cap P_{x,y}$, its adjacent edges from $P_{u,v}$ and its adjacent edges from $P_{x,y}$ are the same edges). Therefore, each pair in $\varphi^{-1}(u,v,x,y)$ has the same $l$ and one of two possible values of $z$, i.e., $|\varphi^{-1}(u,v,x,y)|\leq2$.

So, $\varphi$ is actually a $2$-to-$1$ mapping, and we get:
\[|\{(l,z)\;|\;(l,z)\text{ is mapped}\}|\leq2Im(\varphi)~.\]

Let $(u,v,x,y)\in Im(\varphi)$. We know that there is $z\in V$ such that $z$ is contained both in $P_{u,v}$ and $P_{x,y}$. Recall that both paths were added to the spanner because $v\in B_j^\frac{1}{2}(u)$ and $y\in B_j^\frac{1}{2}(x)$. Thus, if $d(u,v)\leq d(x,y)$, we get:
\[d(x,u)\leq d(x,z)+d(z,u)\leq d(x,y)+d(u,v)\leq2d(x,y)<2\cdot\frac{1}{2}d(x,p_{j+1}(x))=d(x,p_{j+1}(x))~.\]
Therefore $u\in B_j(x)$. Similarly, $v\in B_j(x)$ and also of course $y\in B_j^\frac{1}{2}(x)\subseteq B_j(x)$. If we instead have $d(x,y)\leq d(u,v)$, then symmetrically we get $x,y,v\in B_j(u)$.

Therefore, we can estimate the size of $Im(\varphi)$:
\[Im(\varphi)\subseteq\{(u,v,x,y)\;|\;u\in A_i\text{, }v,x,y\in B_j(u)\}\cup\{(u,v,x,y)\;|\;x\in A_i\text{, }y,u,v\in B_j(x)\}\;\Rightarrow\]
\[\Rightarrow\;\;\;|Im(\varphi)|\leq2\sum_{u\in A_i}|B_j(u)|^3~.\]

The last thing to notice, is that for a given $z\in V$, each added path can increase its degree by at most $2$. The number of times that $z$'s degree is increased, except for the first time, is equal to the number of pairs $(l,z)$ that are mapped. So:
\begin{eqnarray*}
|E_{Q_{i,j}}|&=&\frac{1}{2}\sum_{z\in V}deg_{Q_{i,j}}(z)\\
&\leq&\frac{1}{2}\sum_{z\in V}(2+2|\{l\;|\;(l,z)\text{ is mapped}\}|)\\
&=&\sum_{z\in V}1+\sum_{z\in V}|\{l\;|\;(l,z)\text{ is mapped}\}|\\
&=&n+|\{(l,z)\;|\;(l,z)\text{ is mapped}\}|\\
&\leq&n+2|Im(\varphi)|\\
&\leq&n+4\sum_{u\in A_i}|B_j(u)|^3~.
\end{eqnarray*}

\end{proof}

The previous lemma let us calculate the expected size of $S(k,f,\{\lambda_j\},F)$, and choose the parameters $\{\lambda_j\},F$ such that we get the desired spanner size. The following lemma is analogous to Lemma \ref{lemma:HkfESize}.

\begin{lemma} \label{lemma:S2ESize}
Suppose that the parameters $k,f,\{\lambda_j\},F$ satisfy
\begin{enumerate}
    \item $\sum_{l<F}\lambda_l\geq k$,
    \item $\forall_j\;\lambda_j\leq\frac{1}{3}(1+\sum_{l<f^{-1}(j)}\lambda_l)$.
\end{enumerate}
Then, $\E[|S(k,f,\{\lambda_j\},F)|]=O\left(\max_i(f(i)-i+1)\cdot(Fn+n^{1+\frac{1}{k}})\right)$.
\end{lemma}

\begin{proof}
As in the proof of Lemma \ref{lemma:HkfESize}, we can see that $\E[|A_i|]=\frac{1}{2^i}n^{1-\frac{1}{k}\sum_{j<i}\lambda_j}$ for every $0\leq i\leq F-1$.

Similarly to the proof of Lemma \ref{lemma:S1ESize}, we denote $S(k,f,\{\lambda_j\},F)=S_1\cup S_2$, where
\[S_1=\bigcup_{u\in V}\bigcup_{j=0}^{F-1}E(P_{u,p_j(u)})~,\;\;\;\;S_2=\bigcup_{u\in V}\bigcup_{j=i(u)}^{f(i(u))}\bigcup_{v\in B_j^\frac{1}{2}(u)}E(P_{u,v})~,\]
and with the same proof as of Lemma \ref{lemma:S1ESize} we can prove that $|S_1|\leq Fn$.

Now, notice that $S_2$ is exactly the union of $E_{Q_{i,j}}$ for every $i\in[0,F-1]$ and $j\in[i,f(i)]$, where the notation $E_{Q_{i,j}}$ is from Lemma \ref{lemma:halfBunches}. Therefore, from Lemma \ref{lemma:halfBunches}:
\[\E[|S_2|]\leq\sum_{i<F}\sum_{j=i}^{f(i)}(n+4\sum_{u\in A_i}\E[|B_j(u)|^3])\leq Fn\cdot\max_i(f(i)-i+1)+4\sum_{i=0}^{F-1}\sum_{u\in A_i}\sum_{j=i}^{f(i)}\E[|B_j(u)|^3]~.\]

Recall that in the proof of Lemma \ref{lemma:HkfESize} we saw that $|B_j(u)|+1$ is bounded by a geometric random variable with $p=\frac{1}{2}n^{-\frac{\lambda_j}{k}}$. Generally, for a geometric random variable $X$ with parameter $p$, it can be shown that:
\[\E[X^3]=\frac{6(1-p)}{p^3}+\frac{1}{p}\leq\frac{6}{p^3}+\frac{1}{p^3}=\frac{7}{p^3}~.\]

In our case:
\[\E[|B_j(u)|^3]\leq\E[(|B_j(u)|+1)^3]\leq56n^{\frac{3\lambda_j}{k}}~,\]

so we get
\begin{eqnarray*}
\E[|S_2|]&\leq&Fn\cdot\max_i(f(i)-i+1)+4\sum_{i=0}^{F-1}\frac{1}{2^i}n^{1-\frac{1}{k}\sum_{l<i}\lambda_l}\sum_{j=i}^{f(i)}56n^\frac{3\lambda_j}{k}\\
&=&Fn\cdot\max_i(f(i)-i+1)n+224\sum_{i=0}^{F-1}\sum_{j=i}^{f(i)}\frac{1}{2^i}n^{1+\frac{1}{k}(3\lambda_j-\sum_{l<i}\lambda_l)}~.
\end{eqnarray*}

Recall the assumption $\lambda_j\leq\frac{1}{3}(1+\sum_{l<f^{-1}(j)}\lambda_l)$, for every $j$. For every $i\geq f^{-1}(j)$, we have 
\[\lambda_j\leq\frac{1}{3}\left(1+\sum_{l<f^{-1}(j)}\lambda_l\right)\leq\frac{1}{3}\left(1+\sum_{l<i}\lambda_l\right)~.\] 
But note that $f^{-1}(j)\leq i$ if and only of $f(i)\geq j$. We conclude that for every $i,j$ such that $j\leq f(i)$, we have $3\lambda_j-\sum_{l<i}\lambda_l\leq 1$. Thus, the expected size of $S_2$ is at most
\begin{eqnarray*}
&&Fn\cdot\max_i(f(i)-i+1)+224\sum_{i=0}^{F-1}\sum_{j=i}^{f(i)}\frac{1}{2^i}n^{1+\frac{1}{k}(3\lambda_j-\sum_{l<i}\lambda_l)}\\
&\leq&Fn\cdot\max_i(f(i)-i+1)+224\sum_{i=0}^{F-1}\frac{f(i)-i+1}{2^i}n^{1+\frac{1}{k}}\\
&\leq&Fn\cdot\max_i(f(i)-i+1)+224\max_i(f(i)-i+1)n^{1+\frac{1}{k}}\cdot\sum_{i=0}^{F-1}\frac{1}{2^i}\\
&=&O\left(\max_i(f(i)-i+1)\cdot(Fn+n^{1+\frac{1}{k}})\right)~.
\end{eqnarray*}

We conclude that the total expected size of $S(k,f,\{\lambda_j\},F)$ is
\[\mathbb{E}[|S_1|]+\mathbb{E}[|S_2|]=O\left(\max_i(f(i)-i+1)\cdot(Fn+n^{1+\frac{1}{k}})\right)~.\]

\end{proof}

The following definition chooses the largest $\{\lambda_j\}$ possible and the smallest $F$ possible:
\begin{definition}
Given $k,f$ we define $S(k,f)\coloneqq S(k,f,\{\lambda_j\},F)$, where
\begin{enumerate}
    \item $F\coloneqq\min\{F'\;|\;\sum_{l<F'}\lambda_l\geq k\}$,
    \item $\forall_j\;\lambda_j\coloneqq\frac{1}{3}(1+\sum_{l<f^{-1}(j)}\lambda_l)$.
\end{enumerate}
\end{definition}

We now analyse the stretch of $S(k,f)$, using a similar analysis to that of $H(k,f)$ in Section \ref{Stretch and Hopbound Analysis Method}.

Let $t>0$ be some real number, and define the following sequence:
\[r_0\coloneqq1~,\]
\[r_{i+1}\coloneqq(2+\frac{8}{t})r_i+(3+\frac{8}{t})r_{f^{-1}(i)}+4~.\]

We use the same definition of \textit{score} as before:
\[score(u)=\max\{i>0\;|\;d(u,p_i(u))>r_i\text{ and }\forall_{j\in[f^{-1}(i-1),i-1]}\;d(u,p_j(u))\leq r_j\}~.\]

The following two lemmas are analogous to lemmas \ref{lemma:S1Jump} and \ref{lemma:S1Jump2}, so instead of proving them again, we only specify the necessary changes in the proofs. We denote $S=S(k,f)$.

\begin{lemma} \label{lemma:S2Jump}
Fix $u\in V$ with $score(u)=i$. For every $u'\in V$ such that $d(u,u')\leq\frac{r_i-2r_{i-1}-3r_{f^{-1}(i-1)}}{4}$:
\[d_S(u,u')\leq3d(u,u')+2(r_{i-1}+r_{f^{-1}(i-1)})~.\]
Moreover, if also $d(u,u')\geq\frac{2}{t}(r_{i-1}+r_{f^{-1}(i-1)})$, then:
\[d_S(u,u')\leq(t+3)d(u,u')~.\]
\end{lemma}

\begin{proof}
As in Lemma \ref{lemma:S1Jump}, $P_{u,p_{f^{-1}(i-1)}(u)}$ and $P_{u',p_{i-1}(u')}$ are contained in $S$, and:
\[d(p_{f^{-1}(i-1)}(u),p_{i-1}(u'))\leq2d(u,u')+r_{f^{-1}(i-1)}+r_{i-1}~,\]
\[d(u_0,p_i(u_0))>r_i-r_{f^{-1}(i-1)}~,\]
where $u_0=p_{f^{-1}(i-1)}(u)$.

Then, by $S$'s definition, for the path $P_{u_0,p_{i-1}(u')}$ to be contained in $S$, it is enough that
\[2d(u,u')+r_{f^{-1}(i-1)}+r_{i-1}\leq\frac{1}{2}(r_i-r_{f^{-1}(i-1)})~,\]
i.e.
\[d(u,u')\leq\frac{r_i-2r_{i-1}-3r_{f^{-1}(i-1)}}{4}~,\]
which is given.

Now, the weight of the path $P_{u,u_0}\circ P_{u_0,p_{i-1}(u')}\circ P_{p_{i-1}(u'),u'}$, which we proved that is contained in $S$, is at most:
\[3d(u,u')+2(r_{f^{-1}(i-1)}+r_{i-1})~,\]
and if $d(u,u')\geq \frac{2}{t}(r_{i-1}+r_{f^{-1}(i-1)})$, then this weight is at most
\[3d(u,u')+td(u,u')=(t+3)d(u,u')~.\]
\end{proof}

Denote the shortest path between $u,v$ as $u=u_0,u_1,u_2,...,u_d=v$.

\begin{lemma} \label{lemma:S2Jump2}
Fix $0\leq h\leq d$. Suppose that $score(u_h)=i$. One of the following holds:
\begin{enumerate}
    \item $d(u_h,v)\leq\frac{2}{t}(r_{i-1}+r_{f^{-1}(i-1)})$,
    \item $\exists_{l>h}$ such that $\frac{2}{t}(r_{i-1}+r_{f^{-1}(i-1)})<d(u_h,u_l)\leq\frac{r_i-2r_{i-1}-3r_{f^{-1}(i-1)}}{4}$.
\end{enumerate}

In addition, if (1) holds, then
\[d_S(u_h,v)\leq3d(u_h,v)+4r_F~,\]
and if (2) holds, then
\[d_S(u_h,u_l)\leq(t+3)d(u_h,u_l)~.\]
\end{lemma}

\begin{proof}
Notice that by $\{r_i\}$'s definition:
\begin{eqnarray*}
\frac{r_i-2r_{i-1}-3r_{f^{-1}(i-1)}}{4}&=&\frac{((2+\frac{8}{t})r_{i-1}+(3+\frac{8}{t})r_{f^{-1}(i-1)}+4)-2r_{i-1}-3r_{f^{-1}(i-1)}}{4}\\
&=&\frac{2}{t}r_{i-1}+\frac{2}{t}r_{f^{-1}(i-1)}+1\\
&=&\frac{2}{t}(r_{i-1}+r_{f^{-1}(i-1)})+1~.
\end{eqnarray*}

Using this identity, the rest of the proof is identical to that of Lemma \ref{lemma:S1Jump2}.

\end{proof}

\begin{theorem} \label{thm:Skf}
Fix an integer $k\geq1$, a non-decreasing function $f:\mathbb{N}\rightarrow\mathbb{N}$ such that $\forall_i f(i)\geq i$ and a real parameter $t>0$. Define the values $\{\lambda_j\},F,\{r_i\}$ by $\lambda_j=\frac{1}{3}\left(1+\sum_{l<f^{-1}(j)}\lambda_l\right)$ for every $j\geq0$, $F=\min\{F'\;|\;\sum_{l<F'}\lambda_l\geq k\}$, $r_0=1$ and $r_i=\left(2+\frac{8}{t}\right)r_{i-1}+\left(3+\frac{8}{t}\right)r_{f^{-1}(i-1)}+4$ for every $i>0$. Then, every undirected unweighted graph $G=(V,E)$ admits a spanner with
\begin{enumerate}
    \item Size $O\left(\max_i(f(i)-i+1)\cdot(Fn+n^{1+\frac{1}{k}})\right)$,
    \item Multiplicative stretch $t+3$,
    \item Additive stretch $4r_F$.
\end{enumerate}

\end{theorem}

\begin{proof}
Given $G,k,f$ build $S(k,f)$ on $G$. By Lemma \ref{lemma:S2ESize}, this spanner has the desired size, in expectation. 

Let $u,v\in V$ be a pair of vertices, and let $u=u_0,u_1,u_2,...,u_d=v$ be the shortest path between them. We use Lemma \ref{lemma:S2Jump2} to find a path between $u$ and $v$. Starting with $h=0$, find $l>0$ such that $\frac{2}{t}(r_{i-1}+r_{f^{-1}(i-1)})<d(u_h,u_l)\leq\frac{r_i-2r_{i-1}-3r_{f^{-1}(i-1)}}{4}$, where $i=score(u_h)$. If there is no such $l$, stop the process and denote $v'=u_h$. Otherwise, set $h\leftarrow l$, and continue in the same way.

This process creates a subsequence of $u_0,...,u_d$: $u=v_0,v_1,v_2,...,v_b=v'$, such that for every $q<b$ we have $d_S(v_q,v_{q+1})\leq(t+3)d(v_q,v_{q+1})$ (by Lemma \ref{lemma:S2Jump2}). For $v'=v_b$ we have $d(v',v)\leq\frac{2}{t}(r_{i-1}+r_{f^{-1}(i-1)})$, where $i=score(v')$. For this last segment, we get from Lemma \ref{lemma:S2Jump2} that $d_S(v',v)\leq3d(v',v)+4r_F$.

When summing over the entire path, we get:
\begin{eqnarray*}
d_S(u,v)&\leq&\sum_{q=0}^{b-1}(t+3)d(v_q,v_{q+1})+3d(v',v)+4r_F\\
&=&(t+3)d(u,v')+3d(v',v)+4r_F\\
&\leq&(t+3)d(u,v)+4r_F~.
\end{eqnarray*}

\end{proof}

\vspace{5mm}
\noindent{\textbf{Multiplicative Stretch $\boldsymbol{3+\epsilon}$}.}
For the function $f(i)=i$, we get in Theorem \ref{thm:Skf}, for every $j\geq0$,
\[\lambda_j=\frac{1}{3}(1+\sum_{l<j}\lambda_l)=\frac{1}{3}(1+\sum_{l<j-1}\lambda_l)+\frac{1}{3}\lambda_{j-1}=\lambda_{j-1}+\frac{1}{3}\lambda_{j-1}=\frac{4}{3}\lambda_{j-1}~.\]
Thus, $\lambda_j=\frac{1}{3}\cdot\left(\frac{4}{3}\right)^j$.

As a result, we also get $F=\lceil\log_{\frac{4}{3}}(k+1)\rceil$. Regarding the sequence $\{r_i\}$, we get for every $i>0$
\[r_i=\left(2+\frac{8}{t}\right)r_{i-1}+\left(3+\frac{8}{t}\right)r_{f^{-1}(i-1)}+4=\left(5+\frac{16}{t}\right)r_{i-1}+4~.\]
By induction, it is easy to prove that for every $i\geq0$,
\[r_i=\left(1+\frac{t}{t+4}\right)\left(5+\frac{16}{t}\right)^i-\frac{t}{t+4}~,\]
and thus $r_i\leq2\left(5+\frac{16}{t}\right)^i$, for every $i\geq0$. In particular,
\[r_F\leq2\left(5+\frac{16}{t}\right)^{\lceil\log_{\frac{4}{3}}(k+1)\rceil}=O(k^{\log_{4/3}(5+\frac{16}{t})})~.\]

In Theorem \ref{thm:Skf}, choose $t=\epsilon$. The resulting spanner has multiplicative stretch $3+\epsilon$, additive stretch $O(k^{\log_{4/3}(5+\frac{16}{\epsilon})})$ and size $O(n\log k+n^{1+\frac{1}{k}})$. Compared to the spanner from Section \ref{sec:RoundingFunctionsHopsetsSpanners}, that had stretch $3+\epsilon$, this spanner has a worse exponent for $k$ in the additive stretch (the additive stretch there was $O\left(k^{\log_2\left(3+\frac{8}{\epsilon}\right)}\right)$). However, the size of this spanner is significantly better than the spanner from Section \ref{sec:RoundingFunctionsHopsetsSpanners}, which was $O(k^{\log_2\left(3+\frac{8}{\epsilon}\right)}\cdot n^{1+\frac{1}{k}})$.


\subsection{Emulators}

Recall Lemma \ref{lemma:HJump2} from Section \ref{Stretch and Hopbound Analysis Method}. Given $u,v\in V$ and a shortest path $P$ between them, we assigned a \textit{score} for every vertex on $P$. Then, the score $i$ of such vertex $x$ implied an interval 
\[\mathcal{I}=\left[\frac{2}{t}(r_{i-1}+r_{f^{-1}(i-1)}),\frac{r_i-r_{i-1}}{2}-r_{f^{-1}(i-1)}\right]~,\] 
such that if $y$ is a vertex on $P$ with distance $d\in\mathcal{I}$ from $x$, then there is a $3$-hops detour in $H(k,f)$ between $x,y$ with weight at most $(t+3)d$. In case there is no vertex $y$ on $P$ with distance $d\in\mathcal{I}$ from $x$, we used an additional edge of $G$, besides the three edges of $H(k,f)$, to reach from $x$ to the first vertex $y$ with $d_G(x,y)>\frac{r_i-r_{i-1}}{2}-r_{f^{-1}(i-1)}$. These detours, with three or four edges, enabled us to find a low-hops path in $G\cup H(k,f)$ between $u,v$.

Emulator is a close structure to hopset, in that they both have a low-stretch path between every $u,v\in V$, while using edges that do not necessarily exist in $G$. Emulators, however, must not use edges of $G$ like hopsets do, but on the other hand, do not have a restriction on the number of hops that these low-stretch paths have. That means that we could use $H(k,f)$ as an emulator, if we knew that the detours mentioned above all have three hops in $H(k,f)$ (as opposed to $G\cup H(k,f)$). Recall that this happens when the starting vertex $x$ of the detour has a vertex $y$ further on $P$, with distance $d\in\mathcal{I}$. To ensure that this condition is satisfied, in an \textit{unweighted} graph, we modify the sequence $\{r_i\}$, so that
\[\frac{r_i-r_{i-1}}{2}-r_{f^{-1}(i-1)}\geq\frac{2}{t}(r_{i-1}+r_{f^{-1}(i-1)})+1~,\]
for every $i>0$, similarly to our approach in Sections \ref{NSSpanner} and \ref{SSpanner}. This increases the elements of this sequence by a factor of at most $2$, by Lemma \ref{lemma:RRelations} (with respect to the original sequence from Theorem \ref{thm:Hkf}). The length of the last detour, which does not necessarily have a stretch of $t+3$ with respect to the part of $P$ it skips, is considered as the additive stretch of the emulator, as in Sections \ref{NSSpanner} and \ref{SSpanner}.

As a result, we get the following theorem, regarding a unified construction of emulators. The proof is deferred to Appendix \ref{EmulatorsAppendix}.

\begin{theorem} \label{thm:Ekf}
Let $k>0$ be an integer, let $f:\mathbb{N}\rightarrow\mathbb{N}$ be a monotone non-decreasing function such that $\forall_i f(i)\geq i$, and let $\{\lambda_j\},F,\{r_i\}$ be parameters such that $\forall_j\lambda_j\leq1+\sum_{l<f^{-1}(j)}\lambda_l$, $\sum_{j<F}\lambda_j\geq k$, $r_0=1$ and $r_{i+1}=(1+\frac{4}{t})r_i+(2+\frac{4}{t})r_{f^{-1}(i)}+2$ for every $i\geq0$. Then, every undirected unweighted graph $G=(V,E)$ admits a $(t+3,4r_F)$-emulator with size $O(Fn+\max_i(f(i)-i+1)\cdot n^{1+1/k})$, simultaneously for every $t>0$.
\end{theorem}

The following corollary is obtained from Theorem \ref{thm:Ekf}, when using the function $f(i)=\left\lfloor\frac{i}{c}\right\rfloor\cdot c+c-1$.

\begin{corollary}[Corollary from Theorem \ref{thm:Ekf}] \label{cor:EfkForRounding}
Given two integer parameters $1<c,k$, every undirected unweighted graph $G=(V,E)$ admits a $\left(4c+3,O(k^{1+\frac{2}{\ln c}})\right)$-emulator with size $O(c\log_ck+c\cdot n^{1+\frac{1}{k}})$.
\end{corollary}

\begin{proof}

Given an undirected unweighted graph $G=(V,E)$, apply the emulator from Theorem \ref{thm:Ekf} on the graph $G$, with the function $f(i)=\lfloor\frac{i}{c}\rfloor\cdot c+c-1$ and $t=4c$. In Appendix \ref{appendix:FProperties}, we prove that $F=O(c\log_ck)$ and that $r'_F=O(k^{1+\frac{2}{\ln c}})$, for the sequence $\{r'_i\}$ that satisfies $r'_0=1$ and $r'_i=\left(1+\frac{4}{t}\right)r'_{i-1}+\left(2+\frac{4}{t}\right)r'_{f^{-1}(i-1)}$ for every $i>0$ (see the proof of Theorem \ref{thm:StrongHfk} for more details). By Lemma \ref{lemma:RRelations}, our sequence $\{r_i\}$, that satisfies $r_0=1$ and $r_i=\left(1+\frac{4}{t}\right)r_{i-1}+\left(2+\frac{4}{t}\right)r_{f^{-1}(i-1)}+2$, is at most twice as large as $\{r'_i\}$.

Thus, by Theorem \ref{thm:Ekf}, $G$ admits an emulator with multiplicative stretch $t+3=4c+3$, additive stretch $4r_F=O(k^{1+\frac{2}{\ln c}})$, and size 
\[O(Fn+\max_i(f(i)-i+1)\cdot n^{1+1/k})=O(c\log_ck\cdot n+c\cdot n^{1+\frac{1}{k}})~.\]

\end{proof}

\subsubsection{Special Cases} \label{sec:EmulatorSpecialCases}

In Section \ref{sec:HopsetSpecialCases} we substituted functions $f$ and values $c$, to achieve a variety of results for hopsets. Using Lemma \ref{lemma:RRelations}, we may conclude the same results for every regime of the multiplicative stretch as in Section \ref{sec:HopsetSpecialCases}, except for the additive stretch, which is at most twice the hopbound of the respective hopset.

\vspace{5mm}
\noindent{\textbf{Stretch $\boldsymbol{(1+\epsilon)}$}.}
Choose the function $f(i)=i$. The construction of the emulator from Theorem \ref{thm:Ekf}, which is identical to the hopset $H(k,f)$ from Section \ref{Generalized TZ Construction}, is essentially the same construction as in \cite{TZ06}. In \cite{TZ06}, and later in \cite{EGN19} and other works, it was proved that this construction serves as a $(1+\epsilon,\beta)$-emulator with size $O(n\log k+n^{1+\frac{1}{k}})$, where $\beta=O\left(\frac{\log k}{\epsilon}\right)^{\log k}$.

\vspace{5mm}
\noindent{\textbf{Stretch $\boldsymbol{(3+\epsilon)}$}.}
By Theorem \ref{thm:Ekf} with the function $f(i)=i$ and with $t=\epsilon$, every undirected unweighted graph has an emulator, that simultaneously for every $\epsilon>0$ has multiplicative stretch $3+\epsilon$, additive stretch $O\left(\frac{1}{\epsilon}\cdot k^{\log_2\left(3+\frac{8}{\epsilon}\right)}\right)$ and size $O(n\log k+n^{1+\frac{1}{k}})$.

\vspace{5mm}
\noindent{\textbf{Constant Stretch}.}
For a constant $c$, Corollary \ref{cor:EfkForRounding} gives an $\left(4c+3,O\left(k^{1+\frac{2}{\ln c}}\right)\right)$-emulator with size $O(n\log k+n^{1+\frac{1}{k}})$.

\vspace{5mm}
\noindent{\textbf{Stretch $\boldsymbol{O(k^\epsilon)}$}.}
In Corollary \ref{cor:EfkForRounding}, choose $c=\lceil k^\epsilon\rceil$ for some $0<\epsilon<1$. The result is an emulator with multiplicative stretch $4\lceil k^\epsilon\rceil+3=O(k^\epsilon)$, additive stretch
$O(e^{\frac{2}{\epsilon}}k)=O_\epsilon(k)$ and size $O(\epsilon^{-1}k^\epsilon\cdot n+k^\epsilon\cdot n^{1+\frac{1}{k}})=O_\epsilon(k^\epsilon\cdot n^{1+\frac{1}{k}})$.

\vspace{5mm}
\noindent{\textbf{Stretch $\boldsymbol{2k-1}$}.}
Choosing the constant function $f(i)=k-1$, the set $H(k,f)$ is identical to the emulator that is implicit in \cite{TZ01}. Thus, it has multiplicative stretch $2k-1$, additive stretch $0$ and size $O(kn^{1+\frac{1}{k}})$. We note that while these results cannot be obtained directly from Theorem \ref{thm:Ekf}, they can be obtained using more involved analysis, that is analogous to the proof of Theorem \ref{thm:StrongHfk}. We elaborate below on this approach.

In Section \ref{sec:RoundingFunctionsHopsets}, a modified analysis was given for the stretch and the hopbound of the hopset $H(k,f)$, for the function $f(i)=\left\lfloor\frac{i}{c}\right\rfloor\cdot c+c-1$. In particular, given $u,v\in V$ and a vertex $x$ on the $u-v$ shortest path $P$, instead of a $3$-hops detour in $H(k,f)$, from $x$ to some vertex $y$ on $P$ that has a certain distance from $x$, an alternative detour was given in the case that $score(x)\leq c$. The stretch and the number of hops of this alternative detour were better (stretch $2c-1$ instead of $4c+3$, and $2$ hops instead of $3$ or $4$), but the most crucial property of these detours was that they skip a large part of the path $P$, and thus reduces the hopbound of the hopset $H(k,f)$. In emulators, the number of hops is less significant. However, we still use this alternative detours, due to their additional handy property: for every $x$ on $P$ with $score(x)\leq c$, and for \textit{every} $y$ further on $P$, with $d_G(x,y)\leq\frac{r_c}{c}$, there is a detour from $x$ to $y$ in $H(k,f)$, with stretch $2c-1$. Note that here we do not require $d_G(x,y)$ to have some lower bound as in Lemma \ref{lemma:HJump}. This fact is then utilized in the detours that we use: if the starting vertex of every such detour has $score(x)\leq c$, then the total multiplicative stretch is $2c-1$, and the the total additive stretch is $0$.

Notice that when choosing $c=k$, i.e., when using the constant function $f(i)=k-1$, the score of every vertex is at most $c=k$, because the pivots $p_k(u)$ do not exist. Hence, we get a multiplicative stretch of $2k-1$, additive stretch of $0$, and size 
\[O(c\log_ck\cdot n+c\cdot n^{1+\frac{1}{k}})=O(k\cdot1\cdot n+k\cdot n^{1+\frac{1}{k}})=O(k\cdot n^{1+\frac{1}{k}})~.\]
This matches the emulator that is implicit in \cite{TZ01}.

\section{A Lower Bound for the Unified Algorithm} \label{A Lower Bound for the TZ Hopset}

In the previous sections, we showed that the general hopset $H(k,f)$ can essentially achieve all of the state-of-the-art results for hopsets. We now show that unfortunately, $H(k,f)$ cannot achieve significantly better results. In this section, all of our logarithms are of base $2$.

\subsection{Lower Bound Graph Construction}

We demonstrate our lower bound for the hopset $H(k,f)$ on an $n$-vertex graph $G(k,f,\alpha,n)$. Here, $k>0$ is an integer parameter, $f$ is a non-decreasing function $f:\mathbb{N}\rightarrow\mathbb{N}$ such that $f(i)\geq i$ for every $i$, $\alpha>0$ is some real parameter, and the construction is valid for infinitely many integers $n>0$.

For the construction to be well defined, we need to choose $n$ such that $n^\frac{1}{2k}\in\mathbb{N}$ and also $\log n\in\mathbb{N}$. In addition, we assume that $k\leq\frac{\log n}{4\log\log n}$. 
We define a graph $G(k,f,\alpha,n)$ as follows.

First, let $\{\lambda_j\}$ and $F$ be defined the same way as in Definition \ref{def:Hkf}:
\begin{enumerate}
    \item $\lambda_j=1+\sum_{l<f^{-1}(j)}\lambda_l$ (in particular, $\lambda_0=1$),
    \item $F=\min\{F'\;|\;\sum_{l<F'}\lambda_l\geq k\}$.
\end{enumerate}

Also, let $\{r_j\}$ be defined similarly to its definition in Theorem \ref{thm:Hkf}: $r_0=1$ and
\[\forall_{j\geq0}\;r_{j+1}=(1+\frac{1}{\alpha})r_j+(2+\frac{1}{\alpha})r_{f^{-1}(j)}~,\]
where here we chose $t=4\alpha$.

We define the {\em tower} $T(k,f,\alpha,n)$ to be the following weighted graph. $T(k,f,\alpha,n)$ consists of $F$ layers, indexed by $0$ to $F-1$. The $0$'th layer consists of only one vertex. For $0<i<F-1$, the $i$'th layer contains $\log n\cdot n^{\frac{1}{k}\sum_{l<i}\lambda_l}$ vertices, while the number of vertices in the $(F-1)$'th layer is
\[n^{1-\frac{1}{2k}}-1-\log n\sum_{i=1}^{F-2}n^{\frac{1}{k}\sum_{l<i}\lambda_l}~.\]
That is, the $(F-1)$'th layer completes the number of vertices in $T(k,f,\alpha,n)$ to $n^{1-\frac{1}{2k}}$ (we assume here that $n^{1-\frac{1}{2k}}$ is larger than the sum of the sizes of the previous layers. See Remark \ref{remark:TowerVertices} below for an explanation).

For each $i<F-1$, each vertex of the $i$'th layer is connected by an edge to every other vertex of the $i$'th layer, and to every vertex of the $(i+1)$'th layer. The weight of the edges within a layer is $1$, while the weight of the edges from the $i$'th layer to the $(i+1)$'th is $r_{i+1}-r_i$.

\begin{remark} \label{remark:TowerVertices}
Note that by our assumption that $k\leq\frac{\log n}{4\log\log n}$,
\[n^{\frac{1}{2k}}\geq n^{\frac{2\log\log n}{\log n}}=\log^2n>k\log n\geq F\log n~,\]
where the last inequality is due to the fact that $F\leq k$ (because $\forall_l\;\lambda_l\geq1$, so $\sum_{l<k}\lambda_l\geq k$). Equivalently, 
\begin{equation} \label{eq:WellDefined}
    F\log n\cdot n^{1-\frac{1}{k}}<n^{1-\frac{1}{2k}}~.
\end{equation}
Then,
\begin{eqnarray*}
1+\log n\sum_{i=1}^{F-2}n^{\frac{1}{k}\sum_{l<i}\lambda_l}&\leq&1+(F-1)\log n\cdot n^{\frac{1}{k}\sum_{l<F-2}\lambda_l}\\
&\leq&1+(F-1)\log n\cdot n^{\frac{1}{k}(k-1)}\\
&=&1+(F-1)\log n\cdot n^{1-\frac{1}{k}}\\
&\leq&F\log n\cdot n^{1-\frac{1}{k}}\stackrel{(\ref{eq:WellDefined})}{<}n^{1-\frac{1}{2k}}~,
\end{eqnarray*}
and therefore $T(k,f,\alpha,n)$ is well defined. Here, we used the fact that $\sum_{l<F-2}\lambda_l\leq k-1$. This fact could be proved by noticing that the definition of $\{\lambda_j\}$ implies that $\lambda_j$ is an integer greater or equal to $1$ for every $j$. Then, if $\sum_{l<F-2}\lambda_l>k-1$ then $\sum_{l<F-2}\lambda_l\geq k$, and we get
\[\sum_{l<F-1}\lambda_l=\sum_{l<F-2}\lambda_l+\lambda_{F-2}\geq k+1>k~,\]
in contradiction to $F$ being the minimal such that $\sum_{l<F}\lambda_l\geq k$.
\end{remark}

Now, the graph $G(k,f,\alpha,n)$ consists of a path $P$ of length $n^\frac{1}{2k}-1$, where each of the $n^\frac{1}{2k}$ vertices of $P$ is the $0$'th layer of a disjoint copy of the tower $T(k,f,\alpha,n)$. We denote by $L_{c,i}$ the $i$'th layer of the $c$'th copy of $T(k,f,\alpha,n)$.

The total number of vertices in $G(k,f,\alpha,n)$ is
\[n^\frac{1}{2k}\cdot|T(k,f,\alpha,n)|=n^\frac{1}{2k}\cdot n^{1-\frac{1}{2k}}=n~.\]
See figure \ref{fig:Gkfn} for a visualization of the graph $G(k,f,\alpha,n)$.

\begin{center}
\begin{figure}[ht]
    \centering
    \includegraphics[width=12cm, height=5cm]{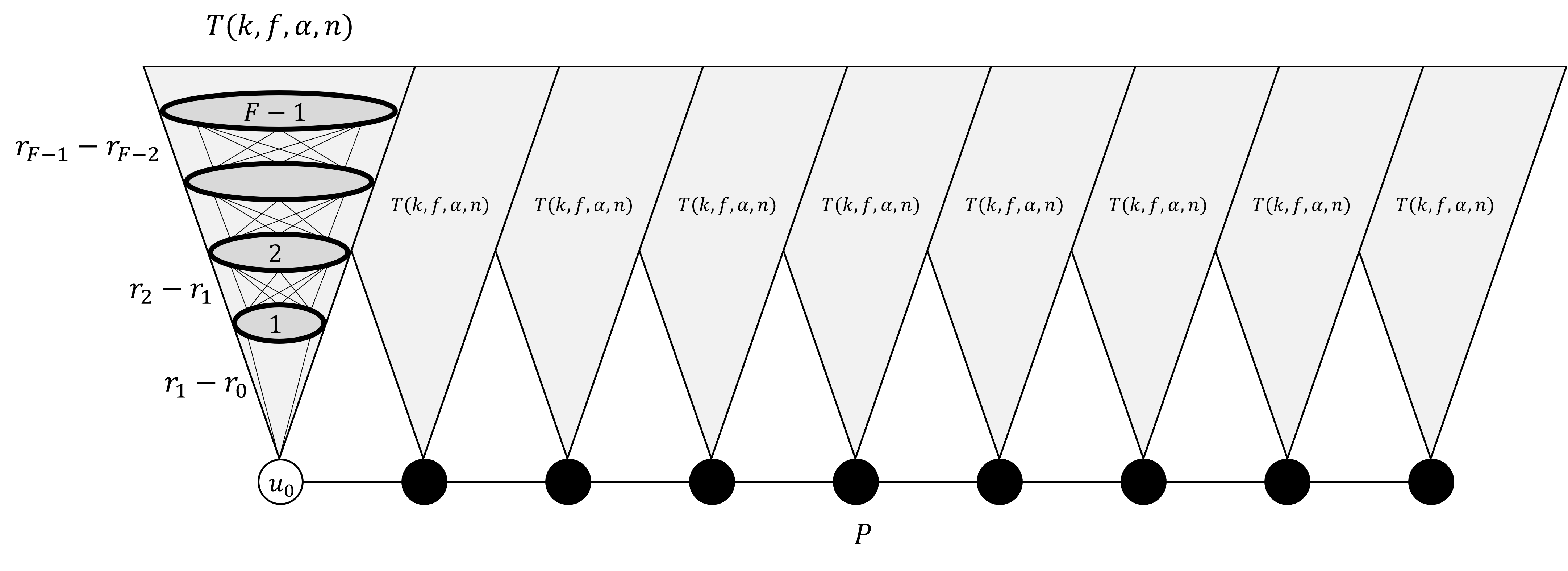}
    \caption{The graph $G(k,f,\alpha,n)$. $P$ contains $n^{\frac{1}{2k}}$ vertices, each of them is the $0$'th layer of a copy of $T(k,f,\alpha,n)$.}
    \label{fig:Gkfn}
\end{figure}
\end{center}

\subsection{Lower Bound Proof}
Suppose we construct the hopset $H(k,f)$ on the graph $G(k,f,\alpha,n)$. We prove the following lemma.

\begin{lemma} \label{lemma:HighProb}
With high probability, for each $c\in[1,n^\frac{1}{2k}]$ and each $j<F-1$, $L_{c,j}$ contains a vertex from $A_j$, but doesn't contain a vertex from $A_{j+1}$.
\end{lemma}

\begin{proof}
For every $c\in[1,n^\frac{1}{2k}]$ and $j\in[0,F-2]$, denote by $E_{c,j}$ the event that $L_{c,j}\cap A_j\neq\emptyset$ and $L_{c,j}\cap A_{j+1}=\emptyset$. Since each vertex of the graph is chosen into the $A_j$'s independently, and since the sets $\{L_{c,j}\}_{c,j}$ are pairwise disjoint, we get that the events $\{E_{c,j}\}_{c,j}$ are pairwise independent. Also, notice that if $L_{c,j}\cap A_j=\emptyset$, then $L_{c,j}\cap A_{j+1}=\emptyset$ as well, since $A_{j+1}\subseteq A_j$. Thus,
\[\Pr[E_{c,j}]=\Pr[L_{c,j}\cap A_{j+1}=\emptyset]-\Pr[L_{c,j}\cap A_j=\emptyset]~.\]

Note that the size of $L_{c,j}$, for every $c$ and $j<F-1$ is at most $\log n\cdot n^{\frac{1}{k}\sum_{l<j}\lambda_l}$. The probability of each of these vertices to be in $A_j$ is exactly $n^{-\frac{1}{k}\sum_{l<j}\lambda_l}$, and therefore,
\[\Pr[L_{1,j}\cap A_j=\emptyset]=(1-n^{-\frac{1}{k}\sum_{l<j}\lambda_l})^{\log n\cdot n^{\frac{1}{k}\sum_{l<j}\lambda_l}}\leq e^{-\log n}<\frac{1}{n}~.\]
On the other hand,
\begin{eqnarray*}
    \Pr[L_{1,j}\cap A_{j+1}=\emptyset]&=&\left(1-n^{-\frac{1}{k}\sum_{l<j+1}\lambda_l}\right)^{\log n\cdot n^{\frac{1}{k}\sum_{l<j}\lambda_l}}\\
    &\geq&4^{-n^{-\frac{1}{k}\sum_{l<j+1}\lambda_l}\cdot\log n\cdot n^{\frac{1}{k}\sum_{l<j}\lambda_l}}\\
    &=&4^{-n^{-\frac{\lambda_{j+1}}{k}}\cdot\log n}
    =n^{-2n^{-\frac{\lambda_{j+1}}{k}}}
    \geq n^{-2n^{-\frac{1}{k}}}~.
\end{eqnarray*}

Here we used the inequality $1-x\geq4^{-x}$, that holds for every $x\in[0,\frac{1}{2}]$. 
For convenience, denote $x=n^{\frac{1}{k}}$. Then, we saw $\Pr[L_{1,j}\cap A_{j+1}=\emptyset]\geq n^{-\frac{2}{x}}$. Thus, the probability of the event $E_{c,j}$ is
\begin{eqnarray*}
    \Pr[E_{c,j}]&=&\Pr[L_{c,j}\cap A_{j+1}=\emptyset]-\Pr[L_{c,j}\cap A_j=\emptyset]\\
    &>&n^{-\frac{2}{x}}-\frac{1}{n}=n^{-\frac{2}{x}}\left(1-n^{-1+\frac{2}{x}}\right)\\
    &\geq&n^{-\frac{2}{x}}\cdot4^{-n^{-1+\frac{2}{x}}}~.
\end{eqnarray*}
We take $\log=\log_2$, and we get $\log\Pr[E_{c,j}]\geq-\frac{2\log n}{x}-2n^{-1+\frac{2}{x}}$. By independence, we observe that
\begin{eqnarray*}
    \log\Pr\left[\bigcap_{c,j}E_{c,j}\right]&=&\log\left(\prod_{c,j}\Pr[E_{c,j}]\right)
    \geq\sum_{c,j}\left(-\frac{2\log n}{x}-2n^{-1+\frac{2}{x}}\right)\\
    &=&(F-1)n^{\frac{1}{2k}}\left(-\frac{2\log n}{x}-2n^{-1+\frac{2}{x}}\right)\\
    &\geq&k\sqrt{x}\left(-\frac{2\log n}{x}-2n^{-1+\frac{2}{x}}\right)\\
    &=&-\frac{2k\log n}{\sqrt{x}}-2k\sqrt{x}n^{-1+\frac{2}{x}}~.
\end{eqnarray*}
We used here the fact that $F\leq k$ (because $\forall_l\;\lambda_l\geq1$, so $\sum_{l<k}\lambda_l\geq k$). Recall that we assumed that $k\leq\frac{\log n}{4\log\log n}$, and thus the absolute value of the first term can be bounded by
\begin{equation} \label{eq:FirstTermBound}
    \frac{2k\log n}{\sqrt{x}}=\frac{2k\log n}{n^{\frac{1}{2k}}}\leq\frac{2k\log n}{n^{\frac{2\log\log n}{\log n}}}=\frac{2k\log n}{\log^2n}=\frac{2k}{\log n}\leq\frac{1}{2\log\log n}\stackrel{n\rightarrow\infty}{\longrightarrow}0~.
\end{equation}

For the second term, we again take $\log=\log_2$, and get
\begin{eqnarray*}
    \log\left(2k\sqrt{x}n^{-1+\frac{2}{x}}\right)&=&1+\log k+\frac{\log x}{2}+\left(-1+\frac{2}{x}\right)\log n\\
    &\leq&1+\log\log n+\frac{\log n}{2k}-\log n+\frac{2\log n}{x}\\
    &=&\left(o(1)+\frac{1}{2k}-1\right)\log n+\frac{2\log n}{x}\\
    &\leq&\left(o(1)-\frac{1}{2}\right)\log n+\frac{2k\log n}{\sqrt{x}}~.
\end{eqnarray*}
Here, by Inequality (\ref{eq:FirstTermBound}), the second term converges to $0$, and therefore the whole expression converges to $-\infty$. Hence, $2k\sqrt{x}n^{-1+\frac{2}{x}}\stackrel{n\rightarrow\infty}{\longrightarrow}0$. Back to $\log\Pr\left[\bigcap_{c,j}E_{c,j}\right]=-\frac{2k\log n}{\sqrt{x}}-2k\sqrt{x}n^{-1+\frac{2}{x}}$, we saw that both of the terms converge to $0$, and thus $\Pr\left[\bigcap_{c,j}E_{c,j}\right]\stackrel{n\rightarrow\infty}{\longrightarrow}1$, as desired.

\end{proof}

Lemma \ref{lemma:HighProb} lets us understand what hop-edges $H(k,f)$ adds to $G(k,f,\alpha,n)$. We prove that each of these edges is either within the same copy of a tower, or has a large weight with respect to the number of towers it skips.

\begin{lemma} \label{lemma:2TCon}
Let $k,f,\alpha,n$ and assume that Lemma \ref{lemma:HighProb} is satisfied\footnote{That is, for every $c\in[1,n^{\frac{1}{2k}}]$ and $j<F-1$, assume we have $L_{c,j}\cap A_j\neq\emptyset$ and $L_{c,j}\cap A_{j+1}=\emptyset$.}. Suppose that $H(k,f)$ contains an edge $(x,y)$, where $x\in L_{c,i}$, $y\in L_{d,j}$, $c\neq d$ and $i,j<F-2$. Then:
\[w(x,y)>(\alpha+1)(|d-c|-2)~.\]
\end{lemma}

\begin{proof}

We first claim that for every vertex $z$ in a tower $T$ in the graph $G=G(k,f,\alpha,n)$, all of its pivots $\{p_i(z)\}_{i=0}^{F-2}$ are in the same tower $T$. This claim is trivial for $i\leq i(z)$ (see Section \ref{Generalized TZ Construction} for the definitions of $p_i(z),i(z)$ and other notations), as for such $i$'s, $p_i(z)=z$. For $i(z)<i\leq F-2$, by Lemma \ref{lemma:HighProb}, we may assume that the $i$'th layer of $T$ contains a vertex in $A_i$. This layer is clearly closer to $z$ than all the $i$'th layers in towers $T'\neq T$. Hence, by the definition of $p_i(z)$, we conclude that $p_i(z)$ is also in $T$.

Recall that $H(k,f)$ contains only two types of edges - edges of the form $(z,p_i(z))$ for some vertex $z$ and $i\in[0,F-1]$, and edges of the form $(z,z')$, where $z'\in B_{j'}(z)$ and $j'\in\left[i(z),f\left(i(z)\right)\right]$. Therefore, by our argument above, it cannot be that $(x,y)$ is of the first type, since $x,y$ are in different towers, and $x,y\notin A_{F-1}$ (because $x\in L_{c,i},\;y\in L_{d,j}$, and $i,j<F-2$, and we assumed that these layers do not contain any vertex from $A_{F-1}$). 

Hence, without loss of generality, it must be that $y\in B_{j'}(x)$, for some $j'\in[i(x),f(i(x))]$. By the definition of $B_{j'}(x)$, we conclude that (1) $y\in A_{j'}$, and that (2) $d(x,y)<d(x,p_{j'+1}(x))$. Since we assumed that there are no vertices of $A_{j+1}$ in $L_{d,j}$, (1) implies that $j'\leq j$. On the other hand, if $j'<j$, then $L_{d,j}$ contains a vertex $z\in A_j\subseteq A_{j'+1}$, and therefore $d(x,y)=d(x,z)\geq d(x,p_{j'+1}(x))$, in contradiction to (2). We therefore conclude that $j'=j$.

Observe that $j=j'\in[i(x),f(i(x))]$, and thus $f(i(x))\geq j$. By the definition of $f^{-1}$, this implies that $f^{-1}(j)\leq i(x)$. We also claim that $i(x)\leq i$, since we assumed that there are no vertices of $A_{i+1}$ in $L_{c,i}$. We finally get
\begin{equation} \label{eq:IndicesInequality}
    f^{-1}(j)\leq i(x)\leq i~.
\end{equation}

We now compute that distances on both sides of the inequality $d(x,y)<d(x,p_{j+1}(x))$. Since $p_{j+1}(x)$ is in the same tower as $x$, $d(x,p_{j+1}(x))$ is simply the distance between the layers $L_{c,i}$ and $L_{c,j+1}$, which is 
\[d(x,p_{j+1}(x))=\sum_{l=i}^{j}(r_{l+1}-r_l)=r_{j+1}-r_i~.\]
For the distance $d(x,y)$, observe that the shortest path from $x\in L_{c,i}$ to $y\in L_{d,j}$ starts by getting from $x$ to the $0$'th layer of its tower, then through the path $P$, reaching the $0$'th layer of $y$'s tower, and finally getting up through the layers to $y$. The weight of this path, which is the weight of the edge $(x,y)$ is
\[d(x,y)=w(x,y)=\sum_{l<i}(r_{l+1}-r_l)+|d-c|+\sum_{l<j}(r_{l+1}-r_l)=r_i+|d-c|+r_j-2~.\]

Substituting in the inequality $d(x,y)<d(x,p_{j+1}(x))$, we get
\begin{eqnarray*}
    |d-c|&<&r_{j+1}-r_i-r_i-r_j+2\stackrel{(\ref{eq:IndicesInequality})}{\leq} r_{j+1}-r_j-2r_{f^{-1}(j)}+2=\frac{1}{\alpha}(r_j+r_{f^{-1}(j)})+2\\
    &\Rightarrow&r_j+r_{f^{-1}(j)}>\alpha|d-c|-2\alpha~,
\end{eqnarray*}
where we used $\{r_j\}$'s definition and the fact that it is a non-decreasing sequence. By these same reasons we finally get
\begin{eqnarray*}
  w(x,y)&=&r_i+|d-c|+r_j-2\\
  &\geq&r_{f^{-1}(j)}+|d-c|+r_j-2\\
  &>&\alpha|d-c|-2\alpha+|d-c|-2\\
  &=&(\alpha+1)(|d-c|-2)~.
\end{eqnarray*}

\end{proof}

We now show that the graph $G(k,f,\alpha,n)$ demonstrates a lower bound for the hopset $H(k,f)$. The following lemma shows that the upper bound for the hopbound of $H(k,f)$ from Theorem \ref{thm:Hkf} is actually tight, up to a $\Theta(\alpha^2)$-factor.

\begin{lemma} \label{lemma:HkfLowerB}
Suppose that $H=H(k,f)$ is an $(\alpha,\beta)$-hopset for $G=G(k,f,\alpha,n)$. Then, if $n$ is large enough, with high probability:
\[\beta\geq\frac{r_{F-2}-1}{5\alpha^2}~.\]
\end{lemma}

\begin{proof}
Denote by $u_0$ the first vertex of the path $P$ and let $v$ be some other vertex in this path. Let $P_{u_0,v}$ be the shortest path between $u_0,v$ in the graph $G\cup H$, that has at most $\beta$ edges. By our assumption that $H$ is an $(\alpha,\beta)$-hopset for $G$, we know that
\[w(P_{u_0,v})\leq\alpha d(u_0,v)~.\]

Denote by $(x_1,y_1),(x_2,y_2),...,(x_b,y_b)$ the edges of $P_{u_0,v}$ that connect two different towers in $G$. Note that $b\leq\beta$, since $P_{u_0,v}$ has at most $\beta$ edges. Let $c_i,d_i$ be the indices of the towers for which $x_i,y_i$ belong respectively. We first assume that all the vertices $\{x_i\},\{y_i\}$ are not in the last two layers of their respective towers. That is, for every $i$, $x_i\notin L_{c_i,F-2}\cup L_{c_i,F-1}$ and $y_i\notin L_{d_i,F-2}\cup L_{d_i,F-1}$. Then, by Lemma \ref{lemma:2TCon}, for every $i$ such that $(x_i,y_i)\in H$,
\[w(x_i,y_i)>(\alpha+1)(|d_i-c_i|-2)~.\]
For indices $i$ such that $(x_i,y_i)\notin H$, it must be that $(x_i,y_i)\in G$, and since the only edges of $G$ that connect two different towers are between adjacent towers, we have $|d_i-c_i|=1$. Therefore, in that case we have $w(x_i,y_i)=1>-\alpha-1=(\alpha+1)(|d_i-c_i|-2)$ as well.

Thus, we conclude that
\begin{eqnarray*}
  \alpha d(u_0,v)&\geq&w(P_{u_0,v})
  \geq\sum_{i=1}^b w(x_i,y_i)
  >(\alpha+1)\sum_{i=1}^b(|d_i-c_i|-2)\\
  &\geq&(\alpha+1)(d(u_0,v)-2b)
  \geq(\alpha+1)(d(u_0,v)-2\beta)~,
\end{eqnarray*}
where we used the fact that $\sum_{i=1}^b|d_i-c_i|\geq d(u_0,v)$, which is true since $\{(x_i,y_i)\}$ are all of the edges of $P_{u_0,v}$ that connect two different towers (other edges does not make any ``progress" towards $v$).

Therefore, we proved that for every vertex $v$ on the path $P$ we have $\alpha d(u_0,v)>(\alpha+1)(d(u_0,v)-2\beta)$ and thus $d(u_0,v)<2\beta(\alpha+1)\leq4\alpha\beta$. However, for large enough values of $n$, that satisfy $|P|=n^\frac{1}{2k}>\lceil4\alpha\beta\rceil$, there must be a vertex $v$ that does not satisfy this inequality. Hence, for such $n$ and $v$, our assumption that the vertices $\{x_i\},\{y_i\}$ are not contained in the last two layers of their respective towers, cannot hold.

In particular, let $v\in P$ be a vertex such that $d(u_0,v)=\lceil4\alpha\beta\rceil$. There must be some $i$ such that either $x_i\in L_{c_i,F-2}\cup L_{c_i,F-1}$ or $y_i\in L_{d_i,F-2}\cup L_{d_i,F-1}$. In both cases, since $c_i\neq d_i$, we have
\[w(x_i,y_i)\geq\sum_{l<F-2}(r_{l+1}-r_l)=r_{F-2}-1~.\]
The weight of this edge is a trivial lower bound for the weight of the path $P_{u_0,v}$, which also satisfies
\[w(P_{u_0,v})\leq\alpha\cdot d(u_0,v)=\alpha\lceil4\alpha\beta\rceil\leq5\alpha^2\beta~,\]
hence we finally get
\[\beta\geq\frac{r_{F-2}-1}{5\alpha^2}~.\]

\end{proof}

The following lemma lets us better understand the lower bound that was found in the previous lemma.

\begin{lemma} \label{lemma:rBound}
For every $k,f$ and $\alpha\geq2$:
\[r_{F-2}\geq\frac{1}{8}k^{1+\frac{1}{2\log\alpha}}~.\]
\end{lemma}

\begin{proof}
We write again the definitions of $\{\lambda_j\},F,\{r_j\}$:
\begin{enumerate}
    \item $\lambda_j=1+\sum_{l<f^{-1}(j)}\lambda_l$,
    \item $F=\min\{F'\;|\;\sum_{l<F'}\lambda_l\geq k\}$,
    \item $r_0=1$ and $r_{j+1}=(2+\frac{1}{\alpha})r_{j}+(2+\frac{1}{\alpha})r_{f^{-1}(j)}$.
\end{enumerate}

Define a new sequence $\{\Lambda_j\}$ as follows.
\[\Lambda_j=1+\sum_{l<j}\lambda_l~.\]

We can immediately observe by $\{\lambda_j\}$'s definition that $\Lambda_{f^{-1}(j)}=\lambda_j$, and therefore,
\[\Lambda_{j+1}=1+\sum_{l<j+1}\lambda_l=1+\sum_{l<j}\lambda_l+\lambda_j=\Lambda_j+\Lambda_{f^{-1}(j)}~.\]

By the definition of $F$, we get $\Lambda_{F}=1+\sum_{l<F}\lambda_l\geq k+1$. It will later be useful to notice that
\[\Lambda_{j+1}=\Lambda_j+\Lambda_{f^{-1}(j)}\leq2\Lambda_j\;\;\Rightarrow\;\;\Lambda_{F-2}\geq\frac{1}{2}\Lambda_{F-1}\geq\frac{1}{4}\Lambda_F\geq\frac{k+1}{4}~.\]

We now try to find some $d>0$ such that we can prove the inequality $r_j\geq\Lambda_j^{1+\frac{1}{d}}$ for every $j\geq0$. If we find such $d$, we will get:
\[r_{F-2}\geq\Lambda_{F-2}^{1+\frac{1}{d}}\geq\left(\frac{k+1}{4}\right)^{1+\frac{1}{d}}~,\]
which is similar to what we want to prove.

Suppose we prove this inequality by induction over $j$. For $j=0$, we have $r_0=1=1^{1+\frac{1}{d}}=\Lambda_0^{1+\frac{1}{d}}$, and for $j>0$:
\begin{eqnarray*}
  r_j&=&(1+\frac{1}{\alpha})r_{j-1}+(2+\frac{1}{\alpha})r_{f^{-1}(j-1)}\\
  &\geq&(1+\frac{1}{\alpha})\Lambda_{j-1}^{1+\frac{1}{d}}+(2+\frac{1}{\alpha})\Lambda_{f^{-1}(j-1)}^{1+\frac{1}{d}}\\
  &\stackrel{?}{\geq}&(\Lambda_{j-1}+\Lambda_{f^{-1}(j-1)})^{1+\frac{1}{d}}\\
  &=&\Lambda_j^{1+\frac{1}{d}}~,
\end{eqnarray*}

So we are only left to prove the inequality marked by ``?". In appendix \ref{rBoundAppendix} we prove that given $d>0$ such that $(1+\frac{1}{\alpha})^{-d}+(2+\frac{1}{\alpha})^{-d}\leq1$,
\[\forall_{x,y\geq0}\;(1+\frac{1}{\alpha})x^{1+\frac{1}{d}}+(2+\frac{1}{\alpha})y^{1+\frac{1}{d}}\geq(x+y)^{1+\frac{1}{d}}~.\]

Therefore, if we show that we can choose $d=2\log\alpha$, i.e. that $(1+\frac{1}{\alpha})^{-2\log\alpha}+(2+\frac{1}{\alpha})^{-2\log\alpha}\leq1$, then our proof by induction holds, and we have 
\[r_{F-2}\geq\Lambda_{F-2}^{1+\frac{1}{2\log\alpha}}\geq\left(\frac{k+1}{4}\right)^{1+\frac{1}{2\log\alpha}}\geq\frac{1}{8}k^{1+\frac{1}{2\log\alpha}}~,\]
where the last inequality holds for every $\alpha\geq2$.

For $\alpha\geq2$, we have $\alpha^2\geq2\alpha$, and then $2\alpha+1\leq\alpha^2+1<2\alpha^2<2\alpha^3$. Now, for $\alpha\geq2$,
\begin{eqnarray*}
  (1+\frac{1}{\alpha})^{-2\log\alpha}+(2+\frac{1}{\alpha})^{-2\log\alpha}&\leq&(1+\frac{1}{\alpha})^{-2}+2^{-2\log\alpha}\\
  &=&\frac{\alpha^2}{(\alpha+1)^2}+\frac{1}{\alpha^2}\\
  &=&\frac{\alpha^4+(\alpha+1)^2}{\alpha^2(\alpha+1)^2}\\
  &=&\frac{\alpha^4+\alpha^2+2\alpha+1}{\alpha^4+2\alpha^3+\alpha^2}\\
  &<&\frac{\alpha^4+\alpha^2+2\alpha^3}{\alpha^4+2\alpha^3+\alpha^2}\\
  &=&1
\end{eqnarray*}

\end{proof}

The following theorem concludes the previous two lemmas:

\begin{theorem} \label{thm:HkfLB}
For every choice of $k,f$ and $\alpha\geq2$, there is a graph $G$ such that if the hopset $H(k,f)$ is an $(\alpha,\beta)$-hopset for $G$, then with high probability,
\[\beta\geq\frac{1}{40\alpha^2}k^{1+\frac{1}{2\log\alpha}}-1~.\]
\end{theorem}

\begin{proof}
Fix some $k,f$, and let $n\in\{2^{2k\cdot a}\;|\;a\in\mathbb{N}\}$ be large enough such that the previous lemmas are true (particularly, choose $n$ such that Lemma \ref{lemma:HighProb} and Lemma \ref{lemma:HkfLowerB} are satisfied). When choosing $G=G(k,f,\alpha,n)$, we know by Lemma \ref{lemma:HkfLowerB} that $\beta\geq\frac{r_{F-2}-1}{5\alpha^2}$, and by Lemma \ref{lemma:rBound} that $r_{F-2}\geq\frac{1}{8}k^{1+\frac{1}{2\log\alpha}}$. Combining these two results, we get
\[\beta\geq\frac{1}{5\alpha^2}(r_{F-2}-1)\geq\frac{1}{5\alpha^2}(\frac{1}{8}k^{1+\frac{1}{2\log\alpha}}-1)\geq\frac{1}{40\alpha^2}k^{1+\frac{1}{2\log\alpha}}-1~.\]

\end{proof}

\begin{remark}
While we achieved a lower bound of $\beta=\Omega(\frac{1}{\alpha^2}k^{1+1/2\log\alpha})$ for the hopbound of $H(k,f)$, we believe that a lower bound of $\beta=\Omega(\frac{1}{\alpha}k^{1+1/2\log\alpha})$ may also be shown. If such lower bound is achieved, notice that it would be tight compared to the known upper bounds for various values of $\alpha$:
\begin{enumerate}
    \item For $\alpha=\Theta(k^\epsilon)$, for a constant $0<\epsilon\leq1$, and with a suitable choice of $f$, the hopset $H(k,f)$ has a hopbound of $O(k^{1-\epsilon})$ (see Section \ref{sec:HopsetSpecialCases}). The lower bound in this case would be
    \[\beta=\Omega(\frac{1}{\alpha}k^{1+1/2\log\alpha})=\Omega_\epsilon(k^{1-\epsilon})~.\]
    \item For $\alpha=O(1)$, and with a suitable choice of $f$, the hopset $H(k,f)$ has a hopbound of $O(k^{1+\frac{2}{\ln\alpha}})$ (see Section \ref{sec:HopsetSpecialCases}). The lower bound in this case would be
    \[\beta=\Omega(k^{1+1/2\log\alpha})~.\]
\end{enumerate}

For $\alpha=O(1)$, notice that our lower bound from Theorem \ref{thm:HkfLB} is also tight, since $O(\frac{1}{\alpha^2})=O(\frac{1}{\alpha})=O(1)$.
\end{remark}

\bibliographystyle{alpha}
\bibliography{hopset}

\newpage
\appendix

\section{Calculating the Parameters for Rounding Functions} \label{appendix:FProperties}

Given the function $f(i)=\lfloor\frac{i}{c}\rfloor\cdot c+c-1$, and a real parameter $t>0$, recall the definitions of the variables $\{\lambda_j\},F,\{r_i\}$:
\begin{enumerate}
    \item $\lambda_j=1+\sum_{l<f^{-1}(j)}\lambda_l$.
    \item $F=\min\{F'\;|\;\sum_{j<F'}\lambda_j\geq k\}$.
    \item $\forall_{i>0}\; r_i=(1+\frac{4}{t})r_{i-1}+(2+\frac{4}{t})r_{f^{-1}(i-1)}$ ($r_0$ is left as a parameter).
\end{enumerate}
Here, $f^{-1}(j)=\min\{i\;|\;f(i)\geq j\}$. In this appendix we prove that
\begin{enumerate}
    \item $f^{-1}(j)=\lfloor\frac{j}{c}\rfloor\cdot c$.
    \item $\lambda_{ac+b}=(c+1)^a$ for every integers $a\geq0$ and $b\in[0,c-1]$.
    \item $F\leq\lceil\log_{c+1}(k+1)\rceil\cdot c$.
    \item $r_{ac+b}=\left[\left(1+\frac{4}{t}\right)^b\left(\frac{t}{2}+2\right)-\left(\frac{t}{2}+1\right)\right]\left[\left(1+\frac{4}{t}\right)^c\left(\frac{t}{2}+2\right)-\left(\frac{t}{2}+1\right)\right]^ar_0$ for every integers $a\geq0$ and $b\in[0,c-1]$.
\end{enumerate}

Given some $j\geq0$, note that
\[f\left(\left\lfloor\frac{j}{c}\right\rfloor\cdot c\right)=\left\lfloor\frac{\left\lfloor\frac{j}{c}\right\rfloor\cdot c}{c}\right\rfloor\cdot c+c-1=\left\lfloor\frac{j}{c}\right\rfloor\cdot c+c-1>\left(\frac{j}{c}-1\right)\cdot c+c-1=j-1~,\]
and since $f(i)$ is an integer for every $i$, we deduce $f\left(\lfloor\frac{j}{c}\rfloor\cdot c\right)\geq j$. On the other hand,
\[f\left(\left\lfloor\frac{j}{c}\right\rfloor\cdot c-1\right)=\left\lfloor\frac{\left\lfloor\frac{j}{c}\right\rfloor\cdot c-1}{c}\right\rfloor\cdot c+c-1=\left(\frac{j}{c}-1\right)\cdot c+c-1=j-1~.\]
By the definition of $f^{-1}$, we conclude that $f^{-1}(j)=\left\lfloor\frac{j}{c}\right\rfloor\cdot c$.

Hence, in the definition of $\{\lambda_j\}$, we get $\lambda_j=1+\sum_{l<\lfloor\frac{j}{c}\rfloor\cdot c}\lambda_l$. From this equation it is easy to see that $\lambda_{ac+b}=\lambda_{ac}$ for every integer $a\geq0$ and $b\in[0,c-1]$. Therefore,
\[\lambda_{ac}=1+\sum_{l<\lfloor\frac{ac}{c}\rfloor\cdot c}\lambda_l=1+\sum_{l<ac}\lambda_l=1+\sum_{a'<a}\sum_{b'<c}\lambda_{a'c+b'}=1+\sum_{a'<a}c\lambda_{a'c}~.\]
As a result,
\[\lambda_{ac}=1+\sum_{a'<a}c\lambda_{a'c}=1+\sum_{a'<a-1}c\lambda_{a'c}+c\lambda_{(a-1)c}=\lambda_{(a-1)c}+c\lambda_{(a-1)c}=(c+1)\lambda_{(a-1)c}~.\]
We conclude that $\lambda_{ac+b}=\lambda_{ac}=(c+1)^a$ for all $a\geq0$ and $b\in[0,c-1]$.

Next, to bound $F$, notice that for every $a\geq0$ and $b\in[0,c-1]$,
\begin{eqnarray*}
\sum_{j<ac+b}\lambda_j&=&c\sum_{a'<a}\lambda_{a'c}+b\lambda_{ac}=c\sum_{a'<a}(c+1)^{a'}+b(c+1)^a=(c+1)^a-1+b(c+1)^a\\
&=&(b+1)(c+1)^a-1~.
\end{eqnarray*}
For $a=\lceil\log_{c+1}(k+1)\rceil$ and $b=0$, this expression is at least $k$, thus $F\leq\lceil\log_{c+1}(k+1)\rceil\cdot c$. 

Finally, we compute the sequence $\{r_i\}$. For every $a\geq0$ and $b\in[0,c-1]$, we have
\[r_{ac+b+1}=\left(1+\frac{4}{t}\right)r_{ac+b}+\left(2+\frac{4}{t}\right)r_{f^{-1}(ac+b)}=\left(1+\frac{4}{t}\right)r_{ac+b}+\left(2+\frac{4}{t}\right)r_{ac}~.\]
Fixing $a$, we prove by induction over $b\in[0,c-1]$ that $r_{ac+b}=\left[\left(1+\frac{4}{t}\right)^b\left(\frac{t}{2}+2\right)-\left(\frac{t}{2}+1\right)\right]r_{ac}$. For $b=0$, this is trivial. Assuming this holds for $b$, we have
\begin{eqnarray*}
r_{ac+b+1}&=&\left(1+\frac{4}{t}\right)r_{ac+b}+\left(2+\frac{4}{t}\right)r_{ac}\\
&=&\left(1+\frac{4}{t}\right)\left[\left(1+\frac{4}{t}\right)^b\left(\frac{t}{2}+2\right)-\left(\frac{t}{2}+1\right)\right]r_{ac}+\left(2+\frac{4}{t}\right)r_{ac}\\
&=&\left[\left(1+\frac{4}{t}\right)^{b+1}\left(\frac{t}{2}+2\right)-\left(\frac{t}{2}+1\right)-\frac{4}{t}\left(\frac{t}{2}+1\right)\right]r_{ac}+\left(2+\frac{4}{t}\right)r_{ac}\\
&=&\left[\left(1+\frac{4}{t}\right)^{b+1}\left(\frac{t}{2}+2\right)-\left(\frac{t}{2}+1\right)\right]r_{ac}-\left(2+\frac{4}{t}\right)+\left(2+\frac{4}{t}\right)r_{ac}\\
&=&\left[\left(1+\frac{4}{t}\right)^{b+1}\left(\frac{t}{2}+2\right)-\left(\frac{t}{2}+1\right)\right]r_{ac}~,
\end{eqnarray*}
which completes the inductive proof.

Notice that the the resulting equation $r_{ac+b+1}=\left[\left(1+\frac{4}{t}\right)^{b+1}\left(\frac{t}{2}+2\right)-\left(\frac{t}{2}+1\right)\right]r_{ac}$ is true for every $b\in[0,c-1]$. In particular, for $b=c-1$, we get
\[r_{(a+1)c}=r_{ac+c-1+1}=\left[\left(1+\frac{4}{t}\right)^c\left(\frac{t}{2}+2\right)-\left(\frac{t}{2}+1\right)\right]r_{ac}~.\]
That is, $r_{ac}=\left[\left(1+\frac{4}{t}\right)^c\left(\frac{t}{2}+2\right)-\left(\frac{t}{2}+1\right)\right]^ar_0$, for every $a\geq0$. We finally conclude that for every integers $a\geq0$ and $b\in[0,c-1]$
\begin{eqnarray*}
r_{ac+b}&=&\left[\left(1+\frac{4}{t}\right)^b\left(\frac{t}{2}+2\right)-\left(\frac{t}{2}+1\right)\right]r_{ac}\\
&=&\left[\left(1+\frac{4}{t}\right)^b\left(\frac{t}{2}+2\right)-\left(\frac{t}{2}+1\right)\right]\left[\left(1+\frac{4}{t}\right)^c\left(\frac{t}{2}+2\right)-\left(\frac{t}{2}+1\right)\right]^ar_0~.
\end{eqnarray*}

\newpage
\section{Completing the Proof of Lemma \ref{lemma:rBound}} \label{rBoundAppendix}
For completing the proof of the Lemma \ref{lemma:rBound}, we prove the following lemma:
\begin{lemma}
For any $a,b>1$, $d>0$, such that $a^{-d}+b^{-d}\leq1$:
\[\forall_{x,y\geq0}\;\;ax^{1+\frac{1}{d}}+by^{1+\frac{1}{d}}\geq(x+y)^{1+\frac{1}{d}}~.\]
\end{lemma}

\begin{proof}
For a fixed $y\geq0$, define a function:
\[h(x)=ax^{1+\frac{1}{d}}+by^{1+\frac{1}{d}}-(x+y)^{1+\frac{1}{d}}~.\]
We want to show that for every $x\geq0$, $h(x)\geq0$.

Derive $h$:
\[h'(x)=a(1+\frac{1}{d})x^{\frac{1}{d}}-(1+\frac{1}{d})(x+y)^{\frac{1}{d}}=(1+\frac{1}{d})x^{\frac{1}{d}}(a-(1+\frac{y}{x})^{\frac{1}{d}})~.\]
In $[0,\infty)$, the derivative is positive when $a>(1+\frac{y}{x})^{\frac{1}{d}}$, i.e. $x>\frac{y}{a^d-1}$, negative when $0<x<\frac{y}{a^d-1}$, and equals zero when $x=\frac{y}{a^d-1}$. That means that $x=\frac{y}{a^d-1}$ is the global minimum of $h$ in $[0,\infty)$, i.e. $\forall_{x\geq0}\;h(x)\geq h(\frac{y}{a^d-1})$.

Compute $h(\frac{y}{a^d-1})$:
\begin{eqnarray*}
  h(\frac{y}{a^d-1})&=&\frac{a}{(a^d-1)^{1+\frac{1}{d}}}y^{1+\frac{1}{d}}+by^{1+\frac{1}{d}}-(\frac{1}{a^d-1}+1)^{1+\frac{1}{d}}y^{1+\frac{1}{d}}\\
  &=&\left[\frac{a}{(a^d-1)^{1+\frac{1}{d}}}+b-(\frac{a^d}{a^d-1})^{1+\frac{1}{d}}\right]y^{1+\frac{1}{d}}\\
  &=&\left[b+\frac{a-a^d\cdot a}{(a^d-1)^{1+\frac{1}{d}}}\right]y^{1+\frac{1}{d}}\\
  &=&\left[b-a(a^d-1)^{-\frac{1}{d}}\right]y^{1+\frac{1}{d}}
\end{eqnarray*}

So $h(\frac{y}{a^d-1})\geq0$ iff
\[a(a^d-1)^{-\frac{1}{d}}\leq b\;\;\iff\;\;\frac{a^d}{a^d-1}\leq b^d\;\;\iff\;\;a^d+b^d\leq a^db^d\;\;\iff\;\;a^{-d}+b^{-d}\leq1~.\]
\end{proof}

\newpage
\section{A Unified Construction of Emulators} \label{EmulatorsAppendix}

In this appendix we repeat the proof of Theorem \ref{thm:Hkf}, with the necessary modifications such that we can obtain Theorem \ref{thm:Ekf}.


\begin{proof}[proof of Theorem \ref{thm:Ekf}]

Given $k$ and the function $f$, consider the hopset $H=H(k,f)$ as an emulator. By Lemma \ref{lemma:HkfESize}, the expected size of $H(k,f)$ is $O(Fn+\max_i(f(i)-i+1)\cdot n^{1+1/k})$, as desired. By Lemma \ref{lemma:HJump}, for every $u,u'\in V$ such that $score(u)=i$ and $d_G(u,u')\leq\frac{r_i-r_{i-1}}{2}-r_{f^{-1}(i-1)}$, we have
\[d_H^{(3)}(u,u')\leq3d_G(u,u')+2(r_{i-1}+r_{f^{-1}(i-1)})~.\]
If also $\frac{2}{t}(r_{i-1}+r_{f^{-1}(i-1)})\leq d_G(u,u')$, then $d_H^{(3)}(u,u')\leq(t+3)d_G(u,u')$.

Now, let $u,v\in V$ be a pair of vertices, and let $u=u_0,u_1,u_2,...,u_d=v$ be the shortest path between them. Starting with $h=0$, find $l>h$ such that $\frac{2}{t}(r_{i-1}+r_{f^{-1}(i-1)})<d_G(u_h,u_l)\leq\frac{r_i-r_{i-1}}{2}-r_{f^{-1}(i-1)}$, where $i=score(u_h)$. If there is no such $l$, stop the process and denote $v'=u_h$. Otherwise, set $h\leftarrow l$, and continue in the same way.

This process creates a subsequence of $u_0,...,u_d$: $u=v_0,v_1,v_2,...,v_b=v'$, such that for every $q<b$ we have $d^{(3)}_H(v_q,v_{q+1})\leq(t+3)d(v_q,v_{q+1})$ (by Lemma \ref{lemma:HJump}). For $v'=v_b$ we claim that $d_G(v',v)\leq\frac{2}{t}(r_{i-1}+r_{f^{-1}(i-1)})$, where $i=score(v')$. Seeking contradiction, assume that $d_G(v',v)>\frac{2}{t}(r_{i-1}+r_{f^{-1}(i-1)})$, and let $u_s$ be the first vertex on the sub-path between $v'$ and $v$ such that $d_G(v',u_s)>\frac{2}{t}(r_{i-1}+r_{f^{-1}(i-1)})$. By the definition of $u_s$, we get 
\[d_G(v',u_{s-1})\leq\frac{2}{t}(r_{i-1}+r_{f^{-1}(i-1)})=\frac{r_i-r_{i-1}}{2}-r_{f^{-1}(i-1)}-1~,\] 
where the last equality is due to the definition of the sequence $\{r_i\}$. But recall that the graph $G$ is unweighted, and thus
$d_G(v',u_s)\leq d_G(v',u_{s-1})+1\leq\frac{r_i-r_{i-1}}{2}-r_{f^{-1}(i-1)}$. We conclude that
\[\frac{2}{t}(r_{i-1}+r_{f^{-1}(i-1)})<d_G(u_h,u_s)\leq\frac{r_i-r_{i-1}}{2}-r_{f^{-1}(i-1)}~,\]
in contradiction to the assumption that there is no such vertex $u_s$ that satisfies this condition.

Hence, we conclude that indeed $d_G(v',v)\leq\frac{2}{t}(r_{i-1}+r_{f^{-1}(i-1)})$, and by Lemma \ref{lemma:HJump}, $d^{(3)}_H(v',v)\leq3d_G(v',v)+4r_F$.

When summing over the entire path, we get
\begin{eqnarray*}
d_H(u,v)&\leq&\sum_{q=0}^{b-1}(t+3)d_G(v_q,v_{q+1})+3d_G(v',v)+4r_F\\
&=&(t+3)d_G(u,v')+3d_G(v',v)+4r_F\\
&\leq&(t+3)d_G(u,v)+4r_F~.
\end{eqnarray*}
This completes the proof that $H=H(k,f)$ is a $(t+3,4r_F)$-emulator.

\end{proof}

\end{document}